%% file: 0_paper.tex
\documentclass{vldb}

\input{macros}

\begin{document}
    \makeatletter
    \def\@copyrightspace{\relax}
    \makeatother

\title{Scalable Package Queries in Relational Database Systems}

\numberofauthors{1}
\author{
\alignauthor
\begin{tabular}{cccc}
Matteo Brucato\mass           & 
Juan Felipe Beltran\nyu & 
Azza Abouzied\nyu      & 
Alexandra Meliou\mass
\end{tabular}
\and  
\alignauthor 
\affaddr{
\begin{tabular}{cc}
\mass College of Information and Computer Sciences & 
\nyu Computer Science             \\
University of Massachusetts   & New York University \\
Amherst, MA, USA              & Abu Dhabi, UAE      \\
\{matteo,ameli\}@cs.umass.edu &
\{juanfelipe,azza\}@nyu.edu
\end{tabular}}
}

\date{}
\maketitle

\begin{abstract}

\looseness -1
\removed{Many modern applications require complex database querying, with
conditions pertaining to the answer set as a whole, in addition to individual answer tuples.}
\added{
Traditional database queries follow a simple model: they define
constraints that each tuple in
the result must satisfy. This model is computationally efficient, as
the database system can evaluate the query conditions on each tuple individually. However, many practical,
real-world problems require a collection of result tuples to satisfy
constraints collectively, rather than individually.
}
\removed{Existing work 
has focused on application- and domain-specific approaches. Our goal
is to create a universal framework 
to support these applications.}
In
this paper, we present \emph{package queries}, a new query model that
extends traditional database queries to handle complex constraints and
preferences over answer sets. We develop a full-fledged package query
system, implemented on top of a traditional database engine. Our work
makes several contributions. First, we design \paql, a \sql-based
query language that supports the declarative specification of package
queries. We prove that \paql is as least as expressive as
integer linear programming, and therefore, evaluation of package queries
is in general NP-hard. Second, we present a fundamental evaluation
strategy that combines the capabilities of databases and constraint
optimization solvers to derive solutions to package queries. The core
of our approach is a set of translation rules that transform a package
query to an integer linear program. Third, we introduce an offline
data partitioning strategy 
allowing query evaluation to scale to large data sizes.
Fourth, we introduce \bt, a scalable algorithm for package evaluation, with strong approximation guarantees ($(1 \pm \epsilon)^6$-factor approximation).
Finally, we present extensive experiments over real-world and benchmark data. The results demonstrate that \bt is effective at deriving high-quality
package results, and achieves runtime performance that is an order of
magnitude faster than directly using \ilp solvers over large datasets.

\end{abstract}

\input{1_intro.tex}

\input{2_package-queries.tex}

\input{4_ilp_evaluation.tex}

\input{5_evaluation.tex}

\input{7_experiments.tex}

\section{Related Work} \label{sec:related}
\removed{We discuss related work from the following areas: package recommendation systems, semantic window queries, how-to queries, approximation techniques for \ilp formulations, and approximation techniques for subclasses of package queries.}

\paratitle{Package recommendations.}
Package or set-based recommendation systems are closely related to package queries. 
A package recommendation system presents users with interesting sets of items that satisfy some global conditions. These systems are usually driven by specific application scenarios. 
For instance, in the CourseRank~\cite{course-rank} system, the items to be recommended are university \emph{courses}, and the types of constraints are course-specific (e.g., prerequisites, incompatibilities, etc.). 
\emph{Satellite packages}~\cite{BasuRoy:2010:CEC:1807167.1807258} are sets of items, such as smartphone accessories, that are compatible with a ``central'' item, such as a smartphone. 
Other related problems in the area of package recommendations are \emph{team formation}~\cite{team,team-ilp},
and recommendation of \emph{vacation} and \emph{travel packages}~\cite{munmund}.
Queries expressible in these frameworks are also expressible in \paql, but the opposite does not hold.
The complexity of set-based package recommendation problems is studied in~\cite{deng}, where the authors show that computing top-$k$ packages with a conjunctive query language is harder than NP-complete.

\paratitle{Semantic window queries and Searchlight.}
Packages are also related to \emph{semantic windows}~\cite{Kalinin:2014:IDE:2588555.2593666}. A semantic window defines a contiguous subset of a grid-partitioned space with certain global properties. For instance, astronomers can partition the night sky into a grid, and look for regions of the sky whose overal brightness is above a specific threshold. If the grid cells are precomputed and stored into an input relation, these queries can be expressed in \paql by adding a global constraint (besides the brightness requirement) that ensures that all cells in a package must form a contiguous region in the grid space. Packages, however, are more general than semantic windows because they allow regions to be non-contiguous, or to contain gaps. Moreover, package queries also allow optimization criteria, which are not expressible in semantic window queries.

A recent extension to methods for answering semantic window queries is
Searchlight~\cite{kalinin2015searchlight}, which expresses these
queries in the form of constraint programs. Searchlight uses
in-memory synopses to quickly estimate aggregate values of contiguous
regions. However, it does not support synopses for non-contiguous
regions, and thus it cannot solve arbitrary package queries.
Searchlight has several other major differences with our work: (1)~it
computes optimal solutions by enumerating the feasible ones and
retaining the optimal, whereas our methods do not require enumeration;
(2)~Searchlight assumes that the solver implements redundant and
arbitrary data access paths while solving the problems, whereas our
approach decouples data access from the solving procedure;
(3)~Searchlight does not provide a declarative query language such as
PaQL; (4)~unlike \bt, Searchlight does not allow solvers to scale up
to a very large number of variables. At the time of this submission,
Searchlight has not been made available by the authors and thus we
could not run a comparison for the types of queries that it can
express.

\paratitle{How-to queries.}
Package queries are related to how-to queries~\cite{meliou2012tiresias}, as they both use an \ilp formulation to translate the original queries. However, there are several major differences between package queries and how-to queries: 
package queries specify tuple collections, whereas how-to queries specify updates to underlying datasets; 
package queries allow a tuple to appear multiple times in a package result, while how-to queries do not model repetitions; 
\paql is SQL-based whereas how-to queries use a variant of Datalog; \paql supports arbitrary Boolean formulas in the \ssf{SUCH} \ssf{THAT} clause, whereas how-to queries can only express conjunctive conditions.

\added{\looseness -1
\paratitle{Constraint query languages.}
The principal idea of constraint query languages
(CQL)~\cite{kanellakis1995constraint} is that a tuple can be
generalized as a conjunction of constraints over variables. This
principle is very general and creates connections between declarative
database languages and constraint programming. However, prior work
focused on expressing constraints over tuple values, rather than over
sets of tuples. In this light, \paql follows a similar approach to CQL
by embedding in a declarative query language methods that handle
higher-order constraints. However, our package query engine design
allows for the direct use of ILP solvers as black box components,
automatically transforming problems and solutions from one domain to
the other. In contrast, CQL needs to appropriately adapt the
algorithms themselves between the two domains, and existing literature
does not provide this adaptation for the constraint types in \paql.
}

\paratitle{\ilp approximations.}
There exists a large body of research in approximation algorithms for problems that can be modeled as integer linear programs. A typical approach is \emph{linear programming relaxation}~\cite{williamson2011design} in which the integrality constraints are dropped and variables are free to take on real values. These methods are usually coupled with \emph{rounding} techniques that transform the real solutions to integer solutions with provable approximation bounds. 
None of these methods, however, can solve package queries on a large scale because they all assume that the LP solver is used on the entire problem. 
Another common approach to approximate a solution to an \ilp problem is the \emph{primal-dual method}~\cite{goemans1997primal}. All primal-dual algorithms, however, need to keep track of all primal and dual variables and the coefficient matrix, which means that none of these methods can be employed on large datasets.
On the other hand, rounding techniques and primal-dual algorithms could potentially benefit from the \bt algorithm to break down their complexity on very large datasets.

\paratitle{Approximations to subclasses of package queries.}
Like package queries, \emph{optimization under parametric aggregation
constraints} (OPAC) queries~\cite{Guha:2003:EAO:1315451.1315518} can
construct sets of tuples that collectively satisfy summation
constraints. However, existing solutions to OPAC queries have several
shortcomings:
(1)~they do not handle tuple repetitions;
(2)~they only address \emph{multi-attribute knapsack queries}, a
subclass of package queries where all global constraints are of the
form \ssf{SUM}$()$ $\le c$, with a \ssf{MAXIMIZE} \ssf{SUM}$()$
objective criterion;
(3)~they may return infeasible packages;
(4)~\removed{More importantly, }they are conceptually different from \bt, as
they generate approximate solutions in a pre-processing step, and
packages are simply retrieved at query time using a multi-dimensional
index. In contrast, \bt does not require pre-computation of packages.
Package queries also encompass \emph{submodular} optimization queries, whose
recent approximate solutions use greedy distributed
algorithms~\cite{mirzasoleiman2013distributed}.

\section{Conclusion and Future Work} \label{sec:conclusion}

In this paper, we introduced a complete system that supports the
specification and efficient evaluation of package queries. We
presented \paql, a declarative extension to \sql, and theoretically
established its expressiveness, and we developed a flexible
approximation method, with strong theoretical guarantees, for the
evaluation of \paql queries on large-scale datasets. Our experiments
on real-world and benchmark data demonstrate that our scalable
evaluation strategy is effective and efficient over varied data sizes
and query workloads, and remains robust under suboptimal conditions
and parameter settings.

In our future work, we plan to extend our evaluation methods to larger
classes of package queries, including multi-relation and non-linear
queries, and we intend to investigate parallelization strategies for
\bt. Package queries pose interesting challenges on usability aspects
as well: Our goal is to develop interaction and learning methods that
let users identify their ideal packages without having to specify a
full and precise \paql query.

\ifshort
	\balance
	\bibliographystyle{abbrv}
	\bibliography{refs}
	\balancecolumns
\else
	\bibliographystyle{abbrv}
	{\small \bibliography{refs}}	
	\clearpage
    \input{10_appendices.tex}
\fi

\end{document}

%% file: macros.tex
\newif\ifshort
\shortfalse

\usepackage{balance}
\usepackage{graphicx}
\usepackage{color}
\usepackage{url}
\usepackage{hyperref}
\usepackage{xspace}
\usepackage{pslatex}
\usepackage{enumitem}
\usepackage{epstopdf}
\usepackage{color}
\usepackage[usenames,dvipsnames]{xcolor}
\usepackage{soul}
\usepackage{appendix}
\usepackage{algorithm}
\usepackage{algpseudocode}
\usepackage{amsmath}
\usepackage{amssymb}
\usepackage[mathcal]{euscript}
\usepackage{bbm}
\usepackage{dsfont}
\usepackage{multirow}
\usepackage{tabularx}
\usepackage{relsize}
\usepackage{thmtools}
\usepackage{thm-restate}
\usepackage[labelfont=bf,textfont=bf,justification=raggedright]{subcaption}
\usepackage[labelfont={bf},textfont={bf}]{caption}
\usepackage[T1]{fontenc}
\usepackage[export]{adjustbox}
\usepackage{turnstile}
\usepackage{soul}
\usepackage{wrapfig}
\usepackage{float}
\usepackage{cleveref}
\usepackage[normalem]{ulem}
\usepackage{fancyvrb}

\usepackage{textcomp} 
\usepackage{microtype}

\usepackage{expl3}
\ExplSyntaxOn
\cs_new_eq:NN \Rrepeat \prg_replicate:nn
\ExplSyntaxOff

\newcommand{\ssf}[1]{${\sf #1}$}
\newcommand{\bsf}[1]{$\textsf{\textbf{#1}}$}

\newcommand{\paql}{PaQL\xspace}
\newcommand{\ilp}{ILP\xspace}

\newcommand{\sql}{SQL\xspace}
\newcommand{\cplex}{CPLEX\xspace}
\newcommand{\tpch}{\mbox{TPC-H}\xspace}
\newcommand{\galaxy}{\mbox{Galaxy}\xspace}
\newcommand{\bt}{\textsc{SketchRefine}\xspace}
\newcommand{\opt}{\textsc{Direct}\xspace}
\newcommand{\sdss}{SDSS\xspace}

\newcommand{\redspace}{\mathcal{P}}
\newcommand{\unsolved}{\mathcal{S}}
\newcommand{\repr}{\tilde{t}}
\newcommand{\tuple}{t}

\newcommand{\pack}{p}
\newcommand{\rpack}{{\sf P}}
\newcommand{\thres}{\tau}
\newcommand{\query}{\mathcal{Q}}

\newcommand{\initquery}{\query[\text{\ssf{\tilde{R}}}]}

\newcommand{\augmquery}{\query[G_j]}

\newcommand{\failed}{\mathcal{F}}
\newcommand{\prioQ}{\mathcal{U}}
\newcommand{\sat}{\sststile[ss]{}{}}
\newcommand{\qrelax}[1]{\sststile[ss]{(1-\epsilon)^{#1}}{}}
\newcommand{\qapprox}[1]{\sststile[ss]{}{(1-\epsilon)^{#1}}}
\newcommand{\qopt}{\sststile[ss]{}{\ast}}
\newcommand{\qoptrelax}[1]{\sststile[ss]{(1-\epsilon)^{#1}}{\ast}}
\newcommand{\qapproxrelax}[2]{\sststile[ss]{(1-\epsilon)^{#1}}{(1-\epsilon)^{#2}}}
\newcommand{\proj}{\pi}

\newcommand{\bet}{\succeq}
\newcommand{\SUM}{\text{\ssf{SUM}}}

\newcommand{\attr}{\text{\ssf{.attr}}}
\newcommand{\csense}{\odot~}

\newtheorem{theorem}{Theorem}
\newtheorem{lemma}[theorem]{Lemma}
\newtheorem{corollary}[theorem]{Corollary}
\newtheorem{definition}{Definition}
\newtheorem{example}{Example}

\newcommand{\initialalgo}{Sketch}
\newcommand{\augmentalgo}{Refine}

\newcommand{\btalgo}{Refine}

\algnewcommand{\LineComment}[1]{\State \(\triangleright\) #1}
\algtext*{EndWhile}
\algtext*{EndIf}
\algtext*{EndFor}
\algtext*{EndProcedure}

\newcommand{\paratitle}{\smallskip\noindent\textbf}

\newcommand{\removed}{\bgroup\markoverwith{\textcolor{red}{\rule[.5ex]{2pt}{0.4pt}}}\ULon}
\renewcommand{\removed}[1]{}
\newcommand{\added}[1]{{\color{ForestGreen}{#1}}}
\renewcommand{\added}[1]{#1}

\newcommand\litem[1]{\item{\bfseries #1.}}

\newcommand{\mass}{\raisebox{5pt}{\includegraphics[width=8pt]{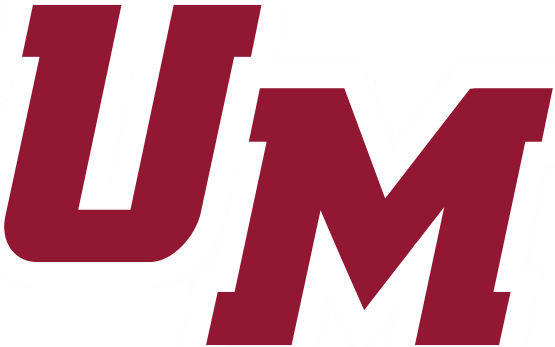}}}
\newcommand{\nyu}{\raisebox{4pt}{\includegraphics[width=5pt]{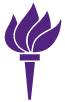}}}

\ifshort
\else
\fi

%% file: 1_intro.tex
\section{Introduction} \label{sec:intro}

Traditional database queries follow a simple model: they define
constraints, in the form of selection predicates, that each tuple in
the result must satisfy. This model is computationally efficient, as
the database system can evaluate each tuple individually to determine
whether it satisfies the query conditions. However, many practical,
real-world problems require a collection of result tuples to satisfy
constraints collectively, rather than individually.

    \begin{example}[Meal planner]\label{ex:meals}
        A dietitian needs to design a daily meal plan for a patient.
        She wants a set of three gluten-free meals, between 2,000 and
        2,500 calories in total, and with a low total intake of
        saturated fats.
    \end{example}

    \begin{example}[Night sky]
        An astrophysicist is looking for rectangular regions of the
        night sky that may potentially contain previously unseen
        quasars. Regions are explored if their overall redshift is
        within some specified parameters, and ranked according to
        their likelihood of containing a
        quasar~\cite{Kalinin:2014:IDE:2588555.2593666}.
    \end{example}

In these examples, there are some conditions that can be verified on
individual data items (e.g., gluten content in a meal), while others
need to be evaluated on a collection of items (e.g., total calories).
Similar scenarios arise in a variety of application domains, such as
investment planning, product bundles, course selection~\cite{course-rank},
team formation~\cite{team-ilp,team},
vacation and travel planning~\cite{munmund}, and computational
creativity~\cite{Pinel:2014:CCC:2559206.2574794}. Despite the clear
application need, database systems do not currently offer support for
these problems, and existing work has focused on application- and
domain-specific approaches~\cite{team-ilp,munmund,team,course-rank}.

In this paper, we present a domain-independent, database-centric
approach to address these challenges: We introduce a full-fledged
system that supports \emph{package queries}, a new query model that
extends traditional database queries to handle complex constraints and
preferences over answer sets.
Package queries are defined over traditional relations, but return
\emph{packages}. A package is a collection of tuples that (a)
individually satisfy \emph{base predicates} (traditional selection
predicates), and (b) collectively satisfy \emph{global predicates}
(package-specific predicates).
Package queries are combinatorial in nature\added{:}\removed{, and multiple packages may
satisfy a query specification. Thus,} the result of a package query is a
(potentially infinite) set of packages, and an \emph{objective
criterion} can define a preference ranking among them.

Extending traditional database functionality to provide \removed{full-fledged}
support for packages, rather than supporting packages at the
application level, is justified by two reasons:
First, the features of packages and the algorithms for
constructing them are not unique to each application; therefore, the
burden of package support should be lifted off application
developers, and database systems should support package queries like traditional queries. Second, the data used to construct packages typically resides in a
database system, and packages themselves are structured data objects
that should naturally be stored in and manipulated by a database
system.

Our work in this paper addresses three important challenges:

\begin{enumerate}[topsep=0ex,partopsep=1ex,parsep=1ex,itemsep=0ex,leftmargin=0mm,listparindent=0mm,labelsep=0.8mm,itemindent=3mm]

\litem{Declarative specification of packages}
\sql enables the declarative specification of properties that result tuples should satisfy.
In
Example~\ref{ex:meals}, it is easy to specify the exclusion of meals
with gluten using a regular selection predicate in \sql. However, it is difficult to specify global constraints (e.g., total calories of a set of meals should be between 2,000 and 2,500 calories). 
Expressing such a query in \sql
requires either complex self-joins that explode the size of the
query,
or recursion, which results in extremely complex
queries that are hard to specify and optimize (Section~\ref{sec:package-queries}).
Our goal is to maintain the declarative power of \sql, while extending
its expressiveness to allow for the easy specification of packages.

\litem{Evaluation of package queries}
\begin{figure}
\centering
\includegraphics[scale=0.25]{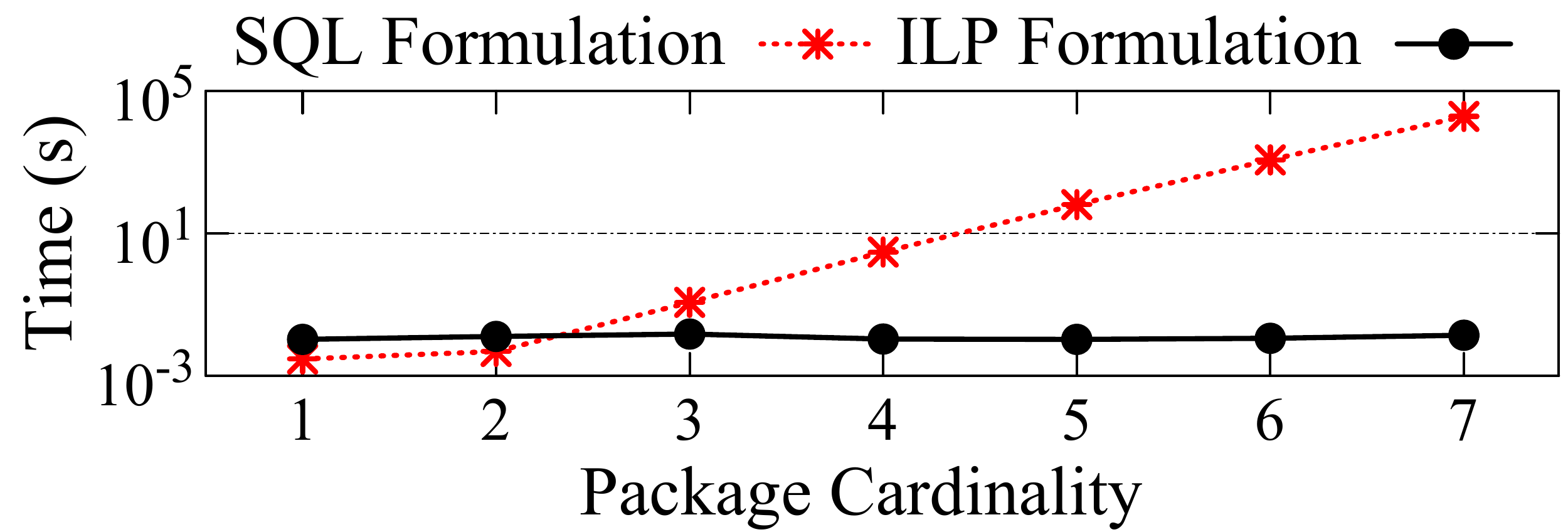}
\vspace{-2mm}
\caption{Traditional database technology is ineffective at package evaluation, and the runtime of the na\"ive \sql formulation (Section~\ref{sec:package-queries}) of a package query grows exponentially.
In contrast, tools such as \ilp solvers are more effective.
}
\label{fig:sql}
\end{figure}
Due to their combinatorial complexity, package queries are harder to
evaluate than traditional database queries~\cite{deng}. Package
queries are in fact as hard as integer linear programs (\ilp)
(Section~\ref{sec:expressiveness}). Existing database technology is
ineffective at evaluating package queries, even if one were to express
them in \sql. Figure~\ref{fig:sql} shows the performance of evaluating
a package query expressed as a multi-way self-join query in traditional \sql
(described in detail in Section~\ref{sec:package-queries}). As the cardinality of
the package increases, so does the number of joins, and the runtime
quickly becomes prohibitive: In a small set of 100 tuples from the
Sloan Digital Sky Survey dataset~\cite{sdss}, \sql evaluation takes
almost 24 hours to construct a package of 7 tuples. Our goal is to extend the
database evaluation engine to take advantage of external tools, such
as ILP solvers, which are more effective for combinatorial problems.

\litem{Performance and scaling to large datasets}
Integer programming solvers have two major limitations: they require
the entire problem to fit in main memory, and they fail when the
problem is too complex (e.g., too many variables and/or too many
constraints). Our goal is to overcome these limitations through
sophisticated evaluation methods that allow solvers to scale to large
data sizes.

\end{enumerate}

In this paper, we address these challenges by designing language and
algorithmic support for package query specification and evaluation.
Specifically, we make the following contributions.

\begin{itemize}[topsep=0pt,partopsep=1ex,parsep=1ex,itemsep=-1ex,leftmargin=3mm]
    \item
    We present \paql (Package Query Language), a declarative language
    that provides simple extensions to standard \sql to support
    constraints at the package level. We prove that \paql is at least
    as expressive as integer linear programming, which implies that
    evaluation of package queries is NP-hard
    (Section~\ref{sec:package-queries}).
    \item
    We present a fundamental evaluation strategy \opt that
    combines the capabilities of databases and constraint optimization
    solvers to derive solutions to package queries. The core of our
    approach is a set of translation rules that transform a package
    query to an integer linear program. This translation allows for
    the use of highly-optimized external tools for the evaluation of
    package queries (Section~\ref{sec:ilp}).
    \item
    We introduce an offline data partitioning strategy 
    that allows package query evaluation to scale to
    large data sizes. The core of our evaluation strategy
    \bt lies on separating the package computation into
    multiple stages, each with small subproblems, which the solver can
    evaluate efficiently. 
    In the first stage, the algorithm ``sketches'' an initial
    sample package from a set of representative tuples, while the
    subsequent stages ``refine'' the current package by solving an \ilp
    within each partition. 
    \bt guarantees a $(1
    \pm \epsilon)^6$-factor approximation for the package results compared to \opt
    (Section~\ref{sec:evaluation}).
    \item
    We present an extensive experimental evaluation on both real-world data and the \tpch benchmark (\Cref{sec:experiments}) that shows that our query evaluation method \bt:
    (1)~is able to produce packages an order of magnitude faster than the \ilp solver used directly on the entire problem; 
    (2)~scales up to sizes that the solver cannot manage directly; 
    (3)~produces packages of very good quality in terms of objective value; 
    (4)~is robust to partitioning built in anticipation of different workloads.
\end{itemize}

%% file: 2_package-queries.tex
\section{Language support for packages} \label{sec:package-queries}

Data management systems do not natively support package queries. While
there are ways to express package queries in \sql, these are
cumbersome and inefficient.

\paratitle{Specifying packages with self-joins.}
When packages have strict cardinality (number of tuples), and only in this case,
it is possible to express package queries using
traditional self-joins.
For instance, self-joins can express the query of
Example~\ref{ex:meals} as follows:

{\small
\begin{tabbing}
\=\hspace*{1.55cm}\=\hspace*{3cm}\= \kill
\>\scalebox{1.0}[1.0]{\ssf{SELECT}} * \>\scalebox{1.0}[1.0]{\ssf{FROM}} \ssf{Recipes} \ssf{R1}, \ssf{Recipes} \ssf{R2}, \ssf{Recipes} \ssf{R3}\\
\>\scalebox{1.0}[1.0]{\ssf{WHERE}} \>\ssf{R1.pk} $<$ \ssf{R2.pk} \scalebox{1.0}[1.0]{\ssf{AND}} \ssf{R2.pk} $<$ \ssf{R3.pk} \ssf{AND}\\
\>~~~~\ssf{R1.gluten} $=$ `\ssf{free}' \scalebox{1.0}[1.0]{\ssf{AND}} \ssf{R2.gluten} $=$ `\ssf{free}' \scalebox{1.0}[1.0]{\ssf{AND}} \ssf{R3.gluten} $=$ `\ssf{free}'\\
\>~~~~\scalebox{1.0}[1.0]{\ssf{AND}} \ssf{R1.kcal} $+$ \ssf{R2.kcal} $+$ \ssf{R3.kcal} \scalebox{1.0}[1.0]{\ssf{BETWEEN}} $2.0$ \scalebox{1.0}[1.0]{\ssf{AND}} $2.5$\\
\>\scalebox{.95}[1.0]{\ssf{ORDER} \ssf{BY}} 
\>\ssf{R1.saturated\_fat} $+$ \ssf{R2.saturated\_fat} $+$ \ssf{R3.saturated\_fat}
\end{tabbing}}
This query is efficient only for constructing
packages with very small cardinality: larger
cardinality requires a larger number of self-joins, quickly rendering
evaluation time prohibitive (Figure~\ref{fig:sql}).
The benefit of this specification is that the optimizer can use the
traditional relational algebra operators, and augment its
decisions with package-specific strategies. 
However, this method does not apply for packages of unbounded cardinality.

\paratitle{Using recursion in \sql.}
More generally, \sql can express package queries by generating and testing each
possible subset of the input relation. This requires recursion to
build\removed{, as an intermediate result,} a \emph{powerset table}; checking
each set in the powerset table for the query conditions will yield the
result packages. This approach has three major drawbacks. First, it is
not declarative, and the specification is tedious and complex.
Second, it is not
amenable to optimization in existing systems. Third, it is extremely
inefficient to evaluate, because the powerset table generates an
exponential number of candidates.

\subsection{\paql: The Package Query Language} \label{sec:paql}

Our goal is to support package specification in a declarative and
intuitive way. In this section, we describe \paql, a declarative query
language that introduces simple extensions to \sql to define package
semantics and package-level constraints. 
We first show how \paql can express the query of Example~\ref{ex:meals},
as our running example, to demonstrate the new language features:

{\small
\begin{tabbing}
\hspace*{4mm}\=\hspace*{2cm}\=\hspace*{3cm}\= \kill
$\query$:
\>\ssf{SELECT}  \>\bsf{PACKAGE}$($\ssf{R}$)$ \ssf{AS} \ssf{P} \\
\>\ssf{FROM}    \>\ssf{Recipes} \ssf{R} \bsf{REPEAT} $0$ \\
\>\bsf{WHERE}   \>\ssf{R.gluten} $=$ `\ssf{free}' \\
\>\bsf{SUCH THAT}\>\ssf{COUNT}$($\ssf{P.}$\ast) =3$ \ssf{AND} \\
\>        \>\ssf{SUM}$($\ssf{P.kcal}$)$ \ssf{BETWEEN} $2.0$ \ssf{AND} $2.5$\\
\>\bsf{MINIMIZE}    \>\ssf{SUM}$($\ssf{P.saturated\_fat}$)$
\end{tabbing}}

\paratitle{Basic semantics.}
The new keyword \ssf{PACKAGE} differentiates \paql from
traditional \sql queries.
{\small
\begin{tabbing}
\hspace*{0.6cm}\=\hspace*{1.4cm}\=\hspace*{1.8cm}\=\hspace*{0.6cm}\=\hspace*{1.4cm}\= \kill
$\query_1$:
\>\ssf{SELECT}    \> *      
\>$\query_2$:
\>\ssf{SELECT} \>\bsf{PACKAGE}$($\ssf{R}$)$ \ssf{AS} \ssf{P}\\
\>\ssf{FROM}      \> \ssf{Recipes} \ssf{R}
\>\>\ssf{FROM}      \> \ssf{Recipes} \ssf{R}
\end{tabbing}}
The semantics of $\query_1$ and $\query_2$ are fundamentally different: $\query_1$ is
a traditional SQL query, with a unique, finite result set (the entire
\ssf{Recipes} table), whereas there are infinitely many packages that
satisfy the package query $\query_2$: all possible multisets of tuples
from the input relation.
\looseness -1
The result of a package query like $\query_2$ is a set of packages. Each
package resembles a relational table containing a collection of tuples
(with possible repetitions) from relation \ssf{Recipes}, and therefore
a package result of $\query_2$ follows the schema of \ssf{Recipes}. 
\removed{The 
\ssf{FROM} clause of a package query may contain multiple relations;}
\added{\paql syntax permits multiple relations in the \ssf{FROM} clause;}
in that case, the packages produced will follow the schema of the join
result. 

\added{In the remainder of this paper, we focus on package queries
without joins. This is for two reasons: (1)~The join operation is part
of traditional \sql and can occur before package-specific
computations. (2)~There are important implications in the
consideration of joins that extend beyond the scope of our work.
Specifically, materializing the join result is not always necessary, but
rather, there are space-time tradeoffs and system-level solutions that
can improve query performance in the presence of joins.
Section~\ref{sec:discussion} offers a high-level discussion of these
points, but these extensions are orthogonal to the techniques we
present in this work. }

\removed{Since the join operation is part of traditional \sql and can
occur before package-specific computations, without loss of
generality, in this paper we assume a single relation in the
\ssf{FROM} clause of \paql.}

Although semantically valid, a query like $\query_2$ would not occur in practice,
as most application scenarios expect few, or even exactly one result.
We proceed to describe the additional constraints in the example query
$\query$ that restrict the number of package results.

\paratitle{Repetition constraint.}
The \ssf{REPEAT} $0$ statement in query $\query$ specifies that no tuple
from the input relation can appear multiple times in a package result.
If this restriction is absent (as in query $\query_2$), tuples can be
repeated an unlimited number of times. By allowing no repetitions,
$\query$ restricts the package space from infinite to $2^n$, where $n$ is
the size of the input relation. Generalizing, the specification
\ssf{REPEAT} $\mathcal{K}$ allows a package to repeat tuples up to $\mathcal{K}$ times,
resulting in $(2+\mathcal{K})^n$ candidate packages.

\paratitle{Base and global predicates.}
A package query defines two types of predicates. 
A \emph{base predicate}, defined in the \ssf{WHERE} clause, is equivalent to
a selection predicate and can be evaluated with standard \sql: any
tuple in the package needs to \emph{individually} satisfy the base
predicate. For example, query $\query$ specifies the base predicate:
\ssf{R.gluten} $=$ `\ssf{free}'. Since base predicates directly filter
input tuples, they are specified over the input relation \ssf{R}.
\emph{Global predicates} are the core of package queries, and they
appear in the new \ssf{SUCH} \ssf{THAT} clause. Global predicates are higher-order
than base predicates: they cannot be evaluated on individual tuples,
but on tuple collections. Since they describe package-level
constraints, they are specified over the package result \ssf{P}, e.g.,
\ssf{COUNT}$($\ssf{P.}$\ast) = 3$, which limits the query results to
packages of exactly 3 tuples.  

The global predicates shown in query $\query$ abbreviate aggregates that are in reality subqueries. For example, \ssf{COUNT}$($\ssf{P.}$\ast) = 3$, is 
an abbreviation for 
$($\ssf{SELECT} \ssf{COUNT}$(\ast)$ \ssf{FROM} \ssf{P}$) =3$.
Using subqueries, \paql can express arbitrarily complex global constraints among aggregates over a package.

\looseness -1
\paratitle{Objective clause.}
The objective clause specifies a ranking among candidate
package results, and appears with either the \ssf{MINIMIZE} or
\ssf{MAXIMIZE} keyword. It is a condition on the
package-level, and hence it is specified over the package result
\ssf{P}, e.g., \ssf{MINIMIZE} \ssf{SUM}$($\ssf{P.saturated\_fat}$)$. 
Similarly to global predicates, this form is a shorthand for \ssf{MINIMIZE} 
$($\ssf{SELECT} \ssf{SUM}$($\ssf{saturated\_fat}$)$ \ssf{FROM} \ssf{P}$)$.
\added{A \paql query with an objective clause returns a single
result: the package that optimizes the value of the objective. The
evaluation methods that we present in this work focus on such queries.
In prior work~\cite{brucato2014packagebuilder}, we described
preliminary techniques for returning multiple packages in the absence
of optimization objectives, but a thorough study of such methods is
left to future work.
}

\smallskip

While \paql can use arbitrary aggregate functions in the global predicates and
the objective clause, in this work, we assume that package queries are
limited to \emph{linear} functions We defer the study of non-linear functions and UDFs to future work. 
\ifshort
\else
Appendix~\ref{sec:paql-syntax} includes a formal specification of the \paql syntax.
\fi

\subsection{Expressiveness and complexity of \paql} \label{sec:expressiveness}
Package queries are at least as hard as integer linear programs,
as the following theoretical results establish.
\begin{theorem}[Expressiveness of \paql] \label{lemma:expressiveness}
\hfill
Every integer linear program can be expressed as a package query in \paql.
\end{theorem}
\ifshort
We include the proofs of our theoretical results in the full version
of the paper~\cite{extended-version}. 
\else
The proof for Theorem~\ref{lemma:expressiveness} is in Appendix~\ref{app:expressiveness}.
\fi
At a high level, the proof employs a reduction from an integer linear
program to a \paql query. The reduction maps the linear constraints
and objective into the corresponding \paql clauses, using the
constraint coefficients to generate the input relation for the package
query.
As a direct consequence of Theorem~\ref{lemma:expressiveness}, we also
obtain the following result about the complexity of evaluating package
queries.
\begin{corollary}[Complexity of Package Queries]
Package queries are NP-hard.
\end{corollary}
In Section~\ref{sec:ilp}, we extend the result of
Theorem~\ref{lemma:expressiveness} to also show that every \paql query
that does not contain non-linear functions can be expressed as an
integer linear program, through a set of translation rules.
This
transformation is the first step in package evaluation, but, due to
the limitations of \ilp solvers, it is not efficient or scalable in
practice. To make package evaluation practical, we develop \bt, a
technique that augments the \ilp transformation with a partitioning
mechanism, allowing package evaluation to scale to large datasets
(Section~\ref{sec:evaluation}).

%% file: 4_ilp_evaluation.tex
\section{\ilp Formulation} \label{sec:ilp}

In this section, we present an \ilp formulation for package queries.
This formulation is at the core of our evaluation methods \opt and \bt.
The results presented in this section are inspired by the translation
rules employed by Tiresias~\cite{meliou2012tiresias} to answer
\emph{how-to queries}. However, there are several important
differences between how-to and package queries, which we discuss
extensively in the overview of the related work (Section~\ref{sec:related}).

\subsection{\protect\scalebox{1.0}[1.0]{\paql to \ilp Translation}} \label{sec:formulation}

Let \ssf{R} indicate the input relation, $n=|{\sf R}|$ the number of tuples in \ssf{R}, \ssf{R.attr} an attribute of \ssf{R}, \ssf{P} a package, $f$ a linear aggregate function (such as \ssf{COUNT} and \ssf{SUM}), 
$\csense \in \{ \le, \ge \}$ 
a constraint inequality, and $v \in \mathbb{R}$ a constant.
    For each tuple $\tuple_i$ from \ssf{R}, $1 \leq i \leq n$, the \ilp problem includes a nonnegative integer variable $x_i$ ($x_i \ge 0$), indicating the number of times $\tuple_i$ is included in an answer package.
    We also use $\bar{x}=\langle x_1,x_2,\dots,x_n \rangle$ to denote the vector of all integer variables.
A \paql query is formulated as an \ilp problem using the following translation rules:
\begin{enumerate}[topsep=1\topsep,partopsep=0pt,parsep=1\parsep,itemsep=0pt,leftmargin=3mm,listparindent=0pt,labelsep=0.8mm]
    \litem{Repetition constraint}
    The \ssf{REPEAT} keyword, expressible in the \ssf{FROM} clause, restricts the domain that the variables can take on. Specifically, \ssf{REPEAT} $\mathcal{K}$ implies ${0 \le x_i \le \mathcal{K}+1}$.
    
    \litem{Base predicate} 
    Let $\beta$ be a base predicate, e.g., \ssf{R.gluten} $=$ `\ssf{free}',
    and ${\sf R}_\beta$ the relation containing tuples from \ssf{R} satisfying $\beta$. We encode $\beta$ by setting $x_i = 0$ for every tuple $\tuple_i\not\in{\sf R}_\beta$.

    \litem{Global predicate} 
    Each global predicate in the \ssf{SUCH} \ssf{THAT} clause takes the form $f({\sf P})$ $\csense v$. For each such predicate, we derive a linear function $f'(\bar x)$ over the integer variables.
    A cardinality constraint $f({\sf P}) = {\sf COUNT}({\sf P.}\ast)$ is linearly translated into ${f'(\bar{x}) = \sum_{i} x_i}$. A summation constraint $f({\sf P}) = {\sf SUM(P.attr)}$ is linearly translated into ${f'(\bar{x}) = \sum_{i} (\tuple_i.{\sf attr}) x_i}$. 
    We further illustrate the translation with two non-trivial examples:
    \begin{itemize}[topsep=0pt,partopsep=1ex,parsep=1ex,itemsep=-1ex,leftmargin=3ex]
    \item 
    ${\sf AVG}({\sf P.attr}) \le v$ is translated as
\[
\textstyle\sum_{i} (\tuple_i.{\sf attr}) x_i / \sum_{i} x_i \le v 
\;\equiv\; 
\textstyle\sum_{i} (\tuple_i.{\sf attr} - v)x_i                   \le 0 
\]
    \item 
    $($\ssf{SELECT} \ssf{COUNT}$(\ast)$ \ssf{FROM} \ssf{P} \ssf{WHERE} \ssf{P.carbs} $> 0) \ge ($\ssf{SELECT} \ssf{COUNT}$(\ast)$ \ssf{FROM} \ssf{P} \ssf{WHERE} \ssf{P.protein} $\le 5)$
    is translated as 
    
\begin{align*}
\begin{array}{l}
{\sf R_c} := \{\tuple_i \in {\sf R} \mid t_i{\sf.carbs} > 0\} \\
{\sf R_p} := \{\tuple_i \in {\sf R} \mid t_i{\sf.protein} \leq 5\} \\
\mathds{1}_{\sf R_c}(\tuple_i) := 1\ {\rm if} \ \tuple_i \in {\sf R_c};\  0\  {\rm otherwise}. \\
\mathds{1}_{\sf R_p}(\tuple_i) := 1\ {\rm if} \ \tuple_i \in {\sf R_p};\  0\  {\rm otherwise}. \\
\sum_{i} (\mathds{1}_{\sf R_c}(\tuple_i) - \mathds{1}_{\sf R_p}(\tuple_i))x_i \geq 0
\end{array}
\end{align*}
    \end{itemize}
General Boolean expressions over the global predicates can be
encoded into a linear program with the help of Boolean variables and
linear transformation tricks found in the
literature~\cite{bisschop2006aimms}.
    
    \litem{Objective clause} 
    We encode \ssf{MAXIMIZE} $f({\sf P})$ as $\max f'(\bar{x})$, where $f'(\bar{x})$ is the encoding of $f({\sf P})$. Similarly \ssf{MINIMIZE} $f({\sf P})$ is encoded as $\min f'(\bar{x})$. If the query does not include an objective clause, we add the \emph{vacuous} objective $\max \sum_{i} 0 \cdot x_i$.

\end{enumerate}

We call the relations ${\sf R_{\beta}}$, ${\sf R_c}$, and ${\sf R_p}$ described above \emph{base relations}.
This formulation, together with Theorem~\ref{lemma:expressiveness},
shows that \added{single-relation} package queries with linear constraints correspond exactly
to \ilp problems.

\subsection{Query Evaluation with \large{\opt}} \label{sec:direct}

Using the \ilp formulation, we develop our basic evaluation method for
package queries, called \opt. We later extend this technique to our
main algorithm, \bt, which supports efficient package evaluation in
large data sets (Section~\ref{sec:evaluation}).

Package evaluation with \opt employs three simple steps:

\begin{enumerate}[topsep=0ex,partopsep=1ex,parsep=1\parsep,itemsep=0pt,leftmargin=3mm,listparindent=0pt,labelsep=0.8mm]
	\litem{\ilp formulation} 
	We transforms a \paql query to an
    \ilp problem using the rules described in
    Section~\ref{sec:formulation}. 

    \litem{Base relations} 
    We compute the base relations, such
    as ${\sf R_{\beta}}$, ${\sf R_c}$, and ${\sf R_p}$, with a series
    of standard \sql queries, one for each, or by simply scanning
    \ssf{R} once and populating these relations simultaneously.
    After this phase, all variables $x_i$ such that ${x_i = 0}$ can be eliminated from the \ilp problem because the corresponding tuple $\tuple_i$ cannot appear in any package solution. This can significantly reduce the size of the problem.

    \litem{\ilp execution} 
    We employ an off-the-shelf \ilp
    solver as a black box to get a solution $\bar{x}^*$ for all the
    integer variables $x_i$ of the problem. Each $x_i^*$ informs the number of times
    tuple $\tuple_i$ should be included in the answer package.
     
\end{enumerate}

\removed{\opt can evaluate top-$k$ package queries by solving $k$ \ilp problems
iteratively: at each iteration after the first, it adds a new
constraint to the problem to ensure that the objective value of the
new solution is strictly lower (for a maximization query) or strictly
greater (for a minimization query) than the previous one.
This method can also be used to generate multiple feasible packages
for queries lacking an objective clause. It can also be optimized by
using solver-specific features. For instance, IBM's \cplex~\cite{cplex} allows users to obtain a ``pool'' of feasible solutions, which can be used to produce more than one package in a single run of the solver.}

The \opt algorithm has two crucial drawbacks. First, it is only applicable if the input relation is small enough to fit entirely in main memory: \ilp solvers, such as IBM's \cplex, require the entire problem to be loaded in memory before execution.
Second, even for problems that fit in main memory, this approach may fail due to the complexity of the integer problem. In fact, integer linear programming is a notoriously hard problem, and modern \ilp solvers use algorithms, such as \emph{branch-and-cut}~\cite{padberg1991branch}, that often perform well in practice, but can ``choke'' even on small problem sizes due to their exponential worst-case complexity~\cite{cook1990complexity}.
This may result in
unreasonable performance due to solvers using too many resources (main memory, virtual memory, CPU time), eventually thrashing the entire system.

%% file: 5_evaluation.tex
\section{Scalable Package Evaluation} \label{sec:evaluation}

In this section, we present \bt, an approximate divide-and-conquer
evaluation technique for efficiently answering package queries on
large datasets.
\bt smartly decomposes a query into smaller queries, formulates
them as \ilp problems, and employs an \ilp solver as a ``black box''
evaluation method to answer each individual query. By breaking down
the problem into smaller subproblems, the algorithm avoids the
drawbacks of the \opt approach. Further, we prove that \bt is guaranteed to always produce
feasible packages with an approximate objective value
(Section~\ref{sec:quality}).

The algorithm is based on an important observation: \emph{similar
tuples are likely to be interchangeable within packages}. A group of
similar tuples can therefore be ``compressed'' to a single
\emph{representative tuple} for the entire group. \bt \emph{sketches} an
initial answer package using only the set of representative tuples,
which is substantially smaller than the original dataset.
This initial solution is then \emph{refined} by evaluating a subproblem
for each group, iteratively replacing the representative tuples
in the current package solution with original tuples from the dataset.
Figure \ref{fig:redspace} provides a high-level illustration of the three main steps of \bt:

\begin{figure*}
\centering
\includegraphics[width=\textwidth]{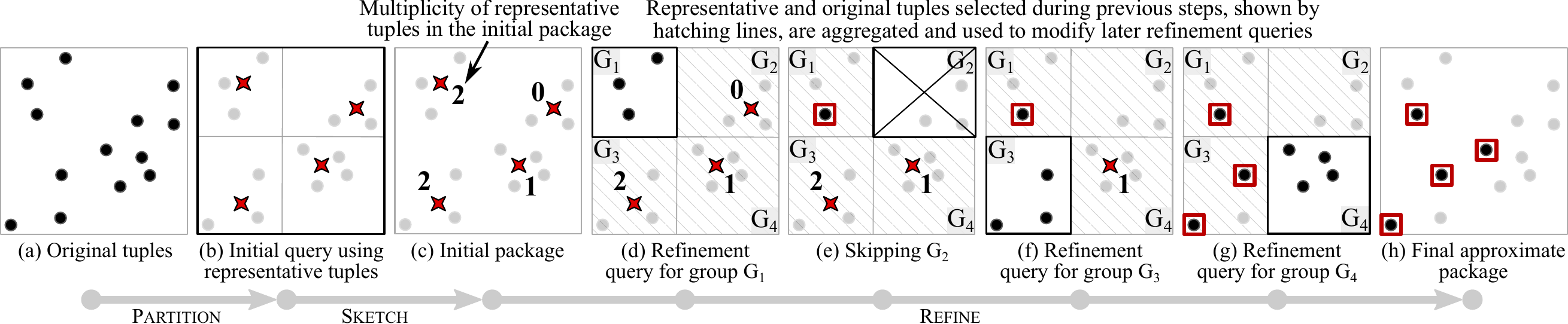}
\caption{
The original tuples (a) are partitioned into four groups and a representative is constructed for each group (b). 
The initial sketch package (c) contains only representative tuples, with possible repetitions up the size of each group. 
The refine query for group $G_1$ (d) involves the original tuples from $G_1$ and the aggregated solutions to  all other groups ($G_2$, $G_3$, and $G_4$). 
Group $G_2$ can be skipped (e) because no representatives could be picked from it. 
Any solution to previously refined groups are used while refining the solution for the remaining groups (f and g). The final approximate package (h) contains only original tuples. 
}
\label{fig:redspace}
\end{figure*}

\begin{enumerate}[topsep=0ex,partopsep=1ex,parsep=1\parsep,itemsep=0pt,leftmargin=3mm,listparindent=0pt,labelsep=0.8mm]
    \litem{Offline partitioning (Section~\ref{sec:index})}\looseness -1
    The algorithm assumes a partitioning of the data into groups of
    similar tuples. This partitioning is performed offline (not at
    query time), and our experiments show that \bt remains very
    effective even with partitionings that do not match the query
    workload (Section~\ref{sec:exp:cover}). In our implementation, we
    partition data using $k$-dimensional quad
    trees~\cite{finkel1974quad}, but other partitioning schemes are
    possible.

    \litem{Sketch (Section~\ref{sec:initial})}
    \bt sketches an initial package 
    by evaluating the package query only over the set of representative tuples. 
    
    \litem{Refine (Section~\ref{sec:refine})} 
    Finally, \bt transforms the initial package into a complete package by replacing each representative tuple with some of the original tuples from the same group, one group at a time.  
\end{enumerate}
\bt always constructs \emph{approximate feasible} packages, i.e., packages that satisfy all the query constraints, but with a possibly sub-optimal objective value that is guaranteed to be within certain approximation bounds (Section~\ref{sec:quality}). 
\bt may suffer from \emph{false infeasibility}, which happens when 
the algorithm reports a feasible query to be infeasible. The probability of false infeasibility is, however, low and bounded (Section~\ref{sec:false-negatives}).

\begin{algorithm}[t]
\caption{Scalable Package Query Evaluation}\label{algo:evaluate}
{\small
    \begin{algorithmic}[1]
        \Procedure{\bt}{$\redspace$:Partitioning, $\query$:Package Query}
            \State $\pack_\unsolved \gets$ \Call{\initialalgo}{$\redspace$, $\query$}            
            \If{\textbf{failure}}
                \State \Return \textbf{infeasible}
            \Else
                \State $\pack \gets$ \Call{\augmentalgo}{$\pack_{\unsolved}$, $\redspace$, $\query$}
                \If{\textbf{failure}}
                    \State \Return \textbf{infeasible}
                \Else
                    \State \Return $\pack$
                \EndIf
            \EndIf
        \EndProcedure
    \end{algorithmic}}
\end{algorithm}

In the subsequent discussion, we use ${\sf R}$ to denote the input relation of $n$ tuples, $\tuple_i \in {\sf R}$, $1 \leq i \leq n$. ${\sf R}$ is partitioned into $m$ groups $G_1, \dots, G_m$. Each group $G_j$, $1 \le j \le m$, has a representative tuple $\repr_j$, which may not always appear in ${\sf R}$. We denote the partitioned space with ${\redspace = \{ (G_j, \repr_j) \mid 1 \le j \le m \}}$. 
We refer to packages that contain some representative tuples as
\emph{sketch packages} and packages with only original tuples as
\emph{complete packages} (or simply \emph{packages}).
We denote a complete package with $\pack$ and a sketch package with $\pack_\unsolved$, where $\unsolved \subseteq \redspace$ is the set of groups that are yet to be refined to transform $\pack_\unsolved$ into a complete answer package $\pack$.

\subsection{Offline Partitioning} \label{sec:index}

\bt relies on a\added{n offline} partitioning of the input relation \ssf{R} into groups
of similar tuples. \removed{This p}\added{P}artitioning is \removed{performed offline,} based on a
set of $k$ \added{numerical} \emph{partitioning attributes}, $\mathcal{A}$, from the input relation \ssf{R}\removed{. Partitioning}\added{, and} uses two parameters: a \emph{size} threshold and a \emph{radius} limit.

\begin{definition}[Size Threshold, $\thres$]
The \emph{size threshold} $\thres$, $1 \leq \thres \leq n$, restricts the size of each partitioning group $G_j$, ${1 \le j \le m}$,
to a maximum of $\thres$ original tuples, i.e., $|G_j| \le \thres$.
\end{definition}

\begin{definition}[Radius limit, $\omega$] \label{def:radius}
The \emph{radius} $r_j \ge 0$ of a group $G_j$ is the greatest absolute distance between the representative tuple of $G_j$, $\repr_j$, and every original tuple of the group, across all partitioning attributes:
\begin{equation*}
r_j = \max_{\forall \tuple_j \in G_j} \max_{\forall {\sf attr} \in \mathcal{A}} \{ | \repr_j.{\sf attr} - \tuple_j.{\sf attr} | \}
\end{equation*}
The \emph{radius limit} $\omega$, $\omega \ge 0$, requires that for every partitioning group $G_j$, ${1 \le j \le m}$, $r_j \leq \omega$. 
\end{definition}

The size threshold, $\thres$, affects the number of clusters, $m$, as smaller clusters (lower $\thres$) implies more of them (larger $m$), especially on skewed datasets. As we discuss later (\Cref{sec:bt}), for best response time of \bt, $\thres$ should be set so that both $m$ and $\thres$ are small. Our experiments show that a proper setting can yield to an order of magnitude improvement in query response time (\Cref{sec:thres}). 
The radius limit, $\omega$, should be set according to a desired approximation guarantee (\Cref{sec:quality}).
\added{
Note that the same partitioning can be used to support a multitude of queries over the same dataset.
In our experiments, we show that a single partitioning performs consistently well across different queries.
}

\paratitle{Partitioning method.}
Different methods can be used for partitioning. Our implementation is
based on $k$-dimensional \emph{quad-tree indexing}~\cite{finkel1974quad}. The method recursively partitions a relation into groups until all the groups satisfy the size threshold and meet the radius limit.
First, relation \ssf{R} is augmented with an extra group ID column \ssf{gid}, such that $\tuple_i.{\sf gid} = j$ iff tuple $\tuple_i$ is assigned to group $G_j$. The procedure initially creates a single group $G_1$ that includes all the original tuples from relation \ssf{R}, by initializing $\tuple_i.{\sf gid} = 1$ for all tuples.
Then, it recursively proceeds as follows:
\begin{itemize}[topsep=2px,partopsep=1ex,parsep=1ex,itemsep=0px,leftmargin=3mm]

    \item The sizes and radii of the current groups are computed via a query that groups tuples by their \ssf{gid} value. 

    \item 

    The same group-by query also computes the \emph{centroid} tuple of each group $G_j$. The centroid is computed by averaging the tuples in group $G_j$ on each of the partitioning attributes $\mathcal{A}$. 
    
    \item If group $G_j$ has more tuples than the size threshold, or a radius larger than the radius limit, the tuples in group $G_j$ are partitioned into $2^{k}$ subgroups ($k = |\mathcal{A}|$). The group's centroid is used as the pivot point to generate sub-quadrants: tuples that reside in the same sub-quadrant are grouped together.
    
\end{itemize}
Our method recursively executes two \sql queries on each subgroup that violates the size or the radius condition.
In the last iteration, the last group-by query computes the centroids for each group. 
These are the representative tuples, $\repr_j$, $1 \le j \le m$, and are stored in a new \emph{representative relation} \ssf{\tilde{R}}$({\sf gid}, \text{\ssf{attr}}_1, \dots, \text{\ssf{attr}}_k)$.

\paratitle{Alternative partitioning approaches.} 
We experimented with different clustering algorithms, such as \emph{$k$-means}~\cite{hartigan1979algorithm}, \emph{hierarchical clustering}~\cite{kaufman2009finding} and DBSCAN~\cite{Ester96adensity-based}, using off-the-shelf libraries such as \emph{Scikit-learn}~\cite{scikit-learn}. Existing clustering algorithms present various problems: 
First, they tend to vary substantially in the properties of the generated clusters.
In particular, none of the existing clustering techniques can natively generate clusters that satisfy the size threshold $\thres$ and radius limit $\omega$. In fact, most of the clustering algorithms take as input the \emph{number of clusters} to generate, without offering any means to restrict the size of each cluster nor their radius.
Second, existing implementations only support in-memory cluster computation, and DBMS-oriented implementations usually need complex and inefficient queries.
On the other hand, space partitioning techniques from multi-dimensional indexing, such as \emph{$k$-d trees}~\cite{bentley1975multidimensional} and \emph{quad trees}~\cite{finkel1974quad}, can be more easily adapted to satisfy the size and radius conditions, and to work within the database: our partitioning method works directly on the input table via simple \sql queries.

\paratitle{One-time cost.}
Partitioning is an expensive procedure. To avoid paying its cost at query time, the dataset is partitioned in advance and used to answer a workload of package queries. 
For a known workload, our experiments show that partitioning the dataset on the union of all query attributes provides the best performance in terms of query evaluation time and approximation error for the computed answer package (Section~\ref{sec:exp:cover}).
We also demonstrate that our query evaluation approach is robust to a
wide range of partition sizes, and to imperfect partitions that cover
more or fewer attributes than those used in a particular query.
\added{
This means that, even without a known workload, a partitioning performed on all of the data attributes still provides good performance.}

The radius limit is necessary for the theoretical guarantee of the
approximation bounds (Section~\ref{sec:quality}). However, we show
empirically that partitioning satisfying the size threshold alone
produces satisfactory answers while reducing the offline partitioning
cost: 
Meeting a size threshold requires fewer partitioning iterations
than meeting a radius limit especially if the dataset is sparse across
the attribute domains (Section~\ref{sec:experiments}).

\added{
\paratitle{Dynamic partitioning.}
In our implementation of \bt, the choice of partitioning is static.
Our technique also works with a dynamic approach to partitioning: by
maintaining the entire hierarchical structure of the quad-tree index,
one can traverse the index at query time to generate the coarsest
partitioning that satisfies the required radius condition. However,
our empirical results show that this approach incurs unnecessary
overhead, as static partitioning already performs extremely well in
practice (Section~\ref{sec:experiments}).
}

\subsection{Query Evaluation with \large{\bt}} \label{sec:bt}

During query evaluation, \bt first \emph{sketches} a
package solution using the representative tuples (\textsc{\initialalgo}), 
and then it \emph{refines} it by replacing representative
tuples with original tuples (\textsc{\btalgo}).
We describe these steps using the example query $\query$ from
\Cref{sec:paql}.

\subsubsection{\fontsize{1em}{1em}\selectfont {\sc \initialalgo}} \label{sec:initial}

Using the representative relation \ssf{\tilde{R}}
(Section~\ref{sec:index}), the {\sc \initialalgo} procedure constructs
and evaluates a \emph{sketch query}, $\initquery$. The result is an
initial sketch package, $\pack_\unsolved$, containing
representative tuples that satisfy the same constraints as the
original query $\query$.

{\small
\begin{tabbing}
\hspace*{7.5mm}\=\hspace*{1.85cm}\=\hspace*{2.8cm}\= \kill
$\initquery$:
\>\ssf{SELECT}    \>\ssf{PACKAGE}$($\ssf{\tilde{R}}$)$ \ssf{AS} $\pack_\unsolved$ \\
\>\ssf{FROM}\>\ssf{\tilde{R}} \\
\>\ssf{WHERE}   \>\ssf{\tilde{R}.gluten} $=$ `\ssf{free}' \\
\>\ssf{SUCH} \ssf{THAT} \\
\>~~ \ssf{COUNT}$(\pack_\unsolved.\ast) = 3$ \ssf{AND} \\
\>~~ \ssf{SUM}$(\pack_\unsolved.$\ssf{kcal}$)$ \ssf{BETWEEN} $2.0$ \ssf{AND} $2.5$ \ssf{AND}\\
\>~~ $($\bsf{SELECT} \bsf{COUNT}$(\ast)$ \bsf{FROM} $\mathbf{\pack_\unsolved}$ \bsf{WHERE} \bsf{gid} $\mathbf{= 1) \le |G_1|}$ \\
\>~~  \ssf{AND} \ssf{\dots} \\
\>~~ $($\bsf{SELECT} \bsf{COUNT}$(\ast)$ \bsf{FROM} $\mathbf{\pack_\unsolved}$ \bsf{WHERE} \bsf{gid} $\mathbf{= m) \le |G_m|}$ \\
\>\ssf{MINIMIZE}    \>\ssf{SUM}$(\pack_\unsolved.$\ssf{saturated\_fat}$)$
\end{tabbing}}
The new global constraints,
highlighted in bold, 
ensure that every representative tuple 
does not appear in $\pack_\unsolved$ more times than the size of its group, $G_j$. This accounts for the repetition constraint \ssf{REPEAT} $0$ in the original query. Generalizing, with \ssf{REPEAT} $\mathcal{K}$, each $\repr_j$ can be repeated up to $|G_j|(1+\mathcal{K})$ times. 
These constraints are simply omitted from $\initquery$ if the original query
does not contain a repetition constraint.

Since the representative relation \ssf{\tilde{R}} contains exactly $m$ representative tuples, the \ilp problem corresponding to this query has only $m$ variables. This is typically small enough for the black box \ilp solver to manage directly, and thus we can solve this package query using the \opt method (\Cref{sec:direct}). 
If $m$ is too large, we can solve this query \emph{recursively} with \bt: the set of $m$ representatives is further partitioned into smaller groups until the subproblems reach a size that can be efficiently solved directly.

The {\sc \initialalgo} procedure \emph{fails} if the sketch query 
$\initquery$ is infeasible, in which case \bt reports the
original query $\query$ as infeasible (Algorithm~\ref{algo:evaluate}).
This may constitute \emph{false infeasibility}, if $\query$ is
actually feasible.
In \Cref{sec:false-negatives}, we show that the probability of
false infeasibility is low and bounded, and we present simple methods
to avoid this outcome.

\subsubsection{\fontsize{1em}{1em}\selectfont {\sc \btalgo}} \label{sec:refine}

Using the sketched solution over the representative tuples, the
\textsc{\btalgo} procedure iteratively replaces the representative
tuples with tuples from the original relation \ssf{R}, until no more
representatives are present in the package. The algorithm
\emph{refines} the sketch package $\pack_\unsolved$, one group at a
time: For a group $G_j$ with representative
$\repr_j\in\pack_\unsolved$, the algorithm derives package
${\bar{\pack}_j}$ from $\pack_\unsolved$ by eliminating all instances
of $\repr_j$; it then seeks to replace the eliminated representatives
with actual tuples, by issuing a \emph{refine query},
$\augmquery$, on group $G_j$:
{\small
\begin{tabbing}
\hspace*{9mm}\=\hspace*{2cm}\=\hspace*{3cm}\= \kill
$\augmquery$:
\>\ssf{SELECT}    \>\ssf{PACKAGE}$(G_j)$ \ssf{AS} $\pack_j$ \\
\>\ssf{FROM}      \>$G_j$ \ssf{REPEAT} $0$ \\
\>\ssf{WHERE}     \>$G_j.$\ssf{gluten} $=$ `\ssf{free}' \\
\> \ssf{SUCH} \ssf{THAT} \\
\>~~ \ssf{COUNT}$(\pack_j.\ast) + $\bsf{COUNT}$(\bar{\pack}_j.\ast) = 3 $ \ssf{AND} \\
\>~~ \ssf{SUM}$(\pack_j.$\ssf{kcal}$) + $\bsf{SUM}$(\bar{\pack}_j.$\ssf{kcal}$)$ \ssf{BETWEEN} $2.0$ \ssf{AND} $2.5$ \\
\>\ssf{MINIMIZE}  \>\ssf{SUM}$(\pack_j.$\ssf{saturated\_fat}$)$ 
\end{tabbing}}
The query derives a set of tuples $\pack_j$, as a replacement for the
occurrences of the representatives of $G_j$ in $\pack_\unsolved$. The
global constraints in $\augmquery$ ensure that the combination of
tuples in $\pack_j$ and $\bar\pack_j$ satisfy the original query
$\query$. Thus, this step produces the new \emph{refined sketch package}
${\pack'_{\unsolved'} = \bar{\pack}_j \cup \pack_j}$, where
${\unsolved' = \unsolved \setminus \{(G_j, \repr_j)\}}$.

Since $G_j$ has at most $\thres$ tuples, the \ilp problem
corresponding to $\augmquery$ has at most $\thres$ variables. This is
typically small enough for the black box \ilp solver to solve
directly, and thus we can solve this package query using the \opt
method (\Cref{sec:direct}).
Similarly to the sketch query, if $\thres$ is too large, we can solve
this query recursively with \bt: the tuples in group $G_j$ are further
partitioned into smaller groups until the subproblems reach a size
that can be efficiently solved directly.

Ideally, the \textsc{\btalgo} step will only process each group with
representatives in the initial sketch package once. However, the order
of refinement matters as each refinement step is greedy: it selects
tuples to replace the representatives of a single group, without
considering the effects of this choice on other groups. As a result, a
particular refinement step may render the query infeasible (no tuples
from the remaining groups can satisfy the constraints). When this
occurs, \textsc{\btalgo} employs a \emph{greedy backtracking} strategy
that reconsiders groups in a different order.

\paratitle{Greedy backtracking.} \textsc{\btalgo} activates
backtracking when it encounters an infeasible \emph{refine query},
$\augmquery$. Backtracking \emph{greedily prioritizes} the infeasible
groups. This choice is motivated by a simple heuristic: if the
refinement on $G_j$ fails, it is likely due to choices made by
previous refinements; therefore, by prioritizing $G_j$, we reduce the
impact of other groups on the feasibility of $\augmquery$. This
heuristic does not affect the approximation guarantees
(Section~\ref{sec:quality}).

\begin{algorithm}[t]
\caption{Greedy Backtracking Refinement}\label{algo:bt}
{\small
    \begin{algorithmic}[1]
        \Procedure{\btalgo}{$\pack_{\unsolved}$, $\redspace$, $\query$}       
        \If{$\unsolved = \emptyset$} \Comment{Base case: all groups already refined}
            \State \Return $\pack_\unsolved$ \label{bt-line:base}
        \EndIf

        \State $\failed \gets \emptyset$ \Comment{Failed groups}
        \LineComment{Arrange $\unsolved$ in some initial order (e.g., random)}
        \State $\prioQ \gets priorityQueue(\unsolved)$ \label{bt-line:prioQ}
        
        \While{$\prioQ \neq \emptyset$} \label{algo:bt:line:for}
            \State $(G_j,\repr_j) \gets dequeue(\prioQ)$
            
            \LineComment{Skip groups that have no representative in $\pack_\unsolved$ }
            \If{$\repr_j \notin \pack_\unsolved$}
            	\State {\bf continue} 
            \EndIf
            
            \State $\pack_{\unsolved'}' \gets$ Refine $\pack_\unsolved$ on group $G_j$ 
                 \label{algo:bt:line:augment}
            \If{$\augmquery$ {\bf is infeasible}}
                \If{$\unsolved \neq \redspace$} \Comment{If $\pack_{\unsolved}$ is not the initial package}
                    \LineComment{Greedily \emph{backtrack} with non-refinable group}
                    \State $\failed \gets \failed \cup \{(G_j, \repr_j)\}$
                    \State \Return \textbf{failure}($\failed$) \label{algo:bt:line:back}
                \EndIf
            \Else
                \LineComment{Greedily \emph{recurse} with refinable group}
                \State $\pack \gets \Call{\btalgo}{\pack_{\unsolved'}', \redspace, \query}$ \label{algo:bt:line:recursive}
                \If{{\bf failure}($\failed'$)}
                    \State $\failed \gets \failed \cup \failed'$
                    \LineComment{Greedily \emph{prioritize} non-refinable groups}
                    \State $prioritize(\prioQ, \failed)$ \label{bt-line:prioritize}
                \Else
                    \State \Return $\pack$
                \EndIf
            \EndIf
        \EndWhile
        
        \LineComment{None of the groups in $\unsolved$ can be refined (invariant: $\failed = \unsolved$)}
        \State \Return \textbf{failure}($\failed$) \label{algo:bt:line:lastbt}
        
        \EndProcedure
    \end{algorithmic}}
\end{algorithm}

Algorithm~\ref{algo:bt} details the {\sc \btalgo} procedure. The
algorithm logically traverses a \emph{search tree} (which is only
constructed as new branches are created and new nodes visited), where
each node corresponds to a unique sketch package $\pack_\unsolved$.
The traversal starts from the \emph{root}, corresponding to the
initial sketch package, where no groups have been refined ($\unsolved
= \redspace$),
and finishes at the first encountered \emph{leaf},
corresponding to a complete package ($\unsolved = \emptyset$). The
algorithm terminates as soon as it encounters a complete package,
which it returns (line~\ref{bt-line:base}). The algorithm assumes a
(initially random) refinement order for all groups in $\unsolved$, and places
them in a priority queue (line~\ref{bt-line:prioQ}).
During refinement, this group order can change by
prioritizing groups with infeasible refinements
(line~\ref{bt-line:prioritize}).

\paratitle{Run time complexity.}
In the best case, all refine queries are feasible and the algorithm never backtracks. In this case, the algorithm makes up to $m$ calls to the \ilp solver to solve problems of size up to $\tau$, one for each refining group. In the worst case, \bt tries every group ordering leading to an exponential number of calls to the \ilp solver. Our experiments show that the best case is the most common and backtracking occurs infrequently.

\subsection{Approximation Guarantees} \label{sec:quality}

\bt provides strong theoretical guarantees. We prove that for a
desired approximation parameter $\epsilon$, we can derive a radius
limit $\omega$ for the offline partitioning that guarantees that \bt
will produce a package with objective value ${(1\pm\epsilon)^6}$-factor close
to the objective value of the solution generated by \opt for the same query.

\begin{restatable}[Approximation Bounds]{theorem}{thapprox} \label{th:approx}
For any feasible package query with a maximization (minimization, resp.) objective and approximation parameter $\epsilon$,  $0 \le \epsilon < 1$ ($\epsilon \ge 0$, resp.),
any database instance, any set of partitioning attributes $\mathcal{A}$, superset of the numerical query attributes, any size threshold $\thres$, and radius limit:
\begin{equation}\label{eq:radius}
    \omega = \min_{
        \substack{1 \le j \le m \\ {\sc attr} \in \mathcal{A}}
        }
        \gamma~|\repr_j \text{\ssf{.attr}}|, 
        \;\text{  where } \gamma = \epsilon
        \;\;\;(\gamma=\tfrac{\epsilon}{1 + \epsilon},\text{ resp.})
\end{equation}
The package produced by \bt (if any) is guaranteed to have objective value  $\ge{(1-\epsilon)^6}OPT$ ($\le{(1+\epsilon)^6}OPT$, resp.), where OPT is the objective value of the \opt solution.
\end{restatable}

{
At a high level, the proof \ifshort \else (Appendix~\ref{app:perf-guarantees}) \fi
includes two steps: In the first step, we prove that the initial sketch package is a ${(1 \pm \epsilon)^3}$-approximation with respect to \opt; In the second step, we prove that the final package produced by \bt is a ${(1 \pm \epsilon)^3}$-approximation with respect to the initial sketch package.
}

\subsection{The Case of False Infeasibility} \label{sec:false-negatives}

For a feasible query $\query$, false negatives, or \emph{false infeasibility}, may happen in two cases: 
(1)~when the sketch query $\initquery$ is infeasible; 
(2)~when greedy backtracking fails (possibly due to suboptimal partitioning). 
In both cases, \bt would (incorrectly) report a feasible package query as infeasible. 
False negatives are, however, extremely rare, 
as the following theorem establishes.

\begin{restatable}[False Infeasibility]{theorem}{thfalseinf} 
\label{th:false-inf}
For any feasible package query, any database instance, any set of partitioning attributes $\mathcal{A}$ that is a superset of the query attributes, any size threshold $\thres$, and any radius limit $\omega$, \bt finds a feasible package with high probability that inversely depends on query selectivity.
\end{restatable}

\ifshort
\else
We include the proof in Appendix~\ref{sec:false-inf-bound}.
\fi
The \emph{selectivity} of a package query denotes the probability of a
random package being infeasible (thus, lower selectivity implies
higher probability of random package being feasible). To prove this
result, we show that if a random package is feasible, then with high
probability the sketch query and all refine queries are feasible.
Thus, lower selectivity implies higher probability that \bt
successfully terminates with a feasible package.

\smallskip
We discuss four potential ways to deal with false infeasibility, which
we plan to explore in future work:

\begin{enumerate}[topsep=0ex,partopsep=1ex,parsep=1ex,itemsep=0ex,leftmargin=0mm,listparindent=0mm,labelsep=0.8mm,itemindent=3mm]

\litem{Hybrid sketch query} 
A simple method to avoid false infeasibility in the {\sc \initialalgo}
step is to merge the sketch query $\initquery$ with one of the
refine queries. This ``hybrid'' sketch query would select original
tuples for one of the groups, and at the same time select
representative tuples for the remaining groups. Groups can be tried in
any order (e.g., in random order), until one of the hybrid sketch
queries is feasible.

\litem{Further partitioning}
In some cases, the centroid of a group may not be a very good
representative (e.g., when the data is skewed). This can sometimes
result in false negatives.
Further partitioning by reducing the size threshold $\thres$ may
eliminate the problem.

\litem{Dropping partitioning attributes}
A more principled approach consists of ``projecting'' the partitioning
onto fewer dimensions by reducing the number of partitioning
attributes. In doing so, some groups merge, increasing the chance that
previously infeasible groups become feasible. The choice of attributes
to remove from the partitioning can be guided by the last infeasible
\ilp problem. Most \ilp solvers provide functionality to identify a
\emph{minimal set of infeasible constraints}: removing any constraint
from the set makes the problem feasible.\footnote{This set is
usually referred to as \emph{irreducible infeasible set} (IIS).} Removing the
attributes that participate in these constraints from the partitioning
can increase the odds of discovering a feasible solution.

\litem{Iterative group merging}
In a brute-force approach, we can merge groups iteratively, until the
sub-queries become feasible. In the worst case, this process reduces
the problem to the original problem (i.e., with no partitioning), and
thus it is guaranteed to find a solution to any feasible query, at the
cost of performance.

\end{enumerate}

\subsection{Discussion}\label{sec:discussion}

\bt is an evaluation strategy for package queries with three important
advantages. First, it scales naturally to very large datasets, by
breaking down the problem into smaller, manageable subproblems, whose
solutions can be iteratively combined to form the final result.
Second, it provides flexible approximations with strong theoretical
guarantees on the quality of the package results.
Third, while our current implementation of \bt employs \ilp solvers to
evaluate the generated subproblems, our algorithm can use any other
black box solution for package queries, even solutions that work
entirely in main memory, and whose efficiency drastically degrades
with larger problem sizes. We plan to explore these alternatives in
our future work.

\paratitle{Parallelizing \bt.}
Its data partitioning and problem division strategies give \bt great
potential for parallelization. However, the proper parallelization
strategy is non-obvious, and is a nontrivial part of future work. A
simple parallelization strategy could perform refinement on several
groups in parallel. However, since refinements make local
decisions, this process is more likely to reach infeasibility,
requiring costly backtracking steps and resulting in wasted
computation. Alternatively, parallelization may focus on the
backtracking process, using additional resources to evaluate different
group orderings in parallel.

\added{
\paratitle{Handling joins.}
In this paper, we focused on package queries over single relations.
Handling joins poses interesting implications, orthogonal to
the evaluation techniques that we presented in this work. In the
presence of joins, the system can simply evaluate and materialize the
join result before applying the package-specific transformations.
However, the materialization of the join result is not always necessary:
\opt generates variables through a single sequential scan
of the join result, and thus the join tuples can be pipelined into the
\ilp generation without being materialized. However, not materializing
the join results means that some of the join tuples will need to be
recomputed to populate the solution package. Therefore, there is a
space-time tradeoff in the consideration of materializing the join.
Further, this tradeoff can be improved with hybrid, system-level
solutions, such as storing the record IDs of joining tuples to enable
faster access during package generation. These considerations are
well-beyond our current scope, and are orthogonal to the techniques
that we present in this work.
}

%% file: 7_experiments.tex
\section{Experimental Evaluation} \label{sec:experiments}

In this section, we present an extensive experimental evaluation of
our techniques for package query execution, both on real-world and on
benchmark data.
Our results show the following properties of our methods:
(1)~\bt evaluates package queries an order of magnitude faster than \opt;
(2)~\bt scales up to sizes that \opt cannot handle directly; 
(3)~\bt produces packages of high quality (similar objective value as the packages returned by \opt); 
(4)~the performance of \bt is robust to partitioning on different sets of attributes as long as a query's attributes are mostly covered. This makes offline partitioning effective for entire query workloads.

\subsection{Experimental Setup}

\paratitle{Software.}
We implemented our package evaluation system
as a layer on top of a traditional relational DBMS. The data itself resides in the database, and the system interacts with the DBMS via \sql when it needs to perform operations on the data. 
We use PostgreSQL v9.3.9 for our experiments.
The core components of our evaluation module are implemented in Python 2.7. The \paql parser is generated in C++ from a context-free grammar, using GNU Bison~\cite{bison}. 
We represent a package in the relational model as a standard relation with schema equivalent to the schema of the input relation.
A package is materialized into the DBMS only when necessary (for example, to compute its objective value).

We employ IBM's \cplex~\cite{cplex} v12.6.1 as our black-box \ilp solver. When the algorithm needs to solve an \ilp problem, the corresponding data is retrieved from the DBMS and passed to \cplex using tuple iterator APIs to avoid having more than one copy of the same data stored in main memory at any time. 
We used the same settings for all solver executions:
we set its working memory 
to 512MB; 
we instructed \cplex to store exceeding data used during the solve procedure on disk in a compressed format, rather than using the operating system's virtual memory, which, as per the documentation, may degrade the solver's performance; 
we instructed \cplex to emphasize optimality versus feasibility to dampen the effect of internal heuristics that the solver may employ on particularly hard problems; 
we enabled \cplex's memory emphasis parameter, which instructs the solver to conserve memory where possible;
we set a solving time limit of one hour;
we also made sure that the operating system would kill the solver process whenever it uses the entire available main memory.

\paratitle{Environment.}
We run all experiments on a ProLiant DL160 G6 server equipped with two twelve-core Intel Xeon X5650 CPUs at 2.66GHz each, with 15GB or RAM, with a single 7200 RPM 500GB hard drive, running CentOS release 6.5.

\paratitle{Datasets and queries.}
We demonstrate the performance of our query evaluation methods using both real-world and benchmark data. 
The real-world dataset consists of approximately 5.5 million tuples extracted from the Galaxy view of the Sloan Digital Sky Survey (\sdss)~\cite{sdss}, data release 12. 
For the benchmark datasets we used \tpch~\cite{tpch}, 
with table sizes up to 11.8 million tuples.

\begin{figure}[t]
\centering
\scriptsize
\begin{tabularx}{\columnwidth}{|c|@{} *7{>{\centering\arraybackslash}X}@{}|}
\hline
\tpch query      & Q1 & Q2 & Q3 & Q4 & Q5   & Q6    & Q7 \\
\hline
Max \# of tuples & 6M & 6M & 6M & 6M & 240k & 11.8M & 6M \\
\hline
\end{tabularx}
\caption{Size of the tables used in the \tpch benchmark.}
\label{tb:tpch-queries}
\end{figure}

\begin{figure}[t]
\centering
\scriptsize
\begin{tabularx}{\columnwidth}{|c|@{} *3{>{\centering\arraybackslash}X}@{}|}
\hline
Dataset      & Dataset size & Size threshold $\thres$ &  Partitioning time \\
\hline
\galaxy      & 5.5M tuples        & 550k tuples    & 348 sec.\\
\tpch        & 17.5M tuples       & 1.8M tuples    & 1672 sec.\\
\hline
\end{tabularx}
\caption{Partitioning time for the two datasets, using the workload attributes and with no radius condition.}
\label{tb:partitioning}
\end{figure}

For each of the two datasets, we constructed a set of seven package queries,\removed{.
We specified the package queries} 
by adapting existing \sql queries originally designed for each of the two datasets.
For the \galaxy dataset, we adapted some of the real-world sample \sql queries available directly from the \sdss website.\footnote{\resizebox{\columnwidth}{!}{\url{http:
//cas.sdss.org/dr12/en/help/docs/realquery.aspx}}}
For the \tpch dataset, we adapted seven of the \sql query templates provided with the benchmark that contained enough numerical attributes.
We performed query specification manually, by transforming \sql aggregates into global predicates or objective criteria whenever possible, selection predicates into global predicates, and by adding cardinality bounds. 
We did not include any base predicates in our package queries because they can always be pre-processed by running a standard \sql query over the input dataset (\Cref{sec:ilp}), and thus eliminated beforehand.
For the \galaxy queries, we synthesized the global constraint bounds by multiplying the original selection constraint bound by the expected size of the feasible packages. 
For the \tpch queries, we generated global constraint bounds uniformly at random by multiplying random values in the value range of a specific attribute by the expected size of the feasible packages.
The original \tpch \sql queries involve attributes across different relations and compute various group-by aggregates. In order to transform these queries into single-relation package queries, we processed the original \tpch tables to produce a single pre-joined table, obtained with full outer joins, containing all attributes needed by all the \tpch package queries in our benchmark. This table contained approximately 17.5 million tuples.
For each \tpch package query, we then extracted the subset of tuples having non-NULL values on all the query attributes.
The size of each resulting table is reported in \Cref{tb:tpch-queries}.

\begin{figure*}[ht]
    \centering
    \begin{subfigure}[b]{1.0\textwidth}
        \centering
        \includegraphics[scale=0.25,trim=150 0 0 0,frame=.0mm]
        {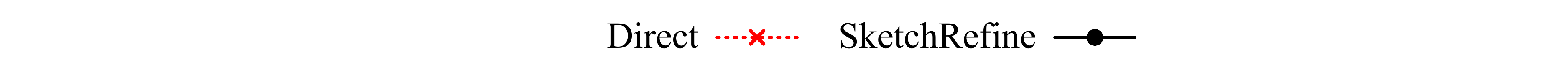}
    \end{subfigure}
    \begin{subfigure}[b]{0pt}
        \centering
        \includegraphics[scale=0.25,trim=50 0 0 0,left]
        {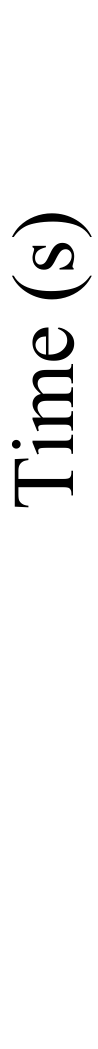}
    \end{subfigure}
    \begin{subfigure}[b]{.138\textwidth}
        \centering
        \caption[short for lof]{Q1}
        \includegraphics[scale=0.25,trim=30 0 0 20,center]
        {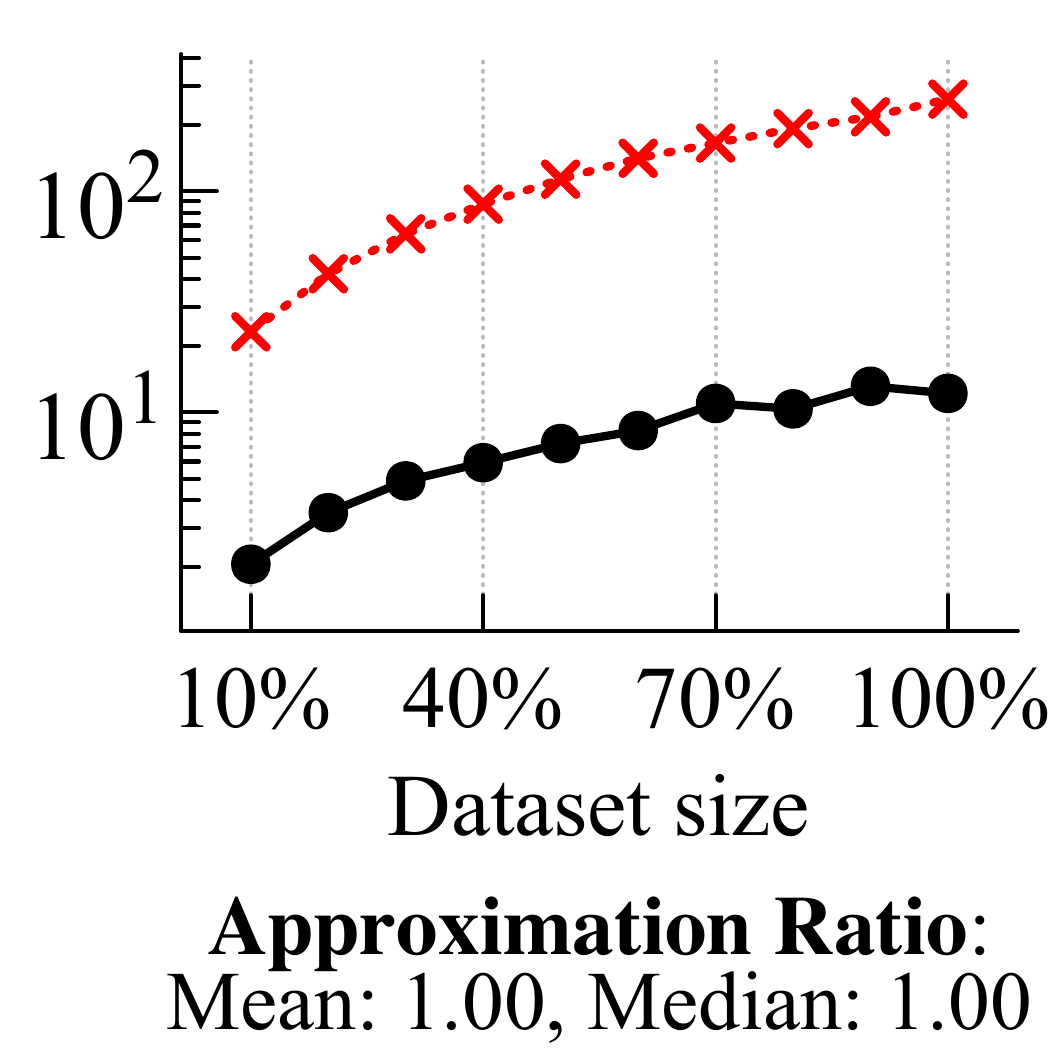}
        \label{fig:scalability_galaxy:q1}
    \end{subfigure}
    \begin{subfigure}[b]{.138\textwidth}
        \centering
        \caption[short for lof]{Q2}
        \includegraphics[scale=0.25,trim=30 0 0 20,center]
        {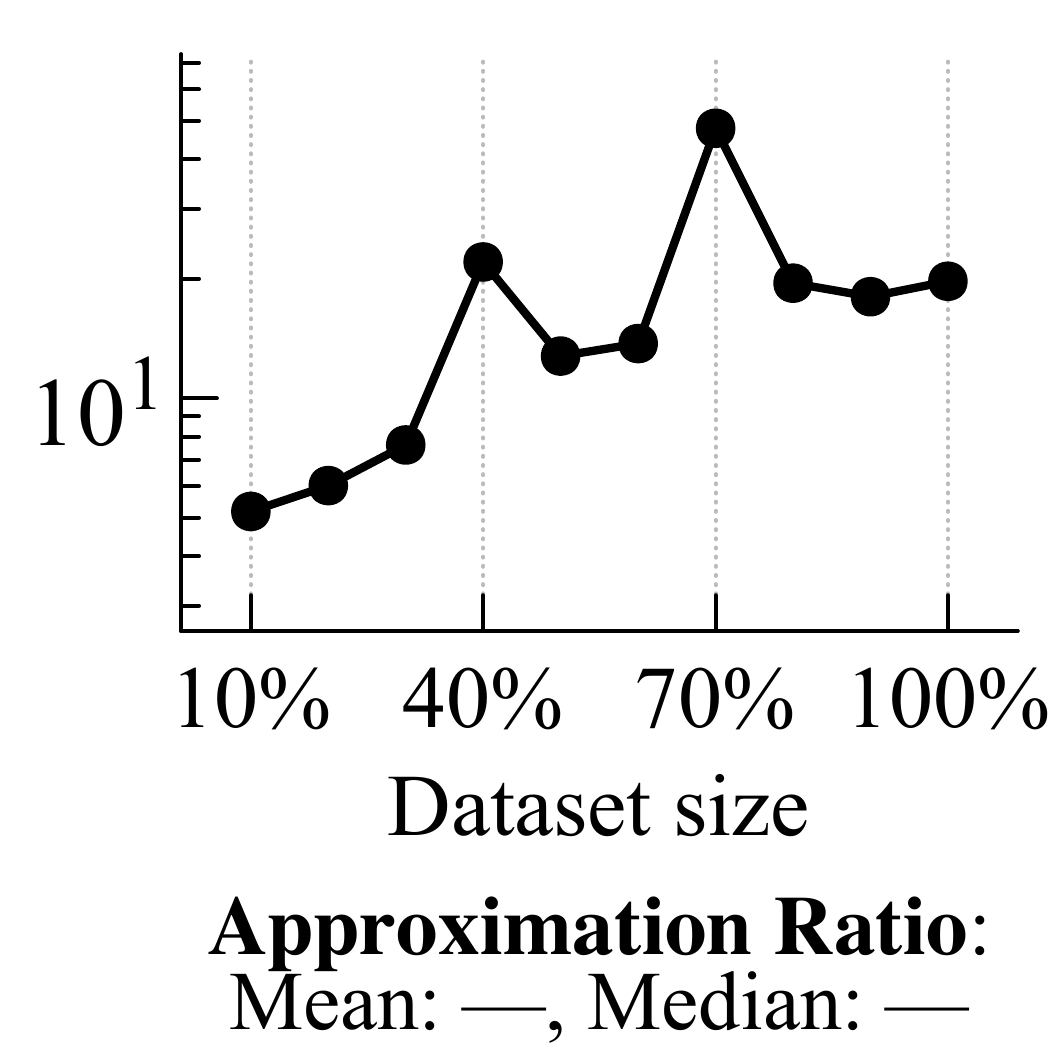}
        \label{fig:scalability_galaxy:q2}
    \end{subfigure}
    \begin{subfigure}[b]{.138\textwidth}
        \centering
        \caption[short for lof]{Q3}
        \includegraphics[scale=0.25,trim=30 0 0 20,center]
        {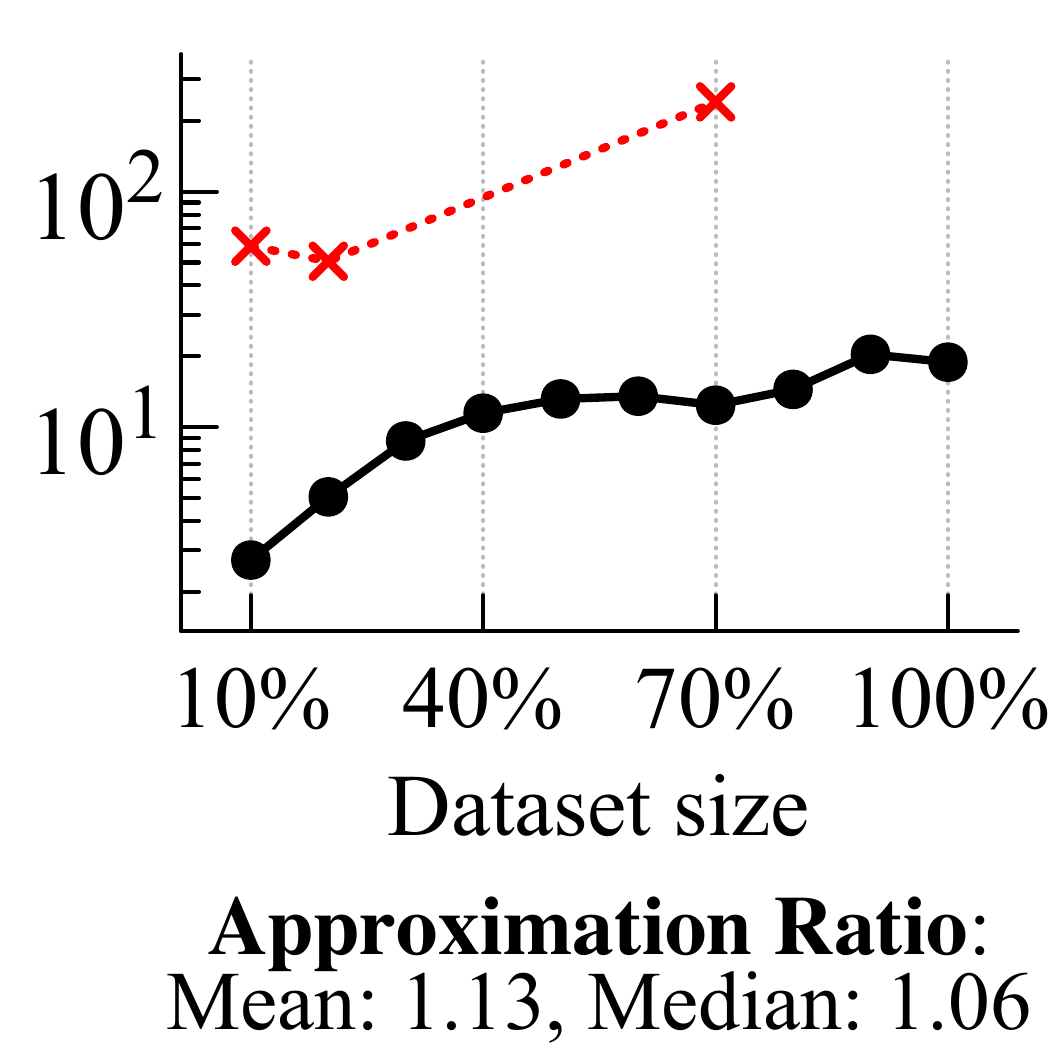}
        \label{fig:scalability_galaxy:q3}
    \end{subfigure}
    \begin{subfigure}[b]{.138\textwidth}
        \centering
        \caption[short for lof]{Q4}
        \includegraphics[scale=0.25,trim=30 0 0 20,center]
        {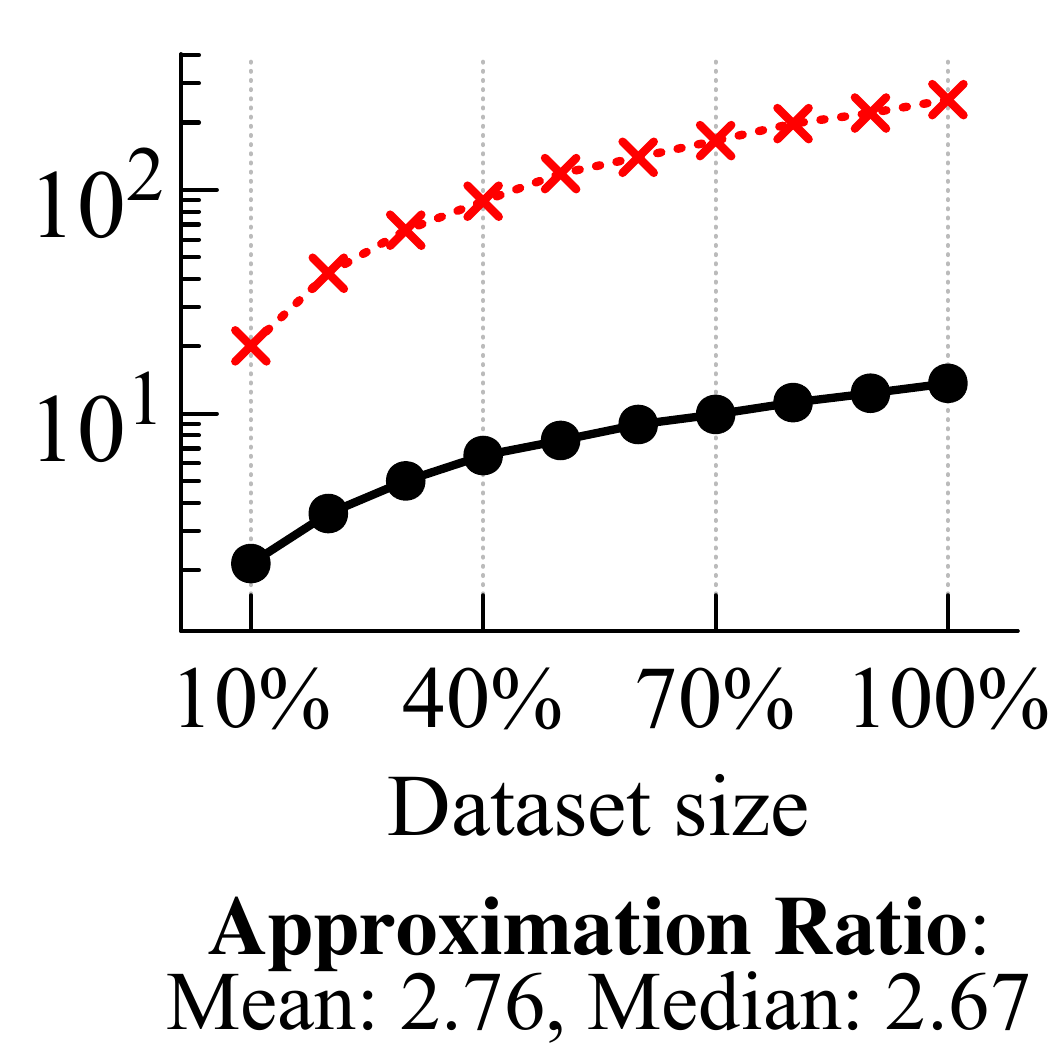}
        \label{fig:scalability_galaxy:q4}
    \end{subfigure}
    \begin{subfigure}[b]{.138\textwidth}
        \centering
        \caption[short for lof]{Q5}
        \includegraphics[scale=0.25,trim=30 0 0 20,center]
        {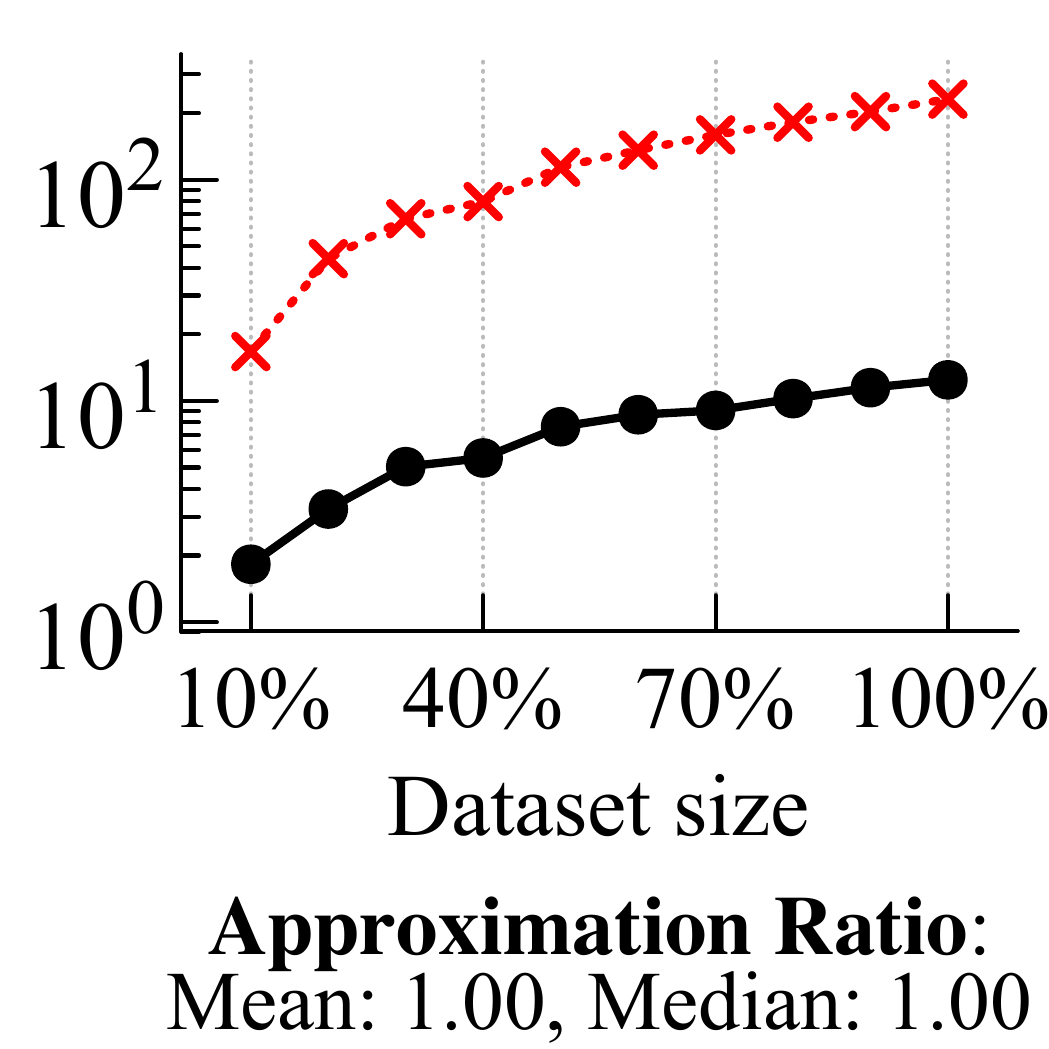}
        \label{fig:scalability_galaxy:q5}
    \end{subfigure}
    \begin{subfigure}[b]{.138\textwidth}
        \centering
        \caption[short for lof]{Q6}
        \includegraphics[scale=0.25,trim=30 0 0 20,center]
        {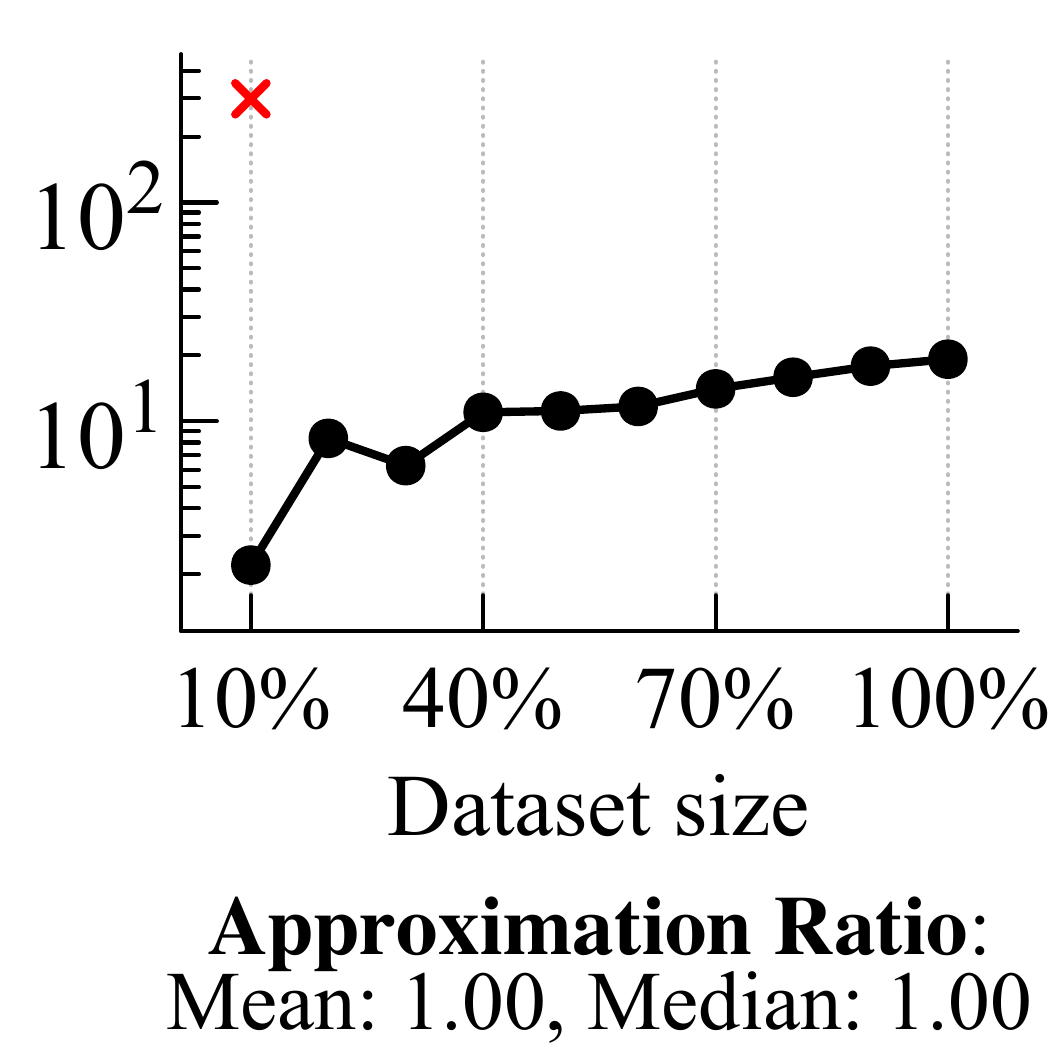}
        \label{fig:scalability_galaxy:q6}
    \end{subfigure}
    \begin{subfigure}[b]{.138\textwidth}
        \centering
        \caption[short for lof]{Q7}
        \includegraphics[scale=0.25,trim=30 0 0 20,center]
        {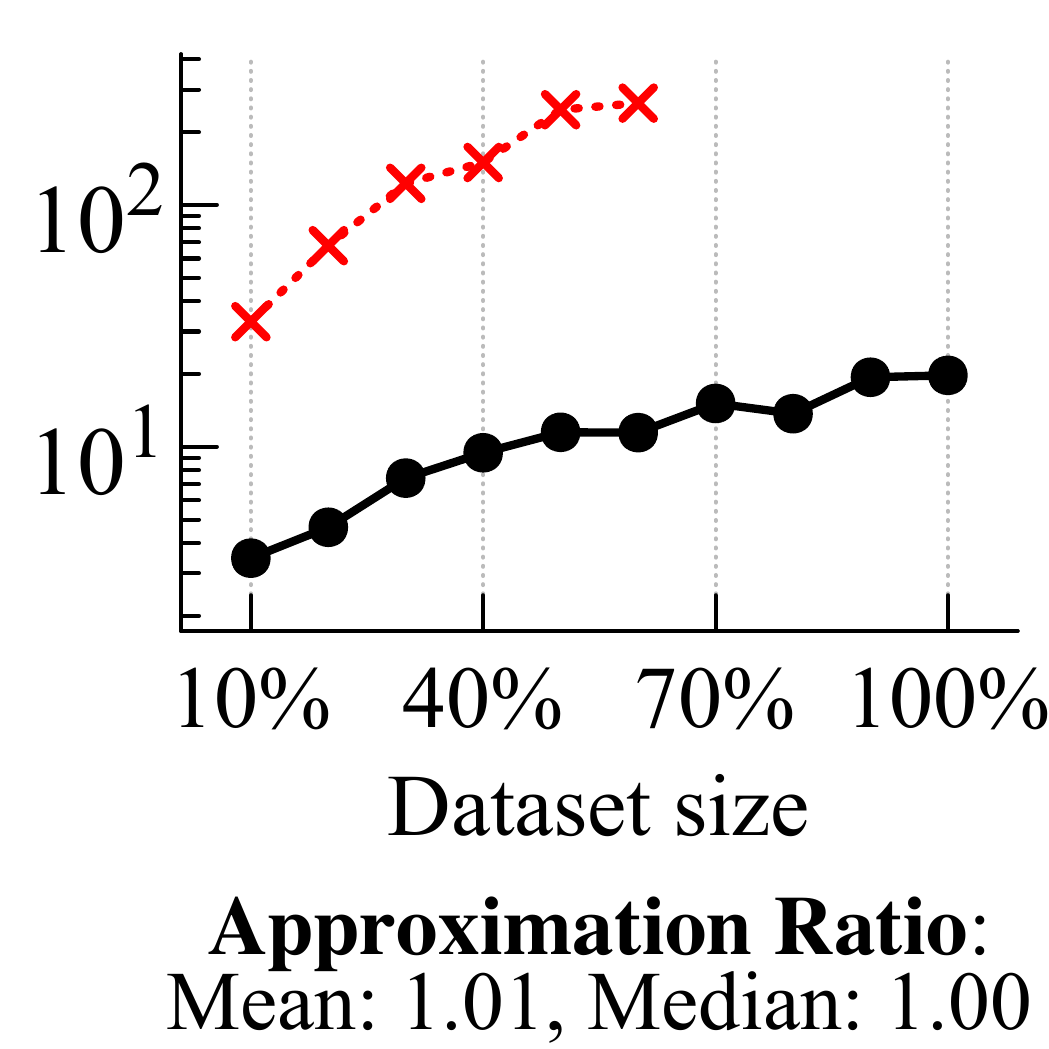}
        \label{fig:scalability_galaxy:q7}
    \end{subfigure}
\vspace{-3mm}
    \caption{
    Scalability on the \galaxy benchmark. 
    \bt uses an offline partitioning computed on the full dataset, using the workload attributes, $\thres =$ 10\% of the dataset size, and no radius condition. 
    \opt scales up to millions of tuples in about half of the queries, but it fails on the other half.
    \bt scales up nicely in all cases, and runs about an order of magnitude faster than \opt. Its approximation ratio is always low, even though the partitioning is constructed without radius condition.
    } 
    \label{fig:scalability_galaxy}
\end{figure*}

\begin{figure*}[ht]
    \centering
    \begin{subfigure}[b]{1.0\textwidth}
        \centering
        \includegraphics[scale=0.25,trim=150 0 0 0,frame=.0mm]
        {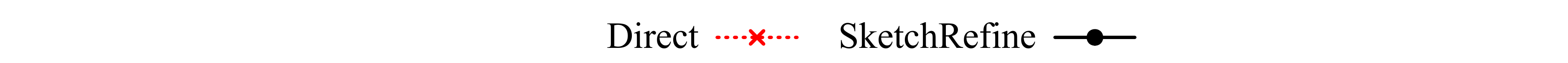}
    \end{subfigure}
    \begin{subfigure}[b]{0pt}
        \centering
        \includegraphics[scale=0.25,trim=50 0 0 0,left]
        {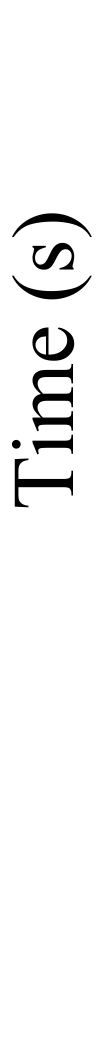}
    \end{subfigure}
    \begin{subfigure}[b]{.138\textwidth}
        \centering
        \caption[short for lof]{Q1}
        \includegraphics[scale=0.25,trim=30 0 0 20,center]
        {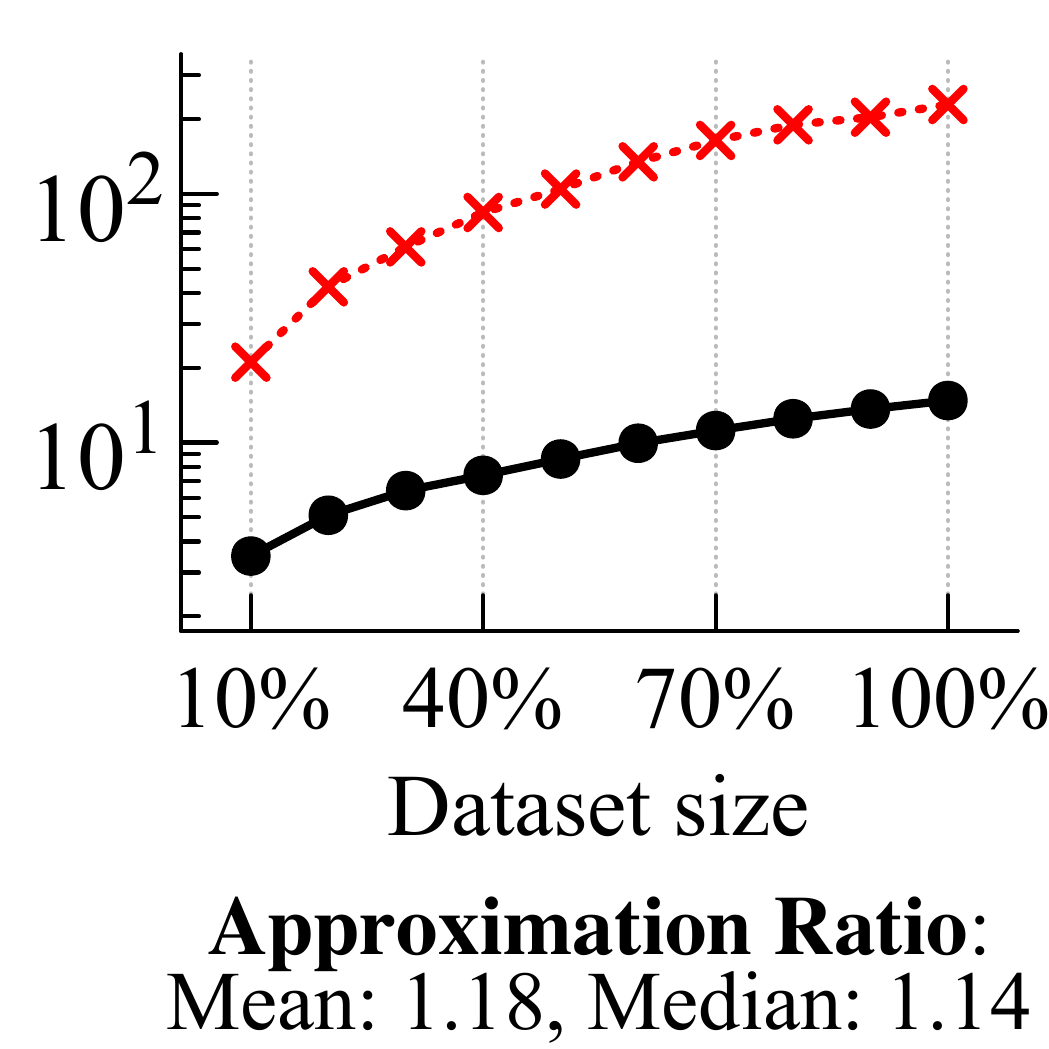}
    \end{subfigure}
    \begin{subfigure}[b]{.138\textwidth}
        \centering
        \caption[short for lof]{Q2}
        \includegraphics[scale=0.25,trim=30 0 0 20,center]
        {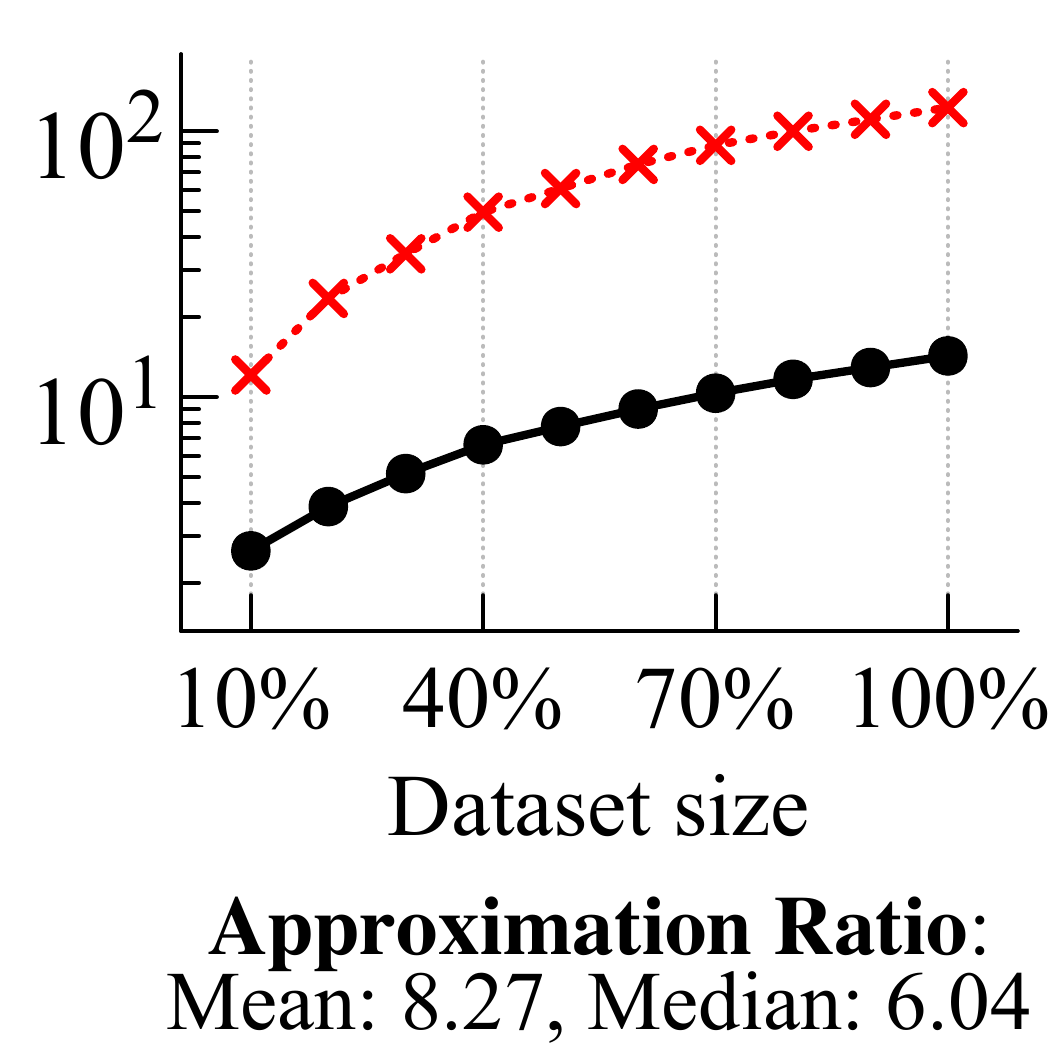}
    \end{subfigure}
    \begin{subfigure}[b]{.138\textwidth}
        \centering
        \caption[short for lof]{Q3}
        \includegraphics[scale=0.25,trim=30 0 0 20,center]
        {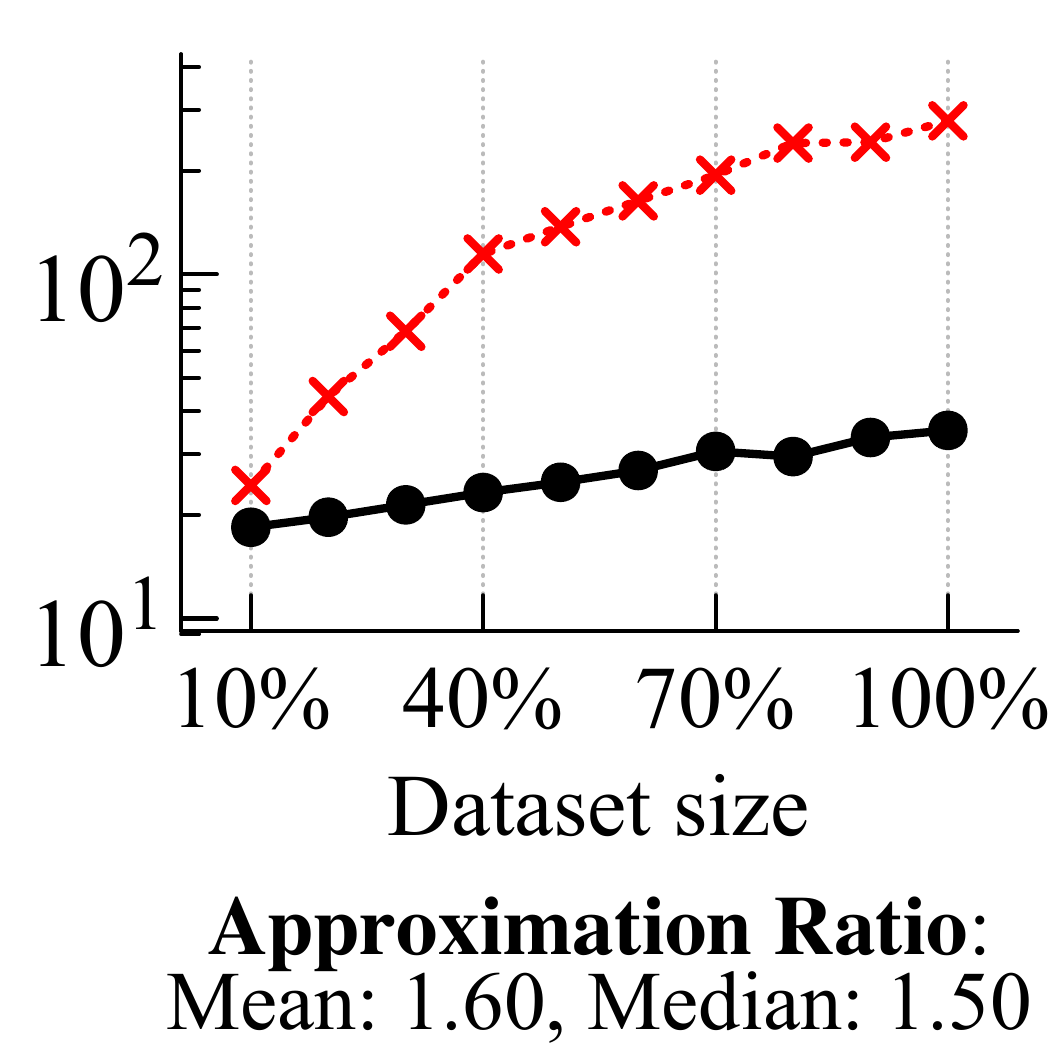}
    \end{subfigure}
    \begin{subfigure}[b]{.138\textwidth}
        \centering
        \caption[short for lof]{Q4}
        \includegraphics[scale=0.25,trim=30 0 0 20,center]
        {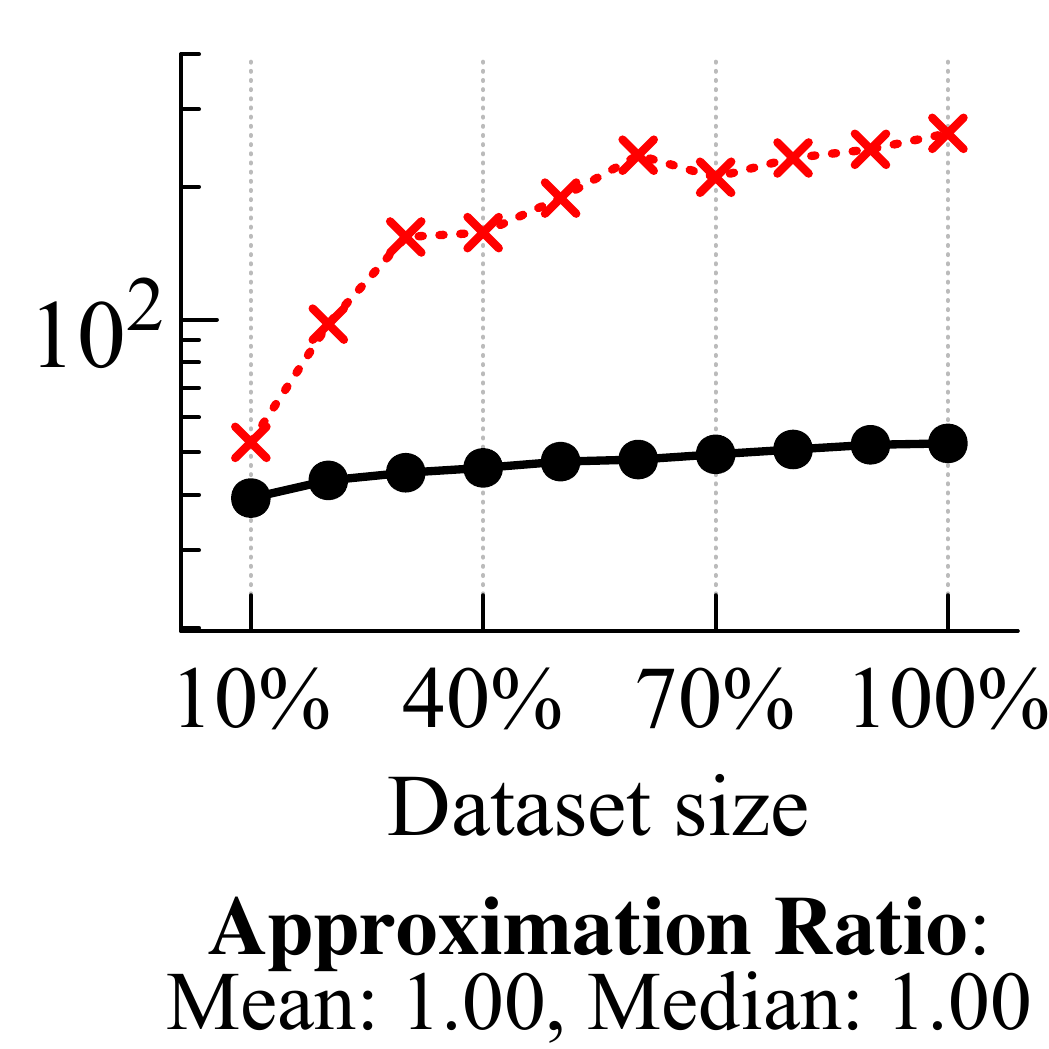}
    \end{subfigure}
    \begin{subfigure}[b]{.138\textwidth}
        \centering
        \caption[short for lof]{Q5}
        \includegraphics[scale=0.25,trim=30 0 0 20,center]
        {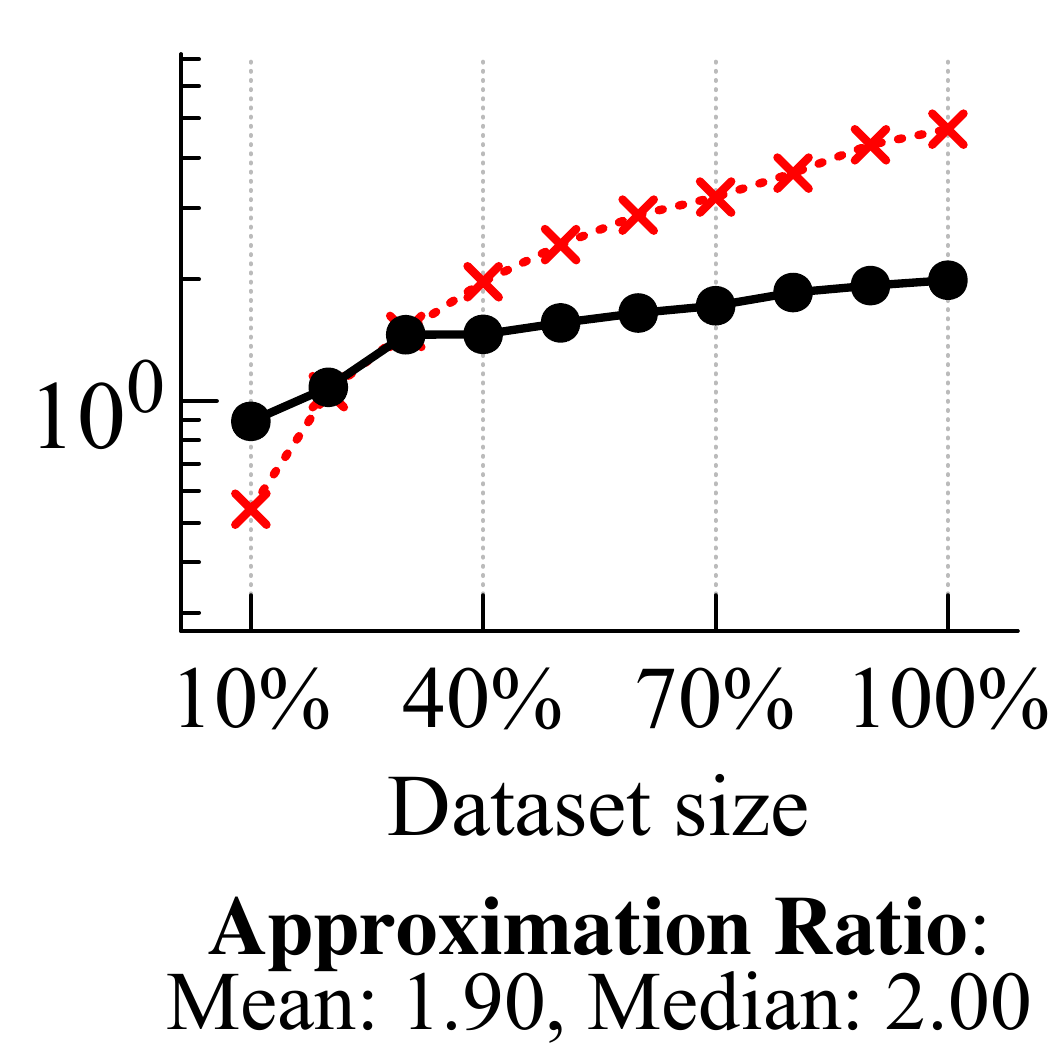}
    \end{subfigure}
    \begin{subfigure}[b]{.138\textwidth}
        \centering
        \caption[short for lof]{Q6}
        \includegraphics[scale=0.25,trim=30 0 0 20,center]
        {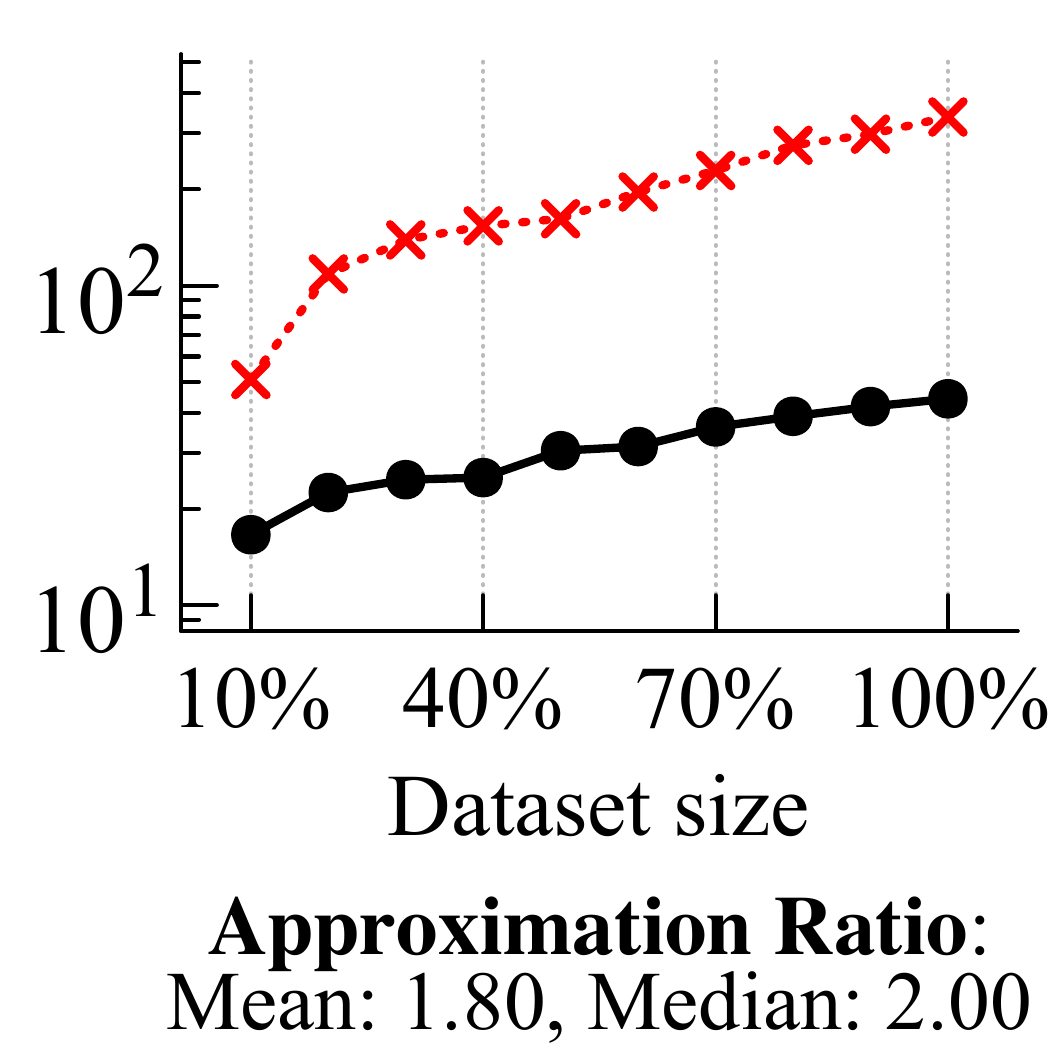}
    \end{subfigure}
    \begin{subfigure}[b]{.138\textwidth}
        \centering
        \caption[short for lof]{Q7}
        \includegraphics[scale=0.25,trim=30 0 0 20,center]
        {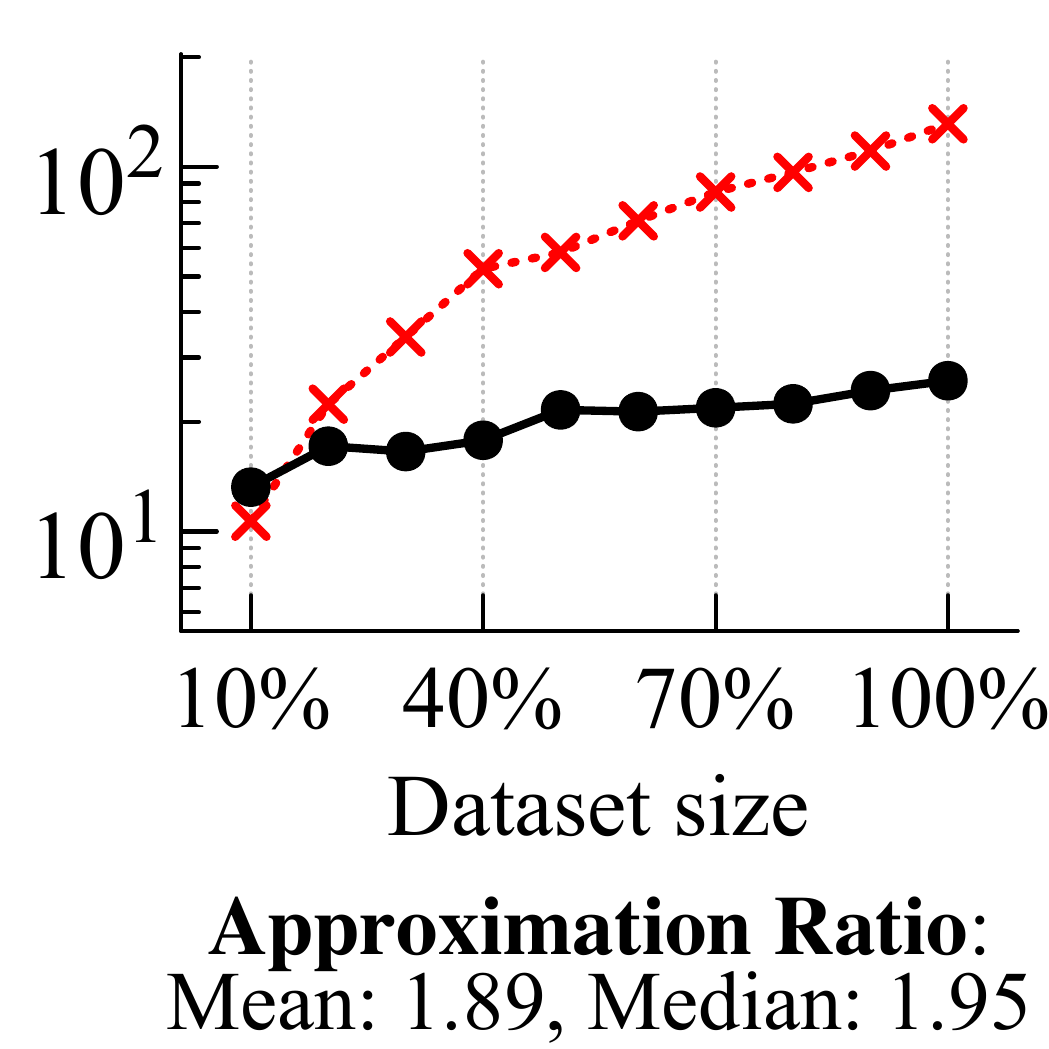}
    \end{subfigure}
    \caption{
    Scalability on the \tpch benchmark. 
    \bt uses an offline partitioning computed on the full dataset, using the workload attributes, $\thres =$ 10\% of the dataset size, and no radius condition. 
    \opt scales up to millions of tuples in all queries.
    The response time of \bt is about an order of magnitude less than \opt, and its approximation ratio is generally very low, even though the partitioning is constructed without radius condition.
    }
    \label{fig:scalability_tpch}
\end{figure*}

\paratitle{Comparisons.}
We compare \opt with \bt.
Both methods use the \ilp formulation (\Cref{sec:ilp}) to transform package queries into \ilp problems: \opt translates and solves the original query; \bt translates and solves the sub-queries (\Cref{sec:evaluation}), and
uses \emph{hybrid sketch query} (\Cref{sec:false-negatives}) as the only strategy to cope with infeasible initial queries.

\paratitle{Metrics.}
We evaluate methods on their efficiency and effectiveness.

\smallskip\noindent\emph{Response time:}
We measure response time as wall-clock time to generate an answer package. This includes the time taken to translate the \paql query into one or several \ilp problems, the time taken to load the problems into the solver, and the time taken by the solver to produce a solution. We exclude the time to materialize the package solution to the database and to compute its objective value.

\smallskip\noindent\emph{Approximation ratio:}
Recall that \bt is always guaranteed to return an approximate answer with respect to \opt (\Cref{sec:quality}).
In order to assess the quality of a package returned by \bt, we compare its objective value with the objective value of the package returned by \opt on the same query.
Using $Obj_S$ and $Obj_D$ to denote the objective values of \bt and
\opt, respectively, we compute the empirical \emph{approximation
ratio} $\frac{Obj_{D}}{Obj_{S}}$ for maximization queries, and
$\frac{Obj_{S}}{Obj_{D}}$ for minimization queries. An approximation
ratio of one indicates that \bt produces a solution with same
objective value as the solution produced by the solver on the entire
problem. Typically, the approximation ratio is greater than or equal
to one. However, since the solver employs several approximations and
heuristics, values lower than one, which means that \bt produces a
better package than \opt, are possible in practice.

\subsection{Results and Discussion}
We evaluate three fundamental aspects of our algorithms: 
(1)~their query response time and approximation ratio with increasing dataset sizes;
(2)~the impact of varying partitioning size thresholds ($\thres$) on \bt's performance;
(3)~the impact of the attributes used in offline partitioning on query runtime.

\subsubsection{Query performance as data set size increases}

\begin{figure*}
    \centering
    \begin{subfigure}[b]{1.0\textwidth}
        \centering
        \includegraphics[scale=0.25,trim=150 0 0 0,frame=.0mm]
        {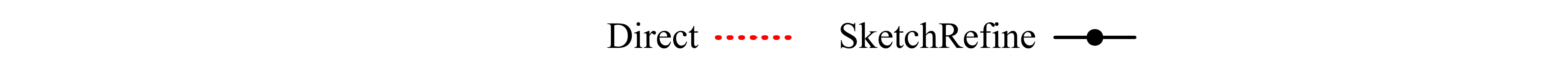}
    \end{subfigure}
    \begin{subfigure}[b]{0pt}
        \centering
        \includegraphics[scale=0.25,trim=50 0 0 0,left]
        {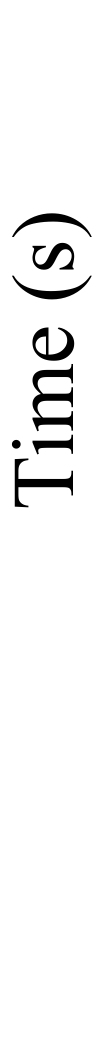}
    \end{subfigure}
    \begin{subfigure}[b]{.138\textwidth}
        \centering
        \caption[short for lof]{Q1}
        \includegraphics[scale=0.25,trim=30 0 0 20,center]
        {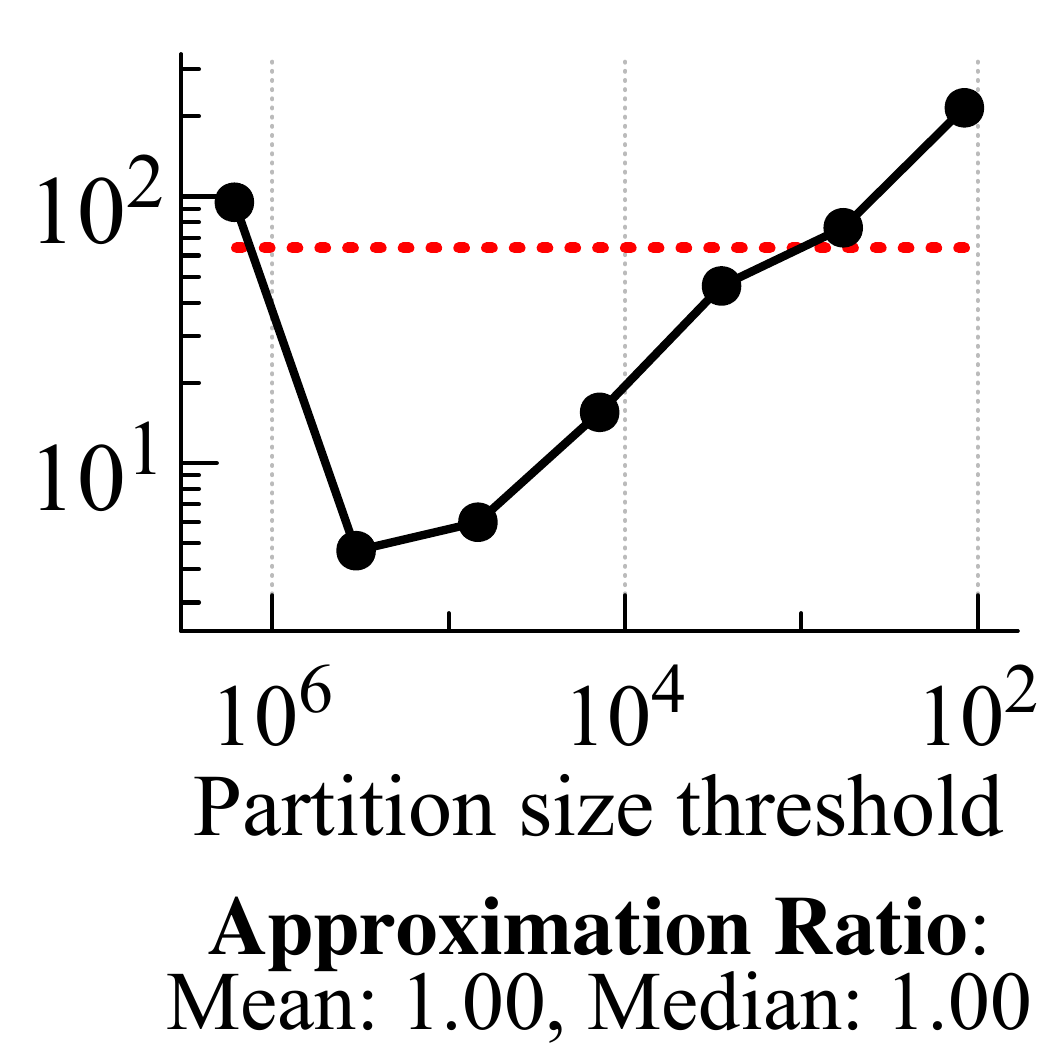}
    \end{subfigure}
    \begin{subfigure}[b]{.138\textwidth}
        \centering
        \caption[short for lof]{Q2}
        \includegraphics[scale=0.25,trim=30 0 0 20,center]
        {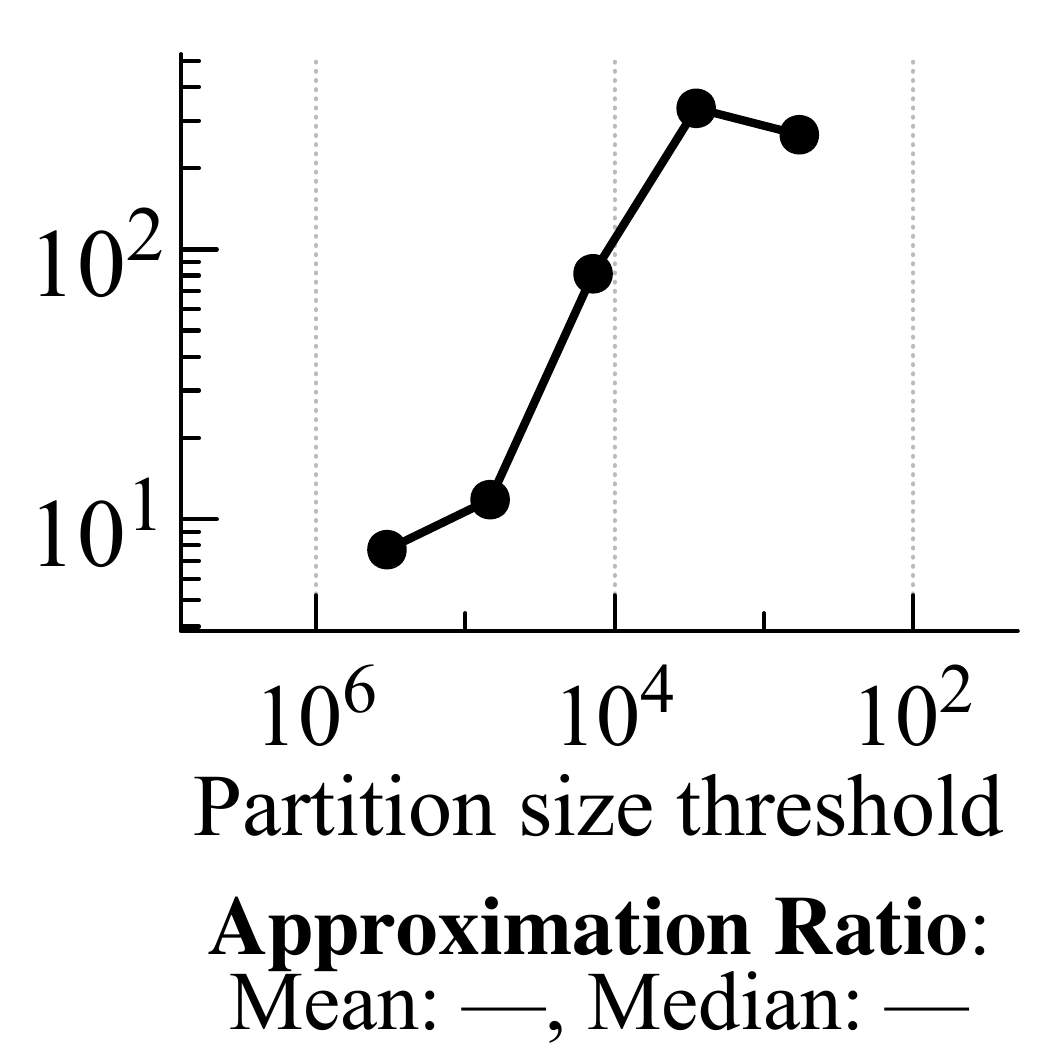}
    \end{subfigure}
    \begin{subfigure}[b]{.138\textwidth}
        \centering
        \caption[short for lof]{Q3}
        \includegraphics[scale=0.25,trim=30 0 0 20,center]
        {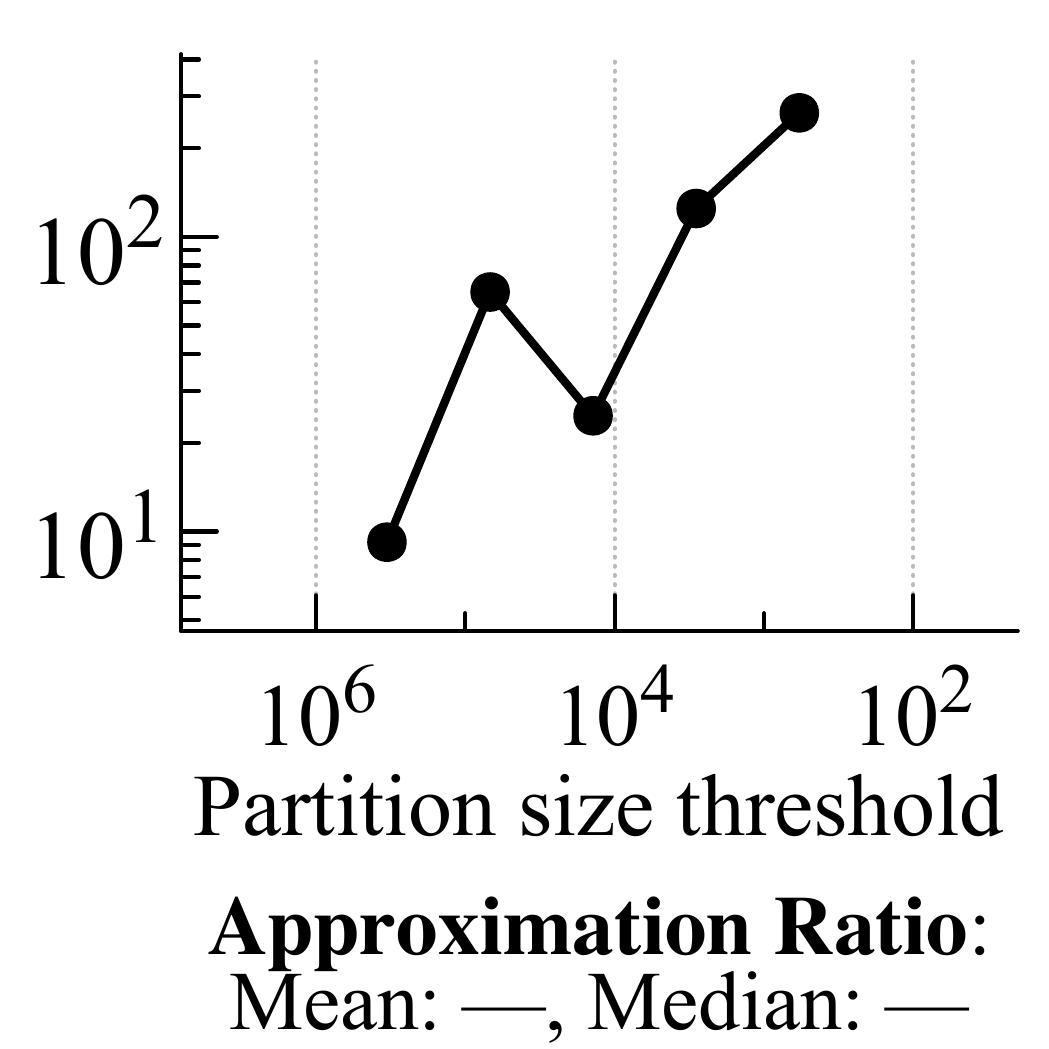}
    \end{subfigure}
    \begin{subfigure}[b]{.138\textwidth}
        \centering
        \caption[short for lof]{Q4}
        \includegraphics[scale=0.25,trim=30 0 0 20,center]
        {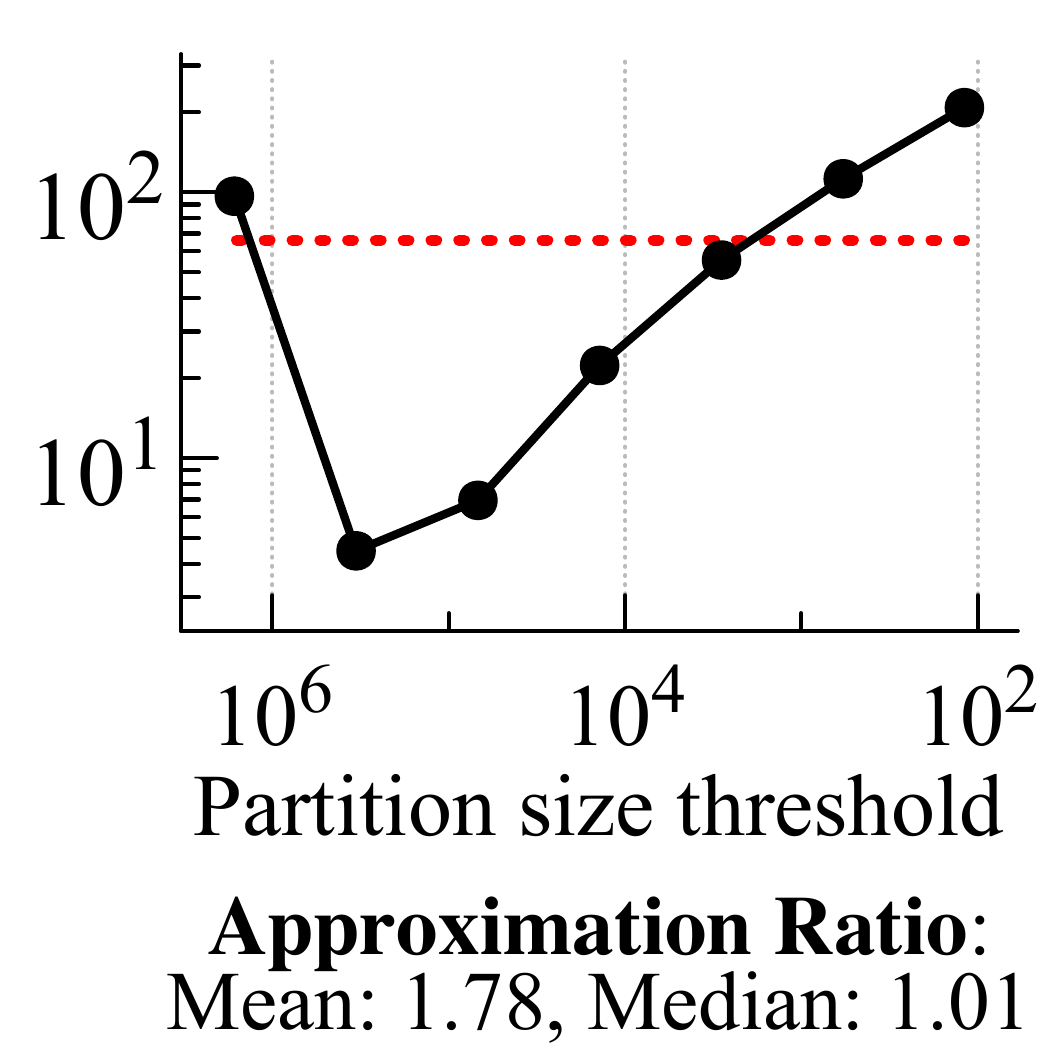}
    \end{subfigure}
    \begin{subfigure}[b]{.138\textwidth}
        \centering
        \caption[short for lof]{Q5}
        \includegraphics[scale=0.25,trim=30 0 0 20,center]
        {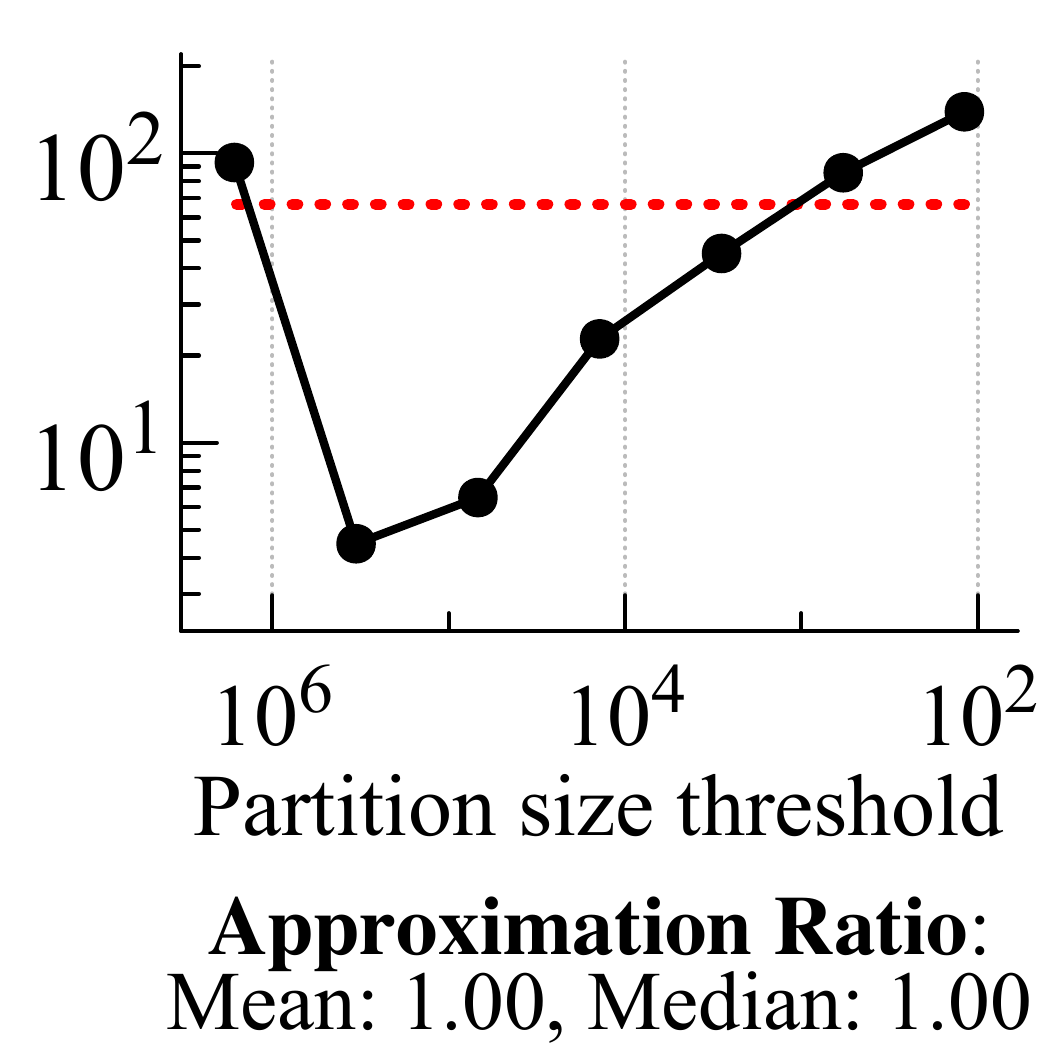}
    \end{subfigure}
    \begin{subfigure}[b]{.138\textwidth}
        \centering
        \caption[short for lof]{Q6}
        \includegraphics[scale=0.25,trim=30 0 0 20,center]
        {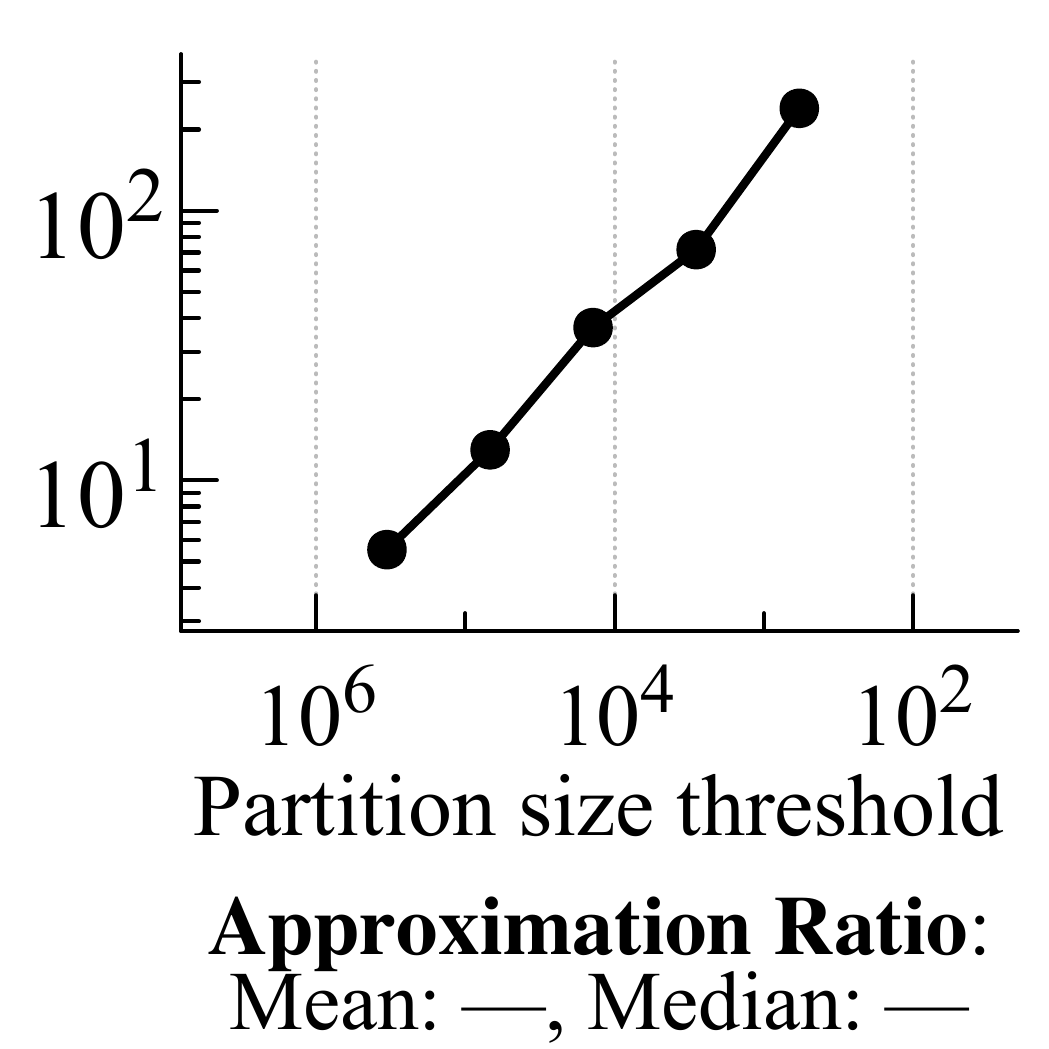}
    \end{subfigure}
    \begin{subfigure}[b]{.138\textwidth}
        \centering
        \caption[short for lof]{Q7}
        \includegraphics[scale=0.25,trim=30 0 0 20,center]
        {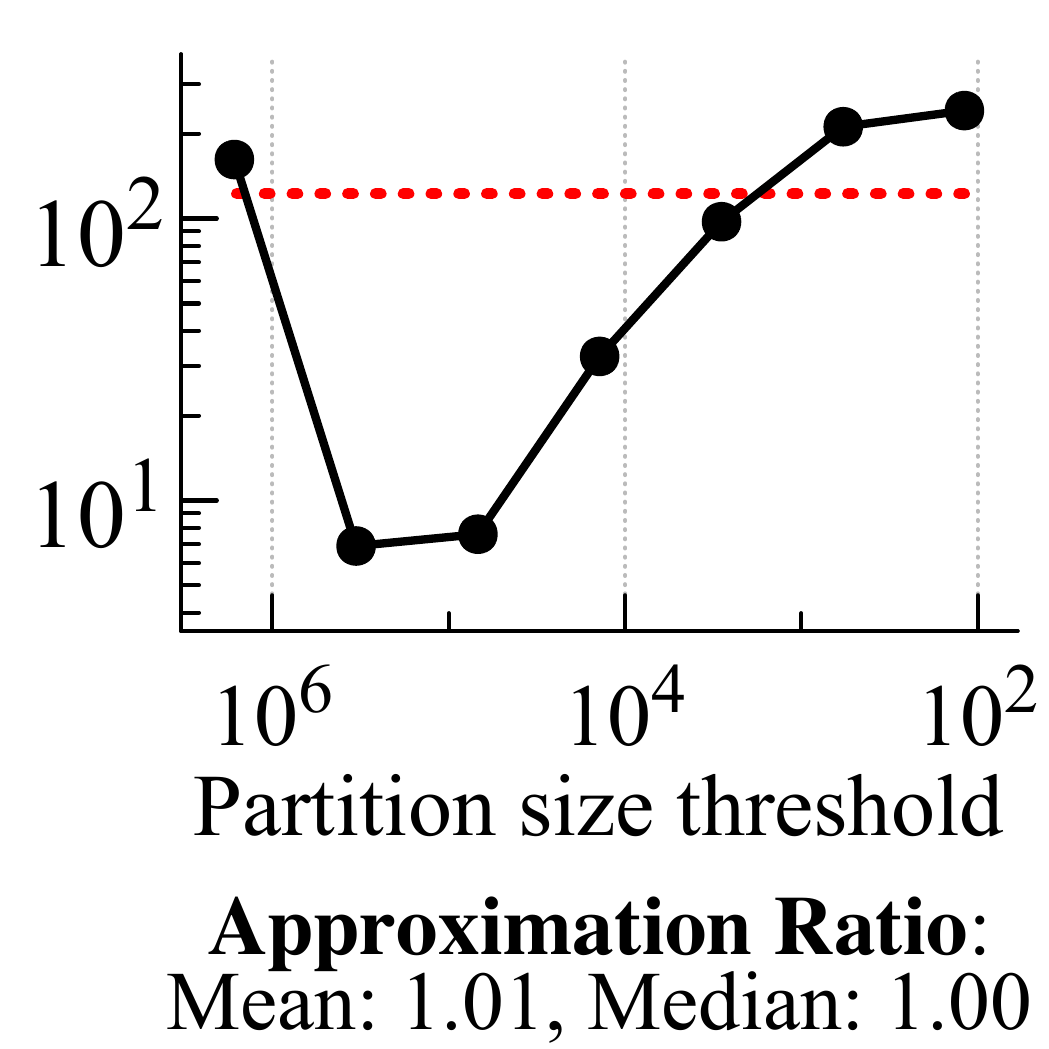}
    \end{subfigure}
    \caption{
    Impact of partition size threshold $\thres$ on the \galaxy benchmark, using 30\% of the original dataset.
    Partitioning is performed at each value of $\thres$ using all the workload attributes, and with no radius condition. 
    The baseline \opt and the approximation ratios are only shown when \opt is successful.
    The results show that $\thres$ has a major impact on the running time of \bt,
    but almost no impact on the approximation ratio.
    \opt can be an order of magnitude faster than \opt with proper tuning of $\thres$.
    }
    \label{fig:threshold_galaxy}
\end{figure*}

\begin{figure*}
    \centering
    \begin{subfigure}[b]{1.0\textwidth}
        \centering
        \includegraphics[scale=0.25,trim=150 0 0 0,frame=.0mm]
        {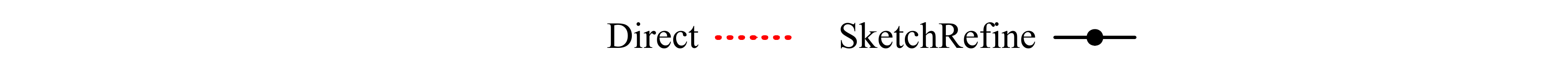}
    \end{subfigure}
        \begin{subfigure}[b]{0pt}
        \centering
        \includegraphics[scale=0.25,trim=50 0 0 0,left]
        {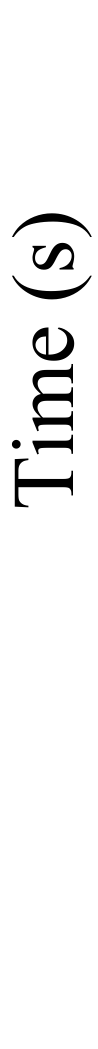}
    \end{subfigure}
    \begin{subfigure}[b]{.138\textwidth}
        \centering
        \caption[short for lof]{Q1}
        \includegraphics[scale=0.25,trim=30 0 0 20,center]
        {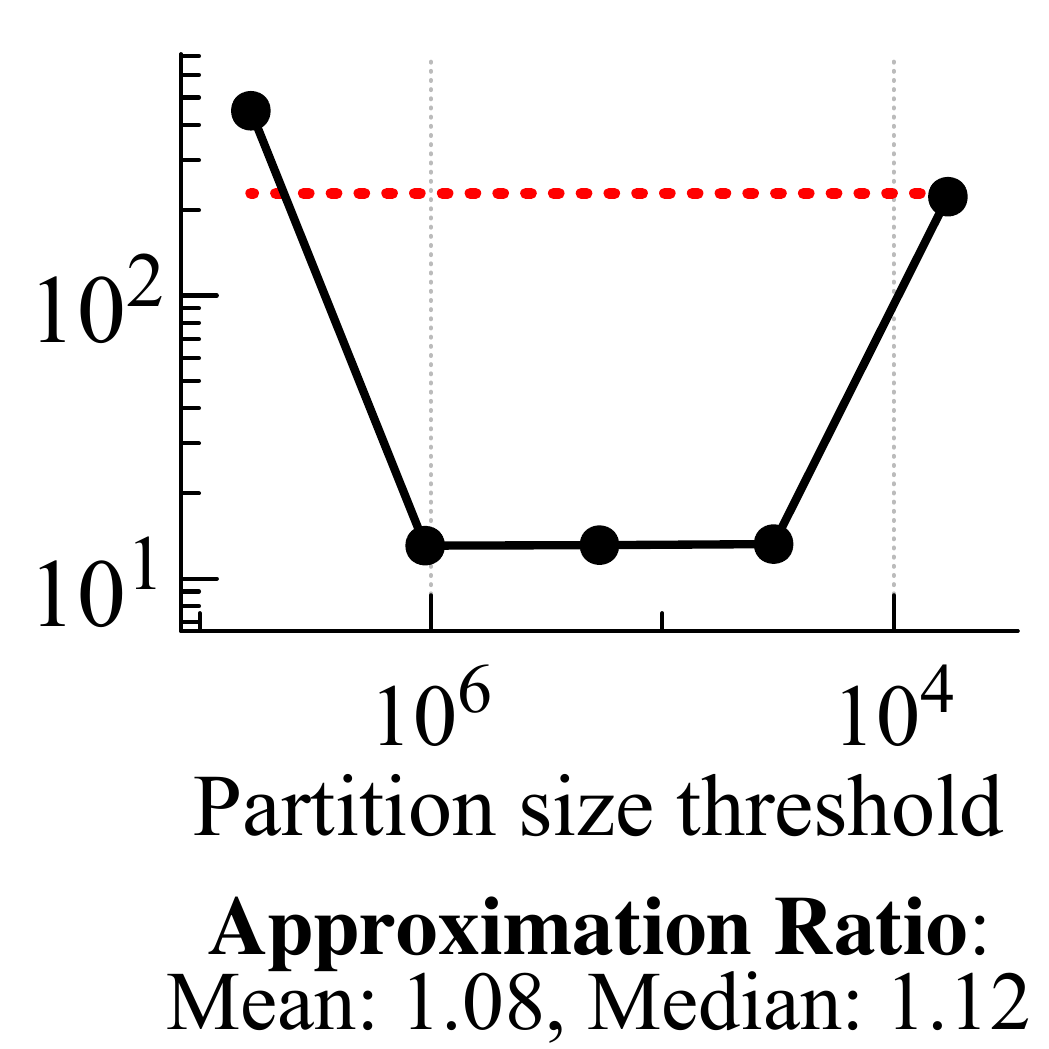}
    \end{subfigure}
    \begin{subfigure}[b]{.138\textwidth}
        \centering
        \caption[short for lof]{Q2}
        \includegraphics[scale=0.25,trim=30 0 0 20,center]
        {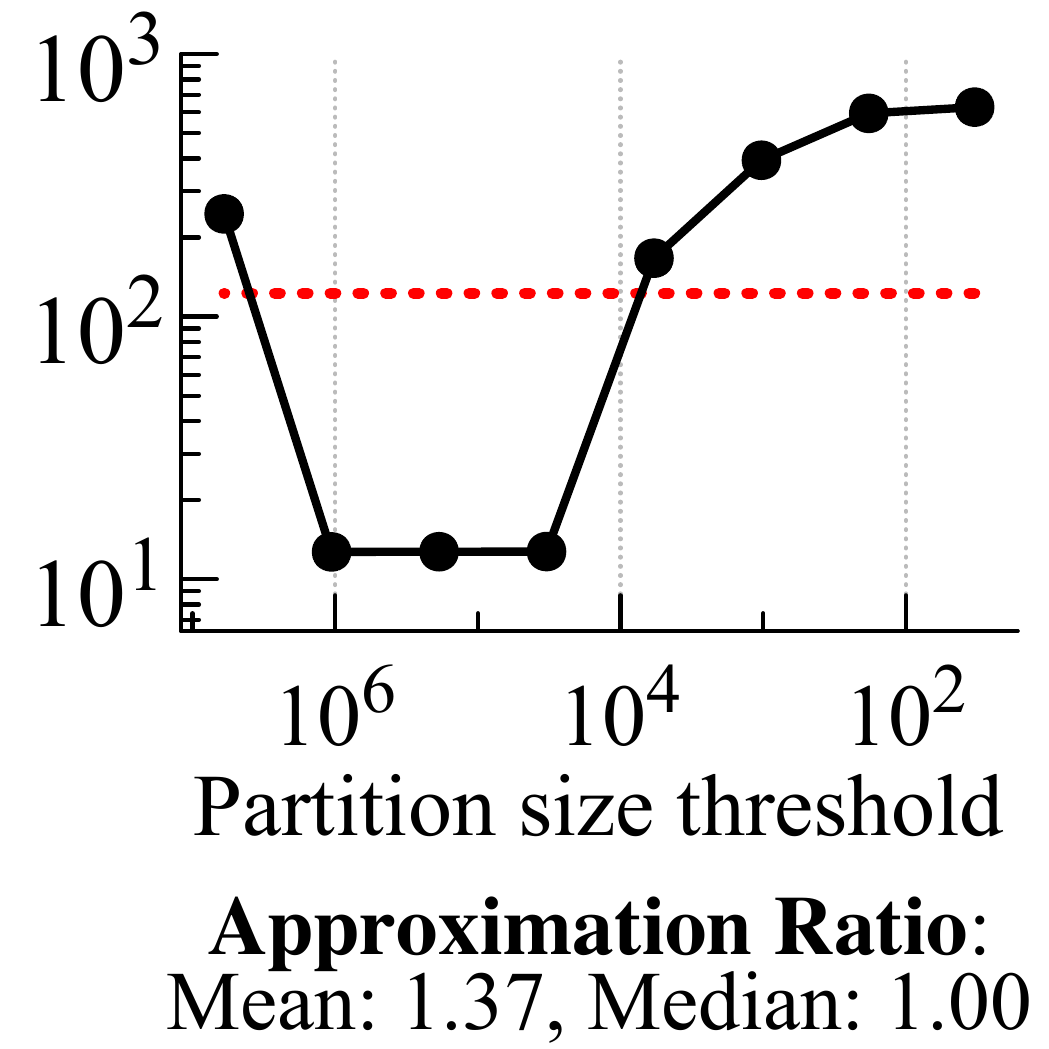}
    \end{subfigure}
    \begin{subfigure}[b]{.138\textwidth}
        \centering
        \caption[short for lof]{Q3}
        \includegraphics[scale=0.25,trim=30 0 0 20,center]
        {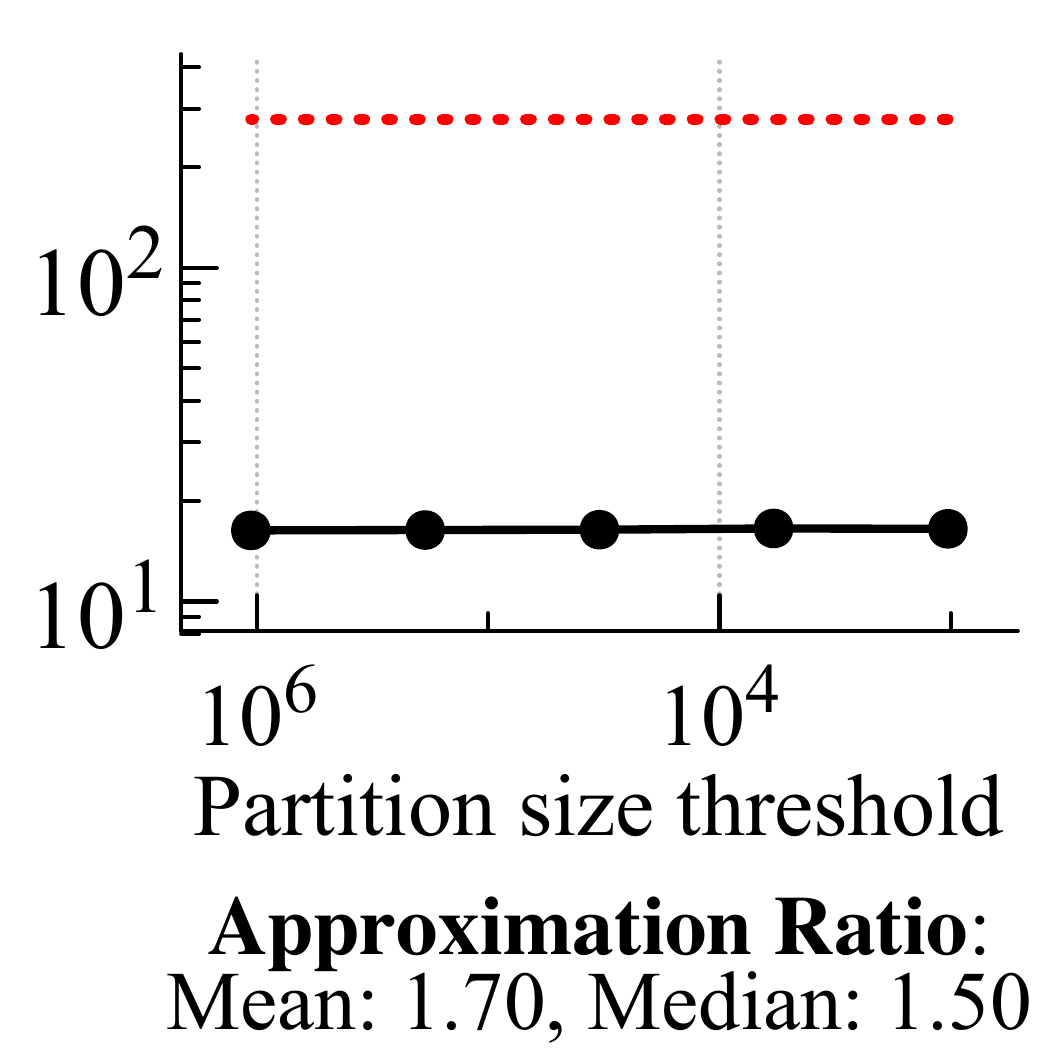}
    \end{subfigure}
    \begin{subfigure}[b]{.138\textwidth}
        \centering
        \caption[short for lof]{Q4}
        \includegraphics[scale=0.25,trim=30 0 0 20,center]
        {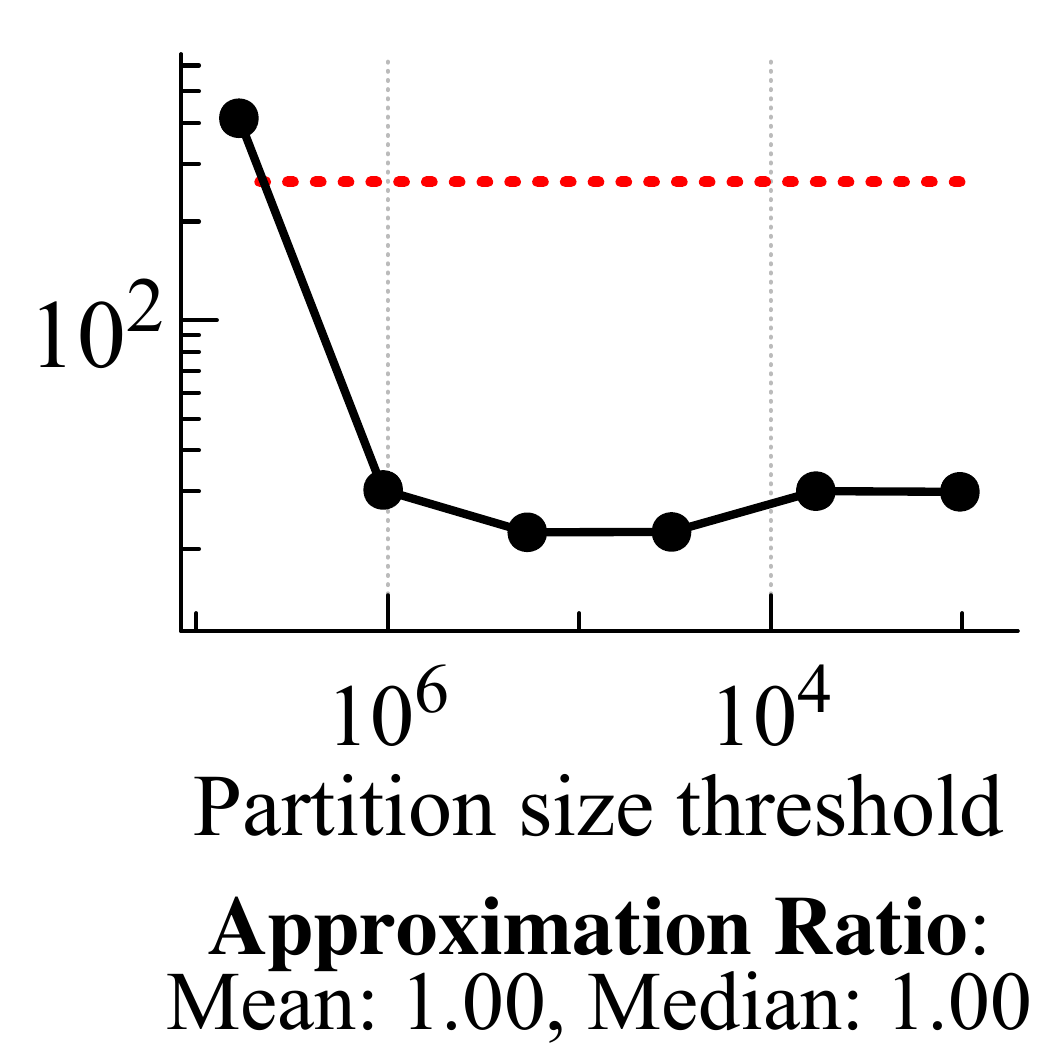}
    \end{subfigure}
    \begin{subfigure}[b]{.138\textwidth}
        \centering
        \caption[short for lof]{Q5}
        \includegraphics[scale=0.25,trim=30 0 0 20,center]
        {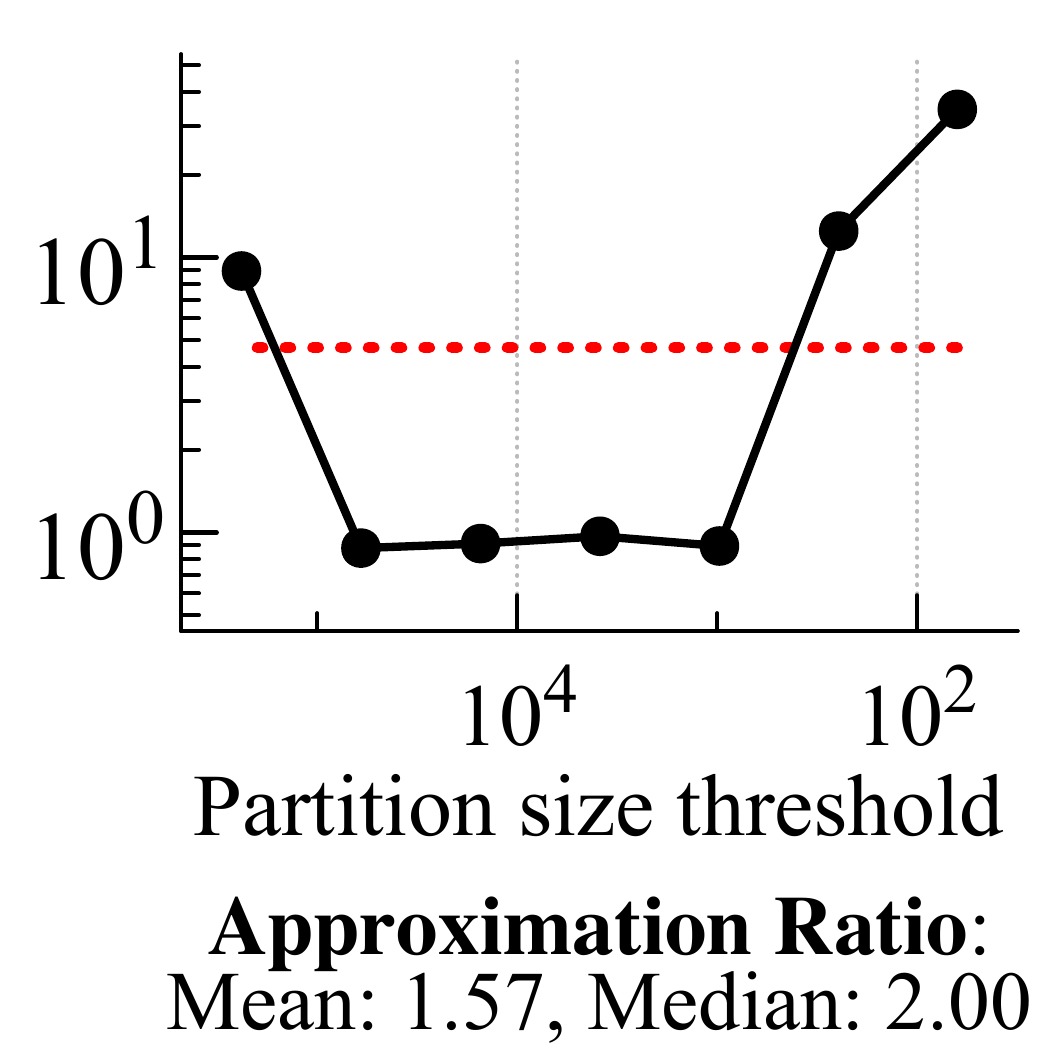}
    \end{subfigure}
    \begin{subfigure}[b]{.138\textwidth}
        \centering
        \caption[short for lof]{Q6}
        \includegraphics[scale=0.25,trim=30 0 0 20,center]
        {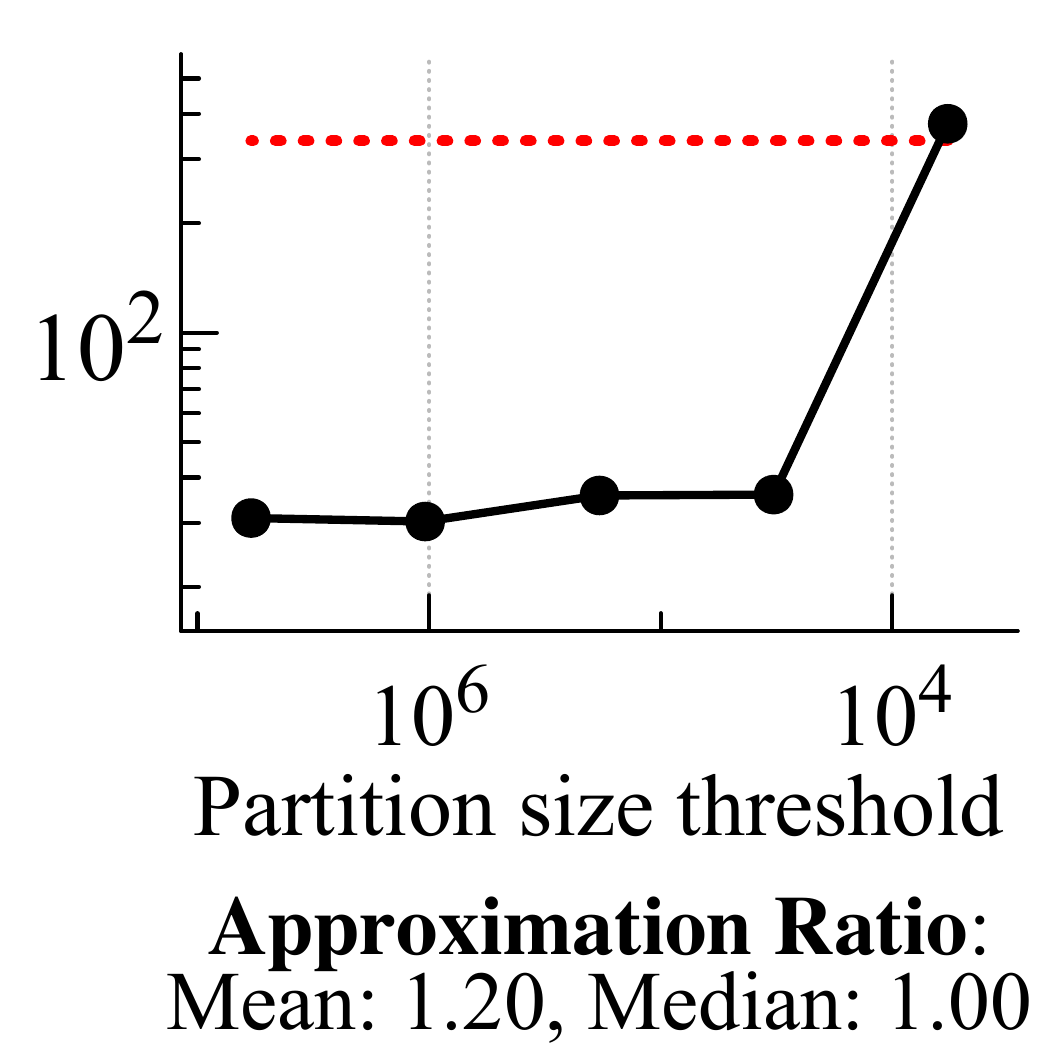}
    \end{subfigure}
    \begin{subfigure}[b]{.138\textwidth}
        \centering
        \caption[short for lof]{Q7}
        \includegraphics[scale=0.25,trim=30 0 0 20,center]
        {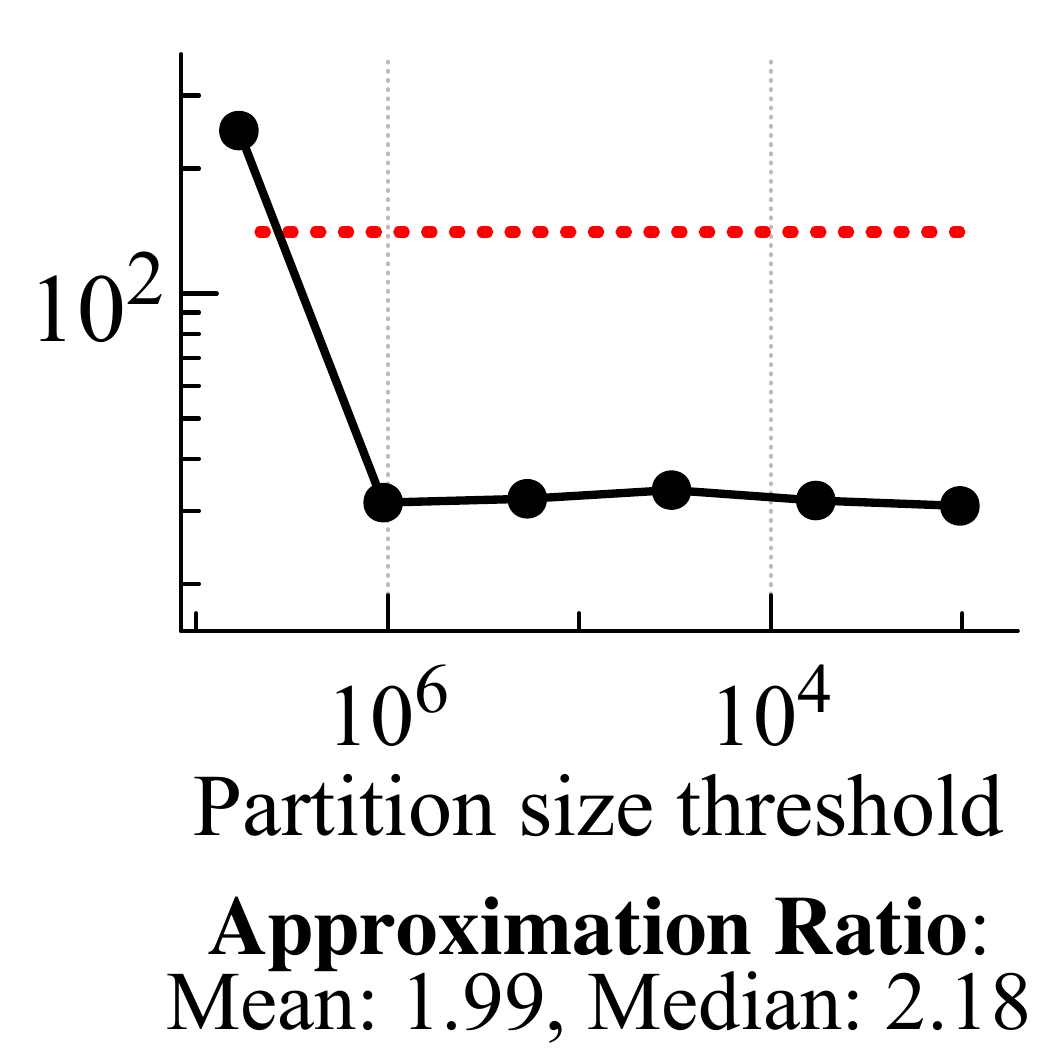}
    \end{subfigure}
    \caption{
    Impact of partition size threshold $\thres$ on the \tpch benchmark, using the full datasets.
    Partitioning is performed at each value of $\thres$ using all the workload attributes, and with no radius condition. 
    The baseline \opt and the approximation ratios are only shown when \opt is successful.
    The results show that $\thres$ has a major impact on the running time of \bt,
    but almost no impact on the approximation ratio.
    \opt can be an order of magnitude faster than \opt with proper tuning of $\thres$.
    }
    \label{fig:threshold_tpch}
\end{figure*}

In our first set of experiments, we evaluate the scalability of our
methods on input relations of increasing size.
First, we partitioned each dataset using the union of all package query attributes in the workload: we refer to these partitioning attributes as the \emph{workload attributes}. We did not enforce a radius condition ($\omega$) during partitioning for two reasons: (1)~to show that an offline partitioning can be used to answer efficiently and effectively both maximization and minimization queries, even though they would normally require different radii; (2)~to demonstrate the effectiveness of \bt in practice, even without having theoretical guarantees in place. 

We perform offline partitioning setting the partition size threshold
$\thres$ to 10\% of the dataset size. \Cref{tb:partitioning}
reports the partitioning times for the two datasets. We derive the
partitionings for the smaller data sizes (less than 100\% of the
dataset) in the experiments, by randomly removing tuples from the
original partitions. This operation is guaranteed to maintain the size
condition.

\looseness -1
\Cref{fig:scalability_galaxy,fig:scalability_tpch} report our scalability results on the \galaxy and \tpch benchmarks, respectively. The figures display the query runtimes in seconds on a logarithmic scale, averaged across 10 runs for each datapoint. 
At the bottom of each figure, we also report the mean and median
approximation ratios across all dataset sizes. The graph for Q2 on the
galaxy dataset does not report approximation ratios, because \opt
evaluation fails to produce a solution for this query across all data
sizes.
We observe that \opt can scale up to millions of tuples in three of the seven \galaxy queries, and in all of the \tpch queries. Its run-time performance degrades, as expected, when data size increases, but even for very large datasets \opt is usually able to answer the package queries in less than a few minutes. 
However, \opt has high failure rate for some of the \galaxy queries, indicated by the missing data points in some graphs (queries Q2, Q3, Q6 and Q7 in \Cref{fig:scalability_galaxy}).
This happens when \cplex uses the entire available main memory while solving the corresponding \ilp problems. For some queries, such as Q3 and Q7, this occurs with bigger dataset sizes. However, for queries Q2 and Q6, \opt even fails on small data. This is a clear demonstration of one of the major limitations of \ilp solvers: they can fail even when the dataset can fit in main memory, due to the complexity of the integer problem. In contrast, our scalable \bt algorithm is able to perform well on all dataset sizes and across all queries. 
\bt consistently performs about an order of magnitude faster than \opt across all queries, both on real-world data and benchmark data. Its running time is consistently below one or two minutes, even when constructing packages from millions of tuples.

Both the mean and median approximation ratios are very low, usually all close to one or two. This shows that the substantial gain in running time of \bt over \opt does not compromise the quality of the resulting packages. 
\added{Our results indicate that the overhead of partitioning with a radius condition is often unnecessary in practice.}
\added{Since the approximation ratio is not enforced, \bt can
potentially produce bad solutions, but this happens rarely. In our
experiments, this occurred with query Q2 from the \tpch benchmark,
which is a minimization query. Re-running the same experiment for Q2
with a partitioning that enforces a radius limit with $\epsilon =
1.0$, achieved perfect approximation ratios (Mean: 1.0, Median: 1.0).}
\removed{resulted in a substantial decrease in both mean and median approximation ratios (Mean: 5.90, Median: 1.99), with no change in running times.}

\subsubsection{Effect of varying partition size threshold} \label{sec:thres}

The size of each partition, controlled by the partition size threshold $\thres$, is an important factor that can impact the performance of \bt: Larger partitions imply fewer but larger subproblems, and smaller partitions imply more but smaller subproblems. 
Both cases can significantly impact the performance of \bt.
In our second set of experiments, we vary $\thres$, which is used during partitioning to enforce the  size condition (\Cref{sec:index}), to study its effects on the query response time and the approximation ratio of \bt. 
In all cases, along the lines of the previous experiments, we do not enforce a radius condition.
\Cref{fig:threshold_galaxy,fig:threshold_tpch} show the results obtained on the \galaxy and \tpch benchmarks, using 30\% and 100\% of the original data, respectively. 
We vary $\thres$ from higher values corresponding to fewer but larger partitions, on the left-hand size of the $x$-axis, to lower values, corresponding to more but smaller partitions. 
When \opt is able to produce a solution, we also report its running time (horizontal line) as a baseline for comparison.

Our results show that the partition size threshold has a major impact on the execution time of \bt, with extreme values of $\thres$ (either too low or too high) often resulting in slower running times than \opt.
With bigger partitions, on the left-hand side of the $x$-axis, \bt takes about the same time as \opt because both algorithms solve problems of comparable size. When the size of each partition starts to decrease, moving from left to right on the $x$-axis, the response time of \bt decreases rapidly, reaching about an order of magnitude improvement with respect to \opt.
Most of the queries show that there is a ``sweet spot'' at which the response time is the lowest: when all partitions are small, and there are not too many of them. 
\added{The point is consistent across different queries, showing that it only depends on the input data size (refer to \Cref{tb:tpch-queries} for the different \tpch data sizes).}
After that point, although the partitions become smaller, the number of partitions starts to increase significantly. This increase has two negative effects: it increases the number of representative tuples, and thus the size and complexity of the initial sketch query, and it increases the number of groups that \textsc{\btalgo} may need to refine to construct the final package.
This causes the running time of \bt, on the right-hand side of the $x$-axis, to increase again and reach or surpass the running time of \opt.
The mean and median approximation ratios are in all cases very close to one, indicating that \bt retains very good quality regardless of the partition size threshold.

\begin{figure}[t]
    \centering
    \begin{subfigure}[b]{0pt}
        \centering
        \includegraphics[scale=0.25,trim=50 0 0 0,left]
        {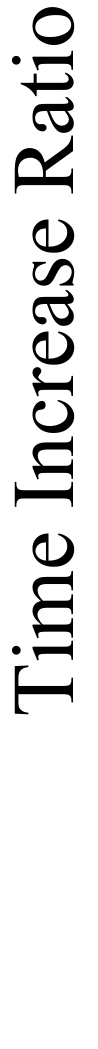}
    \end{subfigure}
    \begin{subfigure}[b]{.17\textwidth}
        \centering
        \caption[short for lof]{\galaxy}
        \includegraphics[scale=0.25,trim=30 0 0 20,center]
        {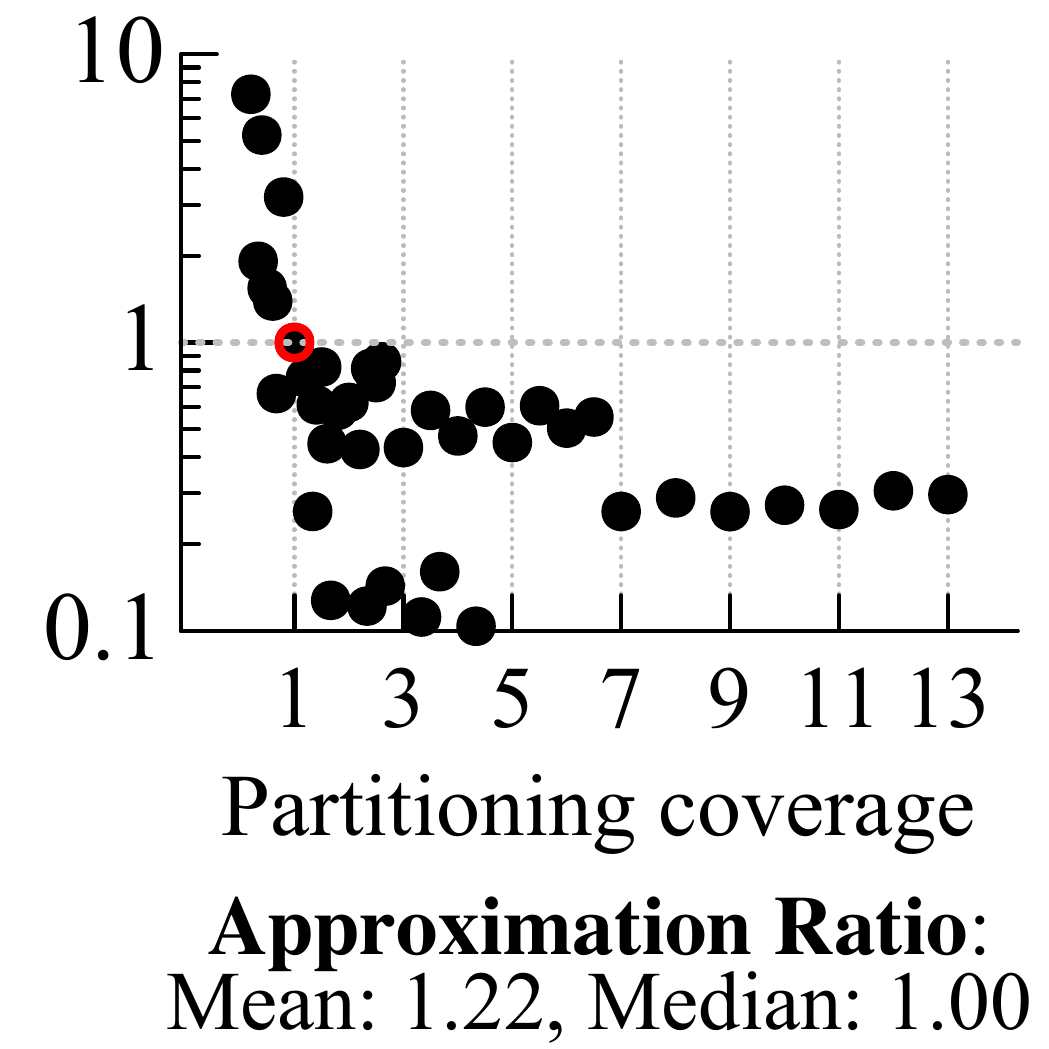}
    \end{subfigure}
    \begin{subfigure}[b]{.17\textwidth}
        \centering
        \caption[short for lof]{\tpch}
        \includegraphics[scale=0.25,trim=0 0 0 20,center]
        {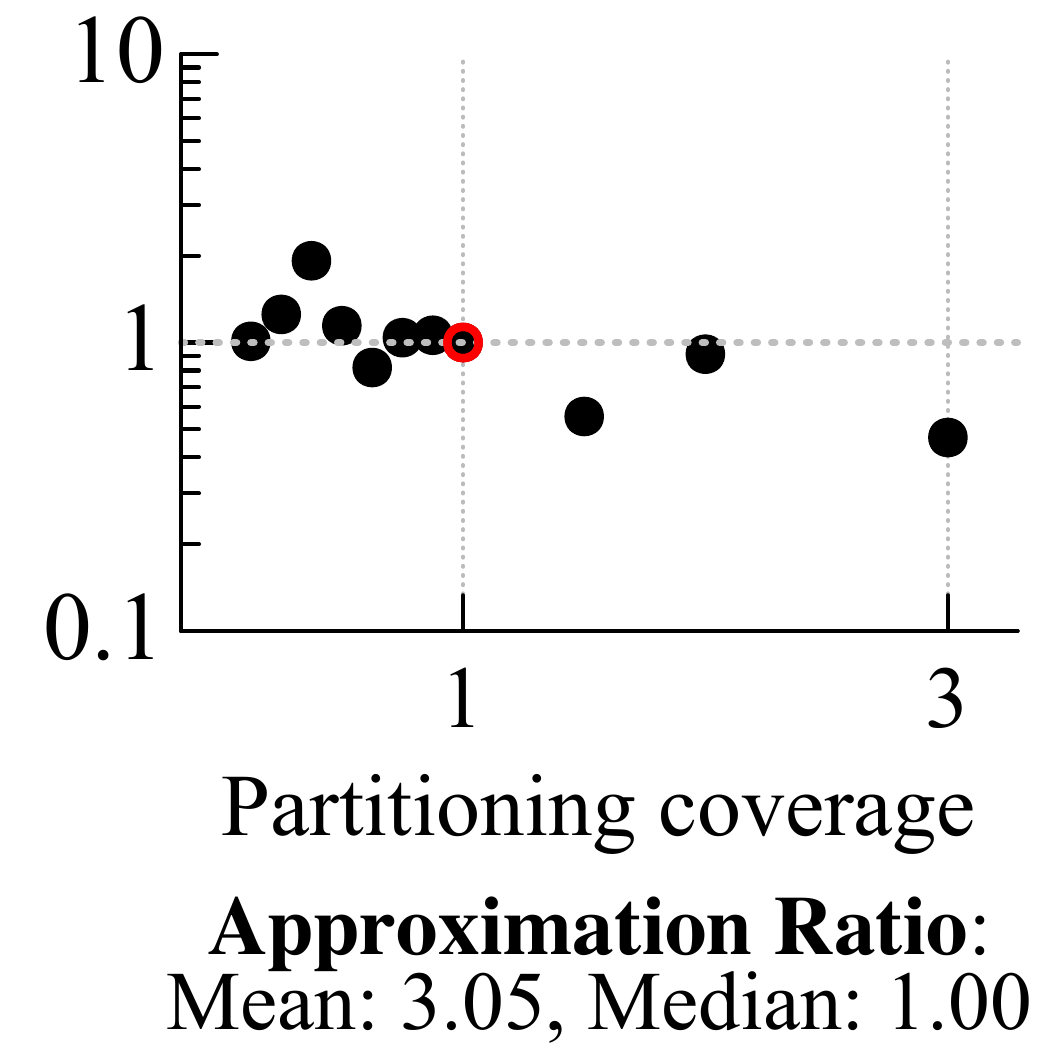}
    \end{subfigure}
    \caption{
    Increase or decrease ratio in running time of \bt with different partitioning coverages. Coverage one, shown by the red dot, is obtained by partitioning on the query attributes. 
    The results show an improvement in running time when partitioning is performed on supersets of the query attributes, with very good approximation ratios.
    }
    \vspace{-2mm}
    \label{fig:coverage}
\end{figure}

\subsubsection{Effect of varying partitioning coverage} \label{sec:exp:cover}

In our final set of experiments, we study the impact of offline partitioning on the query response time and the approximation ratio of \bt.
We define the \emph{partitioning coverage} as the ratio between the
number of partitioning attributes and the number of query attributes.
For each query, we test partitionings created using: (a) exactly the
query attributes (coverage = 1), (b) proper subsets of the query
attributes (coverage $< 1$), and (c) proper supersets of the query
attributes (coverage $> 1$).

For each query, we report the effect of the partitioning coverage on
query runtime as the ratio of a query response time over the same
query's response time when coverage is one: a higher ratio ($> 1$)
indicates slower response time and a lower ratio ($<1$) indicates a
faster response time. \Cref{fig:coverage} reports the results on
the \galaxy and the \tpch datasets. The \galaxy dataset has many more
numerical attributes than the \tpch dataset, allowing us to experiment
with higher values of coverage. The response time of \bt improves on
both datasets when the offline partitioning covers a superset of the
query attributes, whereas it tends to increase when it only considers
a subset of the query attributes. The mean and median approximation
ratios are consistently low, indicating that the quality of the
packages returned by \bt remains unaffected by the partitioning
coverage. These results demonstrate that \bt is robust to imperfect
partitionings, which do not cater precisely to the query attributes.
Moreover, using a partitioning over a superset of a query's attributes
typically leads to better performance. This means that partitioning
can be performed offline using the union of the attributes of an
anticipated workload, or even using all the attributes of a relation.

%% file: 10_appendices.tex
\begin{appendices}

\section{Appendix} \label{app:A}

\subsection{Reduction from \ilp to \paql} \label{app:expressiveness}

\begin{proof}[of Theorem~\ref{lemma:expressiveness}]
We prove this theorem by showing a reduction from an \ilp problem to a \paql query.
The reduction involves two mappings: (1) a mapping from a generic \ilp instance $\mathcal{I}$ into a specific database instance $\mathcal{D}$ and a \paql query $\mathcal{Q}$; (2) a mapping from a variable assignment $\vec{x}$ to the \ilp problem into a package $\pack$. The mappings are such that $\vec{x}$ is an optimal feasible solution to $\mathcal{I}$ \emph{iff} $\pack$ is an optimal feasible package for $\mathcal{Q}$, over the database instance $\mathcal{D}$.
Without loss of generality, let $\mathcal{I}$ be the following \ilp problem instance, involving $n$ integer variables, $k$ linear constraints, and real coefficients $a_{i}$, $b_{ij}$ and $ c_{j}$:
\begin{equation*}
\begin{array}{r l r}
\text{max}      & \sum_{i=1}^{n} a_i x_i \\
\text{s.t}      & \sum_{i=1}^{n} b_{ij} x_i \le c_j &  \forall j = 1, \dots, k \\
                & x_i \ge 0, x_i \in \mathbb{Z}  & \forall i = 1, \dots, n \\
\end{array}
\end{equation*}
The database instance $\mathcal{D}$ corresponding to this problem consists of a relation \ssf{R}$($\ssf{attr}$_{obj}$, \ssf{attr}$_1$, \ssf{attr}$_2$, $\dots$, \ssf{attr}$_k)$, which has $k+1$ attributes, one for each constraint and one extra for the objective. Relation \ssf{R} contains $n$ tuples $\tuple_1, \dots, \tuple_n$, where ${\tuple_i = (a_i, b_{i1}, b_{i2}, \dots, b_{ik})}$. That is, tuple $t_i$ contains all the coefficients $a_i$ and $b_{ij}$ that are multiplied by variable $x_i$ in $\mathcal{I}$ (using the \ilp jargon, tuple $\tuple_i$ is the $i$-th column of the \emph{constraint matrix} of $\mathcal{I}$).
The \paql query $\mathcal{Q}$ associated with $\mathcal{I}$ is:
\begin{tabbing}
\hspace*{3mm}\=\hspace*{2cm}\=\hspace*{3cm}\= \kill
\>\ssf{SELECT}    \>\ssf{PACKAGE}$($\ssf{R}$)$ \ssf{AS} \ssf{P} \\
\>\ssf{FROM}      \>\ssf{R} \\
\>\ssf{SUCH} \ssf{THAT}   \>\ssf{SUM}$($\ssf{P.attr}$_j)$ $\le$ $c_j ~~~~~~~~~ \forall j=1, \dots, k$ \\
\>\ssf{MAXIMIZE}  \>\ssf{SUM}$($\ssf{P.attr}$_{obj})$
\end{tabbing}

Let $\vec{x}$ be an \emph{assignment} to the variables in $\mathcal{I}$. Package $\pack$ is constructed from $\vec{x}$ by including tuple $t_i$ exactly $\vec{x}_i$ times. That is, $\pack$ is the multiset $\{ (t_1, \vec{x}_1), \dots, (t_n, \vec{x}_n) \}$\footnote{We denote multiset elements as pairs $(e,m)$, where $e$ is the element and $m$ its multiplicity.}. 

($\Rightarrow$) If $\vec{x}$ is an optimal feasible solution to $\mathcal{I}$, then $\pack$ is an optimal feasible package for query $\mathcal{Q}$ over the database instance $\mathcal{D}$. In fact, because $\vec{x}$ is a the optimal solution to $\mathcal{I}$, the following is true:
$$
\begin{array}{l l}
\sum_{i=1}^{n} a_i \vec{x}_i            & \text{is maximal} \\
\sum_{i=1}^{n} b_{ij} \vec{x}_i \le c_j & \forall j=1, \dots, k \\
\end{array}
$$
Thus, by construction of $\pack$, \ssf{SUM}$(\pack$\ssf{.attr}$_j) = \sum_{i=1}^{n} b_{ij} \vec{x}_i \le c_j$, and \ssf{SUM}$(\pack$\ssf{.attr}$_{obj}) = \sum_{i=1}^{n} a_i \vec{x}_i$ is maximal.

($\Leftarrow$) If $\pack$ is an optimal package for query $\mathcal{Q}$ over the database instance $\mathcal{D}$, then $\vec{x}$ is an optimal feasible solution to $\mathcal{I}$. The proof is analogous to the previous case.
\end{proof}

\subsection{Approximation Guarantees} \label{app:perf-guarantees}

In this section, we establish approximation bounds for the evaluation algorithm \bt presented in Section~\ref{sec:evaluation}. 
These bounds relate the quality of the solution produced by \bt with the quality of the solution produced by \opt. 
Both algorithms use a black box \ilp solver to evaluate the query. 
It is important to notice that the 
performance of \ilp solvers may not be consistent among different \ilp problems, especially of different size, because different problems can trigger different approximations and heuristics inside the solver. 
To simplify our discussion, we assume that the black box solver behave consistently among different queries. 
To abstract from all internal approximations and heuristics of solvers, we refer to the solutions produced by the solver as \emph{optimal}. 
This way, we can treat the solution produced by \bt as approximate with respect to \opt.

We present the complete analysis for a maximization query. We then discuss how to easily adapt the same analysis for a minimization query (Section~\ref{sec:min-case}).

\subsubsection{Linear maximization package query}
Without loss of generality, let us consider the following linear maximization package query, where $\mathcal{A} = \{ {\sf attr}_1, \dots, {\sf attr}_k \}$ is a set of $k$ partitioning attributes, and ${\sf attr}_{obj}$ is an attribute among those in $\mathcal{A}$ used to express the objective criterion:
\begin{tabbing}
\hspace*{3mm}\=\hspace*{2cm}\=\hspace*{3cm}\= \kill
\>\ssf{SELECT}    \>\ssf{PACKAGE}$($\ssf{R}$)$ \ssf{AS} \ssf{P}\\
\>\ssf{FROM}      \>\ssf{R} \\
\>\ssf{SUCH} \ssf{THAT} \\
\>~~~~~~~~~~~~~~\ssf{SUM}$($\ssf{P.attr}$_l)$ \ssf{BETWEEN} $L_l$ \ssf{AND} $U_l ~~~~~~~ 1 \le l \le k$ \\
\>\ssf{MAXIMIZE}  \>\ssf{SUM}$($\ssf{P}$.{\sf attr}_{obj})$
\end{tabbing}

\paratitle{Notation and definitions.}
We denote a generic package by $\pack$, 
the subset of tuples and representatives from $\pack$ that belong to group $G_j$ by $\pack_j$,
a generic tuple from group $G_j$ by $\tuple_j$, 
and its representative tuple by $\repr_j$.

\begin{definition}[Feasible Package]
For any query $\query$, $\pack$ is a \emph{feasible package} for $\query$, denoted by $\pack \sat \query$, if for every $1 \le l \le k$:
\begin{equation*}
L_l \le \SUM(\pack\attr_l) \le U_l
\end{equation*}
\end{definition}

\begin{definition}[$\bet$]
Package $\pack_1$ is \emph{at least as good as} package $\pack_2$, denoted by $\pack_1 \bet \pack_2$, if:
\begin{equation*}
\SUM(\pack_1\attr_{obj}) \ge \SUM(\pack_2\attr_{obj})
\end{equation*}

Furthermore, for any $\alpha \ge 1$, package $\pack_1$ is \emph{approximately at least as good as} package $\pack_2$, denoted by $\pack_1 \bet (1-\epsilon)^\alpha ~ \pack_2$, if:
\begin{equation*}
\SUM(\pack_1\attr_{obj}) \ge (1-\epsilon)^\alpha ~ \SUM(\pack_2\attr_{obj})
\end{equation*}
\end{definition}

\begin{definition}[Optimal Package]
Package $\pack$ is an \emph{optimal package} to query $\query$, denoted by $\pack \qopt \query$, if $\pack \sat \query$, and for all $\pack' \sat \query$, $\pack \bet \pack'$.
\end{definition}

\begin{definition}[$(1-\epsilon)^{\alpha}$-approximation]
For any query $\query$, $\alpha \ge 1$,
package $\pack$ is a \emph{${(1-\epsilon)^{\alpha}}$-approximation} to query $\query$, denoted by $\pack \qapprox{\alpha} \query$, 
if $\pack \sat \query$, and for all $\pack' \sat \query$, $\pack \bet (1-\epsilon)^\alpha ~ \pack'$.
\end{definition}

\begin{definition}[$(1-\epsilon)^{\beta}$-relaxation]
For any query $\query$, $\beta \ge 1$, 
package $\pack$ is a \emph{${(1-\epsilon)^{\beta}}$-relaxation} to $\query$, denoted by $\pack \qrelax{\beta} \query$, if for every $1 \le l \le k$:
\begin{equation*}
(1-\epsilon)^\beta ~ L_l \le \SUM(P\attr_l) \le (1-\epsilon)^{-\beta} ~ U_l
\end{equation*}

Also, package $\pack$ is an \emph{optimal ${(1-\epsilon)^{\beta}}$-relaxation} to $\query$, denoted by $\pack \qoptrelax{\beta} \query$, if $\pack \qrelax{\beta} \query$ and for all $\pack' \qrelax{\beta} \query$, $\pack \bet \pack'$.

Furthermore, package $\pack$ is a \emph{$(1-\epsilon)^\alpha$-approximation, ${(1-\epsilon)^{\beta}}$-relaxation} to $\query$, denoted by $\pack \qapproxrelax{\beta}{\alpha} \query$, if $\pack \qrelax{\beta} \query$ and for all $\pack' \qrelax{\beta} \query$, $\pack \bet (1-\epsilon)^\alpha~ \pack'$.
\end{definition}

\begin{definition}[Package Projection]\hfill
The \emph{package projection} of a package $\pack$, denoted by $\proj(\pack)$, is obtained from $\pack$ by replacing each original tuple with its corresponding representative, and each representative with a random original tuple from the same group.
\end{definition}

\paratitle{Analysis.}
The following lemma establishes an approximation bound for an approximate solution to a relaxed version of the query.

\begin{lemma}[Relaxed Approximation Bound] \label{lemma:relaxation}
For all packages $\pack_1$, $\pack_2$, for any query $\query$, and for any $\alpha \ge 1$, $\beta \ge 1$:
\begin{equation*}
\pack_1 \qapproxrelax{\beta}{\alpha} \query,~ \pack_2 \sat \query \implies \pack_1 \bet (1-\epsilon)^\alpha ~ \pack_2
\end{equation*}
\end{lemma}

\begin{proof}
Let $\pack^* \qopt \query$ be an optimal package for $\query$ (thus, $\pack^* \bet \pack_2$).
Furthermore, let $\pack' \qoptrelax{\beta} \query$ be an optimal $(1-\epsilon)^\beta$-relaxation to $\query$. Relaxing the query can only improve the objective value of an optimal solution, thus $\pack' \bet p^* \bet \pack_2$.
By hypothesis:
\begin{align*}
\pack_1 &\bet (1-\epsilon)^\alpha ~ \pack' \\
        &\bet (1-\epsilon)^\alpha ~ \pack_2
\end{align*}
\end{proof}

The following lemma establishes a (bi-directional) relation between the objective value of packages and their projections.

\begin{lemma}[Projection Bounds] \label{lemma:projection}
For any package $\pack$, for every $l$, $1 \le l \le k$:
\begin{align*}
\SUM(\pack\attr_l)        &\ge (1-\epsilon) ~ \SUM(\proj(\pack)\attr_l) \\
\SUM(\proj(\pack)\attr_l) &\ge (1-\epsilon) ~ \SUM(\pack\attr_l)
\end{align*}
\end{lemma}

\begin{proof}
The radius (defined in Section~\ref{sec:index}) satisfies the following condition for every group $G_j$, and for all $\tuple_j \in G_j$:
\begin{equation}\label{eq:radius0}
\repr_j\attr_l - r_j \le \tuple_j\attr_l \le \repr_j\attr_l + r_j
\end{equation}
For every group $G_j$, the following statements are true and easily verifiable from Equations~\ref{eq:radius} and~\ref{eq:radius0}:
\begin{equation} \label{eq:c}
\begin{array}{lc}
    \forall~ \tuple_j \in G_j, 1 \le l \le k:
& 
    \begin{array}{l}
        \tuple_j\attr_l \ge (1-\epsilon) ~ \repr_j\attr_l  \\
        \repr_j\attr_l  \ge (1-\epsilon) ~ \tuple_j\attr_l
    \end{array}
\end{array}
\end{equation}

This establishes that
tuples and representatives within a group do not differ from each other by more than a factor $(1 - \epsilon)$. With $\epsilon = 0$, tuples and their representatives are indistinguishable because the radius must be zero (by Definition~\ref{def:radius} and Equation~\ref{eq:radius}), which entails that the algorithm is guaranteed to return an optimal package.
Equation~\ref{eq:c} is sufficient to prove the lemma:
\begin{align*}
\SUM(\pack\attr_l) &=   \sum_{j=1}^{m} \sum_{t_j \in \pack_j} \tuple_j\attr_l \\
                   &\ge \sum_{j=1}^{m} \sum_{t_j \in \pack_j} (1-\epsilon)~ \repr_j\attr_l  \tag{By Equation~\ref{eq:c}} \\
                   &=   \sum_{j=1}^{m} |\pack_j| (1-\epsilon)~ \repr_j\attr_l \\
                   &=   (1-\epsilon)~ \SUM(\proj(\pack)\attr_l)
\end{align*}
\begin{align*}
\SUM(\proj(\pack)\attr_l) 
				   &=   \sum_{j=1}^{m} |\pack_j|~ \repr_j\attr_l \\
                   &\ge \sum_{j=1}^{m} |\pack_j|~ (1-\epsilon)~ \tuple_j\attr_l \tag{By Equation~\ref{eq:c}} \\
                   &=   \sum_{j=1}^{m} \sum_{t_j \in \pack_j} (1-\epsilon)~ \tuple_j\attr_l \\
                   &=   (1-\epsilon)~ \SUM(\pack\attr_l)
\end{align*}
\end{proof}

The following lemma states that the projection of a feasible package is a $(1-\epsilon)$-relaxation.

\begin{lemma}[Projection Relaxation] \label{lemma:proj-relax}
For any package $\pack$ and query $\query$:
\begin{equation*}
\pack \sat \query \implies \proj(\pack) \qrelax{} \query 
\end{equation*}
\end{lemma}
\begin{proof}
For every $l$, $1 \le l \le k$:
\begin{align*}
L_l &\le \SUM(\pack\attr_l) \tag{By hypothesis} \\
    &\le (1-\epsilon)^{-1} ~ \SUM(\proj(\pack)\attr_l) \tag{By Lemma~\ref{lemma:projection}} \\
U_l &\ge \SUM(\pack\attr_l) \tag{By hypothesis} \\
    &\ge (1-\epsilon) ~ \SUM(\proj(\pack)\attr_l) \tag{By Lemma~\ref{lemma:projection}}
\end{align*}

Therefore:
\begin{equation*}
(1-\epsilon) ~ L_l \le \SUM(\proj(\pack)\attr_l) \le (1-\epsilon)^{-1} ~ U_l
\end{equation*}
\end{proof}

The following lemma establishes that the projection of an optimal package is a $(1-\epsilon)^2$-approximate, $(1-\epsilon)$-relaxed package.

\begin{lemma}[Projection Approximation] \label{lemma:approx-projection}
For any package $\pack$ and any query $\query$:
\begin{equation*}
\pack \qopt \query \implies \proj(\pack) \qapproxrelax{}{2} \query
\end{equation*}
\end{lemma}
\begin{proof}
Let $\pack' \sat \query$ be a feasible solution for $\query$.
By Lemma~\ref{lemma:proj-relax}, we know that $\proj(\pack) \qrelax{} \query$ and $\proj(\pack') \qrelax{} \query$.
Because $\pack$ is optimal, $\pack \bet \pack'$.
Therefore:
\begin{align*}
\proj(\pack) &\bet (1-\epsilon) ~ \pack \tag{By Lemma~\ref{lemma:projection}}\\
             &\bet (1-\epsilon) ~ \pack' \\
             &\bet (1-\epsilon)^2 ~ \proj(\pack')  \tag{By Lemma~\ref{lemma:projection}}
\end{align*}
\end{proof}

We are now ready to prove the theorem.

\begin{proof}[of \Cref{th:approx}]
Let the initial sketch package be denoted by $\pack^{(0)}$.
Suppose, without loss of generality, that the algorithm refines the initial package in the order: $G_1, G_2, \dots, G_m$. Let $\pack^{(j)}$ denote the intermediate refined package produced at the $j$-th iteration of the algorithm. The finial complete package returned by the algorithm is thus $\pack^{(m)}$. Let $\pack^*$ be an optimal package. The following two statements are sufficient to prove the theorem:
\begin{align*}
\pack^{(0)} &\bet (1-\epsilon)^3 ~ \pack^* \tag{{\sc \initialalgo}} \label{lemma:B} \\
\pack^{(m)} &\bet (1-\epsilon)^3 ~ \pack^{(0)} \tag{{\sc \augmentalgo}} \label{lemma:A}
\end{align*}

(\ref{lemma:B})
First, notice that $p^{(0)} \qopt \query$ because $p^{(0)}$ optimizes the sketch query $\initquery$, which has identical constraints and maximization objective as $\query$. 
Thus, by \Cref{lemma:approx-projection}, we have that 
${\proj(p^{(0)}) \qapproxrelax{}{2} \query}$.
Since $\pack^* \sat \query$, by \Cref{lemma:relaxation} $\proj(\pack^{(0)}) \bet (1-\epsilon)^2 ~ \pack^*$, and therefore:
\begin{align*}
\pack^{(0)} &\bet (1-\epsilon) ~ \proj(\pack^{(0)}) \tag{By Lemma~\ref{lemma:projection}} \\
            &\bet (1-\epsilon)^3 ~ \pack^*
\end{align*}

(\ref{lemma:A})
Consider the $j$-th iteration of the algorithm. Let $\pack_j^{(j)}$ be the result of the $j$-th refine query $\augmquery$. 
First, notice that $\pack_j^{(j)} \qopt \augmquery$ because it is the optimal solution to the refine query. 
Thus, by \Cref{lemma:approx-projection}, 
$\proj(\pack_j^{(j)}) \qapproxrelax{}{2} \augmquery$.
Let $\pack_j^{(j-1)}$ be the solution for group $G_j$ obtained at iteration $(j-1)$, before $\augmquery$ is executed. 
It is easy to show that ${\pack_j^{(j-1)} \sat \augmquery}$, and thus, by \Cref{lemma:relaxation}, 
$\proj(\pack_j^{(j)}) \bet (1-\epsilon)^2 \pack_j^{(j-1)}$.
Therefore, we have that:
\begin{align*}
\pack_j^{(j)} &\bet (1-\epsilon) ~ \proj(\pack_j^{(j)}) \tag{By Lemma~\ref{lemma:projection}} \\
              &\bet (1-\epsilon)^3 ~ \pack_j^{(j-1)}
\end{align*}
Because the solution to group $G_j$ is only changed at the $j$-th iteration, $\pack_j^{(j-1)} = \pack_j^{(0)}$ and $\pack_j^{(j)} = \pack_j^{(m)}$. Thus:
\begin{equation*}
\pack_j^{(m)} \bet (1-\epsilon)^3 ~ \pack_j^{(0)}
\end{equation*}
To conclude the proof, it is sufficient to notice that $\pack^{(m)} = \SUM(\pack^{(m)}\attr_{obj}) = \sum_{j=1}^{m}\pack_j^{(m)}$, and similarly for $\pack^{(0)}$.
\end{proof}

\subsubsection{Linear minimization package queries} \label{sec:min-case}

To extend the analysis to minimization package queries, it is sufficient to notice that Equation~\ref{eq:c} becomes:
\begin{equation*}
\begin{array}{lc}
    \forall~ \tuple_j \in G_j, 1 \le l \le k:
& 
    \begin{array}{l}
        \tuple_j\attr_l \le (1+\epsilon) ~ \repr_j\attr_l  \\
        \repr_j\attr_l  \le (1+\epsilon) ~ \tuple_j\attr_l
    \end{array}
\end{array}
\end{equation*}

\balance

\subsection{False Infeasibility Bound} \label{sec:false-inf-bound}

\begin{definition}[Package Query Selectivity]\hfill
Given a query $\query$, with input relation \ssf{R},
let $\rpack$ be a \emph{random package}, generated by selecting tuples from \ssf{R} uniformly at random, and with replacement (i.e., the probability of every input tuple being included in $\rpack$ is equally $1/n$, where $n=|{\sf R}|$ is the number of tuples).
The \emph{selectivity} of a package query, denoted by $sel(\query)$, is the probability of $\rpack$ being \emph{infeasible}:
\begin{equation*}
sel(\query) := 1-Pr\left[\rpack \sat \query\right]
\end{equation*}
\end{definition}

We first prove the following lemma:

\begin{lemma} \label{lemma:feas-init}
For any query $\query$ and any random package $\rpack$,
if ${\rpack \sat \query}$, then with high probability: 
(1)~the sketch query $\initquery$ is feasible;
(2)~all refine queries $\augmquery$ are feasible.
\end{lemma}

\begin{proof}
Let \ssf{A} be an attribute from the schema of the input relation \ssf{R}, and let query $\query$ include a global constraint on the sum over attribute \ssf{A}, \ssf{SUM}$($\ssf{A}$)$.
Because $\rpack$ is random, its package projection $\proj(\rpack)$, constructed by including the representative tuples corresponding to each random tuple in $\rpack$, is also a random package.
Thus, ${\sf SUM}(\rpack.{\sf A})$ and ${\sf SUM}(\proj(\rpack).{\sf A})$, i.e., the global constraint values over packages $\rpack$ and $\proj(\rpack)$, are both random variables.
We show that, with high probability, ${\sf SUM}(\proj(\rpack).{\sf A})$ does not differ from the expected ${\sf SUM}(\rpack.{\sf A})$ by more than a small constant and, thus, if $\rpack$ is feasible so is $\proj(\rpack)$.

As a first step, we apply Hoeffding's inequality~\cite{hoeffding1963probability} on ${\sf SUM}(\proj(\rpack).{\sf A})$, and establish that with high probability its value does not deviate from its expectation by more than a small constant $c \ge 0$:
\begin{equation} \label{eq:hoeff}
Pr \left[ | {\sf SUM}(\proj(\rpack){\sf .A}) - E[{\sf SUM}(\proj(\rpack){\sf .A})] | \ge c \right] \le \gamma_{c}
\end{equation}
where $\gamma_{c} = 2 \exp \left(- \frac{2c^2}{|\rpack|({\sf MAX}({\sf A}) - {\sf MIN}({\sf A}))^2} \right)$.

Let $A$ be the random variable corresponding to the value of attribute \ssf{A} of a random tuple in $\rpack$, and let $E[A]$ be its expected value. 
Similarly, let $\tilde{A}$ be the random variable corresponding to a random representative in $\proj(\rpack)$, and $E[\tilde{A}]$ its expected value. 
Finally, let $G$ be the random variable corresponding to the group a random representative in $\proj(\rpack)$ belongs to. 
Because representative tuples are the centroids of all the tuples in their group, we have that:
\begin{equation*}
E[\tilde{A}] = E\left[ \frac{1}{|G|} \sum_{G} A \right]
             = \frac{1}{|E[G]|} \sum_{E[G]} E[A]
             = E[A]
\end{equation*}

The expected sum over package $\proj(\rpack)$ is therefore:
\begin{align*}
E[{\sf SUM}(\proj(\rpack){\sf .A})]
    &= E \left[ \sum_{k=1}^{|\rpack|} \tilde{A} \right] \\
    &= \sum_{k=1}^{|\rpack|} E[\tilde{A}] \\
    &= \sum_{k=1}^{|\rpack|} E[A] \\
    &= E[{\sf SUM}(\rpack{\sf .A})] 
\end{align*}

By substituting the expected sum in Equation~\ref{eq:hoeff} we obtain:
\begin{equation} \label{eq:hoeff2}
Pr \left[ | {\sf SUM}(\proj(\rpack){\sf .A}) - E[{\sf SUM}(\rpack{\sf .A})] | \ge c \right] \le \gamma_{c}
\end{equation}

Equation~\ref{eq:hoeff2} shows that the probability that the sum on $\proj(\rpack)$ differs from the expected sum of $\rpack$ by more than $c \ge 0$ is bounded by a small constant.
This means that if $\rpack$ is feasible, with high probability $\proj(\rpack)$ is also feasible. 
Because $\proj(\rpack)$ is a package in the representative space, then the sketch query $\initquery$ is also feasible.

\smallskip
Let $\rpack_j$ be the tuples in $\rpack$ that belong to group $G_j$. 
Then, $\proj(\rpack_j)$ is the set of representatives in $\proj(\rpack)$ that belong to group $G_j$. 
It is easy to verify that if $\proj(\rpack)$ is feasible, then $\proj(\rpack_j)$ is a feasible package for the $j$-th refine query $\augmquery$. 
With a reasoning analogous to the one used above, it is possible to show that the summation on $\proj(\rpack_j)$ must be close, with high probability, to the expected summation on $\rpack_j$, and thus the refine query $\augmquery$ is also feasible.
\end{proof}

\begin{proof}[of \Cref{th:false-inf}]
The lower the selectivity, the higher the probability $Pr\left[ \rpack \sat \query \right]$, and thus, by Lemma~\ref{lemma:feas-init}, the higher the probability that $\initquery$ and all $\augmquery$ are feasible, which implies that \bt will eventually find a feasible package with high probability as well.
\end{proof}

\subsection{\paql Syntax Specification} \label{sec:paql-syntax}

\begin{Verbatim}[commandchars=\\\{\},codes={\catcode`$=3\catcode`^=7}]
SELECT PACKAGE(\textbf{rel_alias} [, $\dots$]) [AS] \textbf{package_name}
FROM \textbf{rel_name} [AS] \textbf{rel_alias} [REPEAT \textbf{repeat}] [, $\dots$]
[ WHERE \textbf{w_condition} ]
[ SUCH THAT \textbf{st_condition} ]
[ (MINIMIZE|MAXIMIZE) \textbf{objective} ]
\end{Verbatim}
where \texttt{\textbf{rel\_name}} is a relation name, 
\texttt{\textbf{rel\_alias}} a substitute name for a relation name, 
\texttt{\textbf{package\_name}} a name for the package result,
\texttt{\textbf{repeat}} a non-negative integer,
\texttt{\textbf{w\_condition}} a Boolean expression over tuple values (similarly to standard \sql),
\texttt{\textbf{st\_condition}} a Boolean expression over aggregate functions or suq-queries with aggregate functions,
\texttt{\textbf{objective}} an expression over aggregate functions or suq-queries with aggregate functions.

\balancecolumns

\end{appendices}